\documentclass[10pt]{article}
\usepackage{enumerate}
\usepackage{amsmath,amssymb}
\usepackage{caption}
\usepackage{dsfont}
\usepackage[usenames]{color}
\usepackage{bm}
\usepackage{ying}
\usepackage{multirow}
\usepackage{rotating}
\usepackage{fullpage} 
\usepackage{authblk}
\usepackage{enumerate}
\usepackage{comment}
\usepackage{geometry}
\usepackage{float}

\usepackage{mathtools}
\mathtoolsset{showonlyrefs}
 
\newcolumntype{g}{>{\columncolor{red}}c}

 \usepackage{tikz}
\usepackage{pgf}
\usepackage{pgfplots} 
\usepackage{listings}
\usepackage[colorlinks,
            linkcolor=red,
            anchorcolor=blue,
            citecolor=blue
            ]{hyperref}

\usepgfplotslibrary{fillbetween}
\usetikzlibrary{patterns}

\pgfmathdeclarefunction{gauss}{3}{%
  \pgfmathparse{#2+1/(#3*sqrt(2*pi))*exp(-((x-#1)^2)/(2*#3^2))}%
}

\newenvironment{customlegend}[1][]{%
    \begingroup
    \csname pgfplots@init@cleared@structures\endcsname
    \pgfplotsset{#1}%
}{%
    \csname pgfplots@createlegend\endcsname
    \endgroup
}%

\def\addlegendimage{\csname pgfplots@addlegendimage\endcsname}

\usepackage{graphicx,amssymb}
\usepackage{xcolor}

\let\emptyset\varnothing
\let\hat\widehat
\let\tilde\widetilde

\def\asto{{\stackrel{\textrm{a.s.}}{\to}}}

\newcommand{\cond}{ {\textrm{cond}} }
\newcommand{\new}{ {\textrm{new}} }
\newcommand{\shift}{ {\textrm{shift}} }
\newcommand{\trans}{ {\textrm{trans}} }

\def\asto{{\stackrel{\textrm{a.s.}}{\to}}}

\def \iid {\stackrel{\text{i.i.d.}}{\sim}}
\def \iidtext {\textrm{i.i.d.}}

\def\given{{\,|\,}}
\def\biggiven{\,\big{|}\,}
\def\Biggiven{\,\Big{|}\,}
\def\bigggiven{\,\bigg{|}\,}

\usepackage{xcolor}

\def\##1\#{\begin{align}#1\end{align}}
\def\$#1\${\begin{align*}#1\end{align*}}

\begin{document} 
 
\title{Tailored inference for finite populations: conditional validity \\
and transfer across distributions
} 
\author{Ying Jin} 
\author{Dominik Rothenh\"ausler}
\affil{Department of Statistics, Stanford University}
\date{\today}

\maketitle

\begin{abstract}
 
Parameters of sub-populations  
can be 
more relevant than 
super-population ones. 
For example, a healthcare provider 
may be interested in 
the effect of 
a treatment plan for a specific subset of their patients;  
policymakers may be concerned with 
the impact of 
a policy in a particular state within a given population.  
In these cases, the focus is on a specific finite population, 
as opposed to an infinite super-population. 
Such a population can be 
characterized by fixing 
some attributes that are intrinsic to them, 
leaving unexplained variations like 
measurement error as random. 
Inference for  a population  
with fixed attributes 
can then be modeled
as inferring parameters 
of a conditional distribution. 
Accordingly, it is desirable 
that  confidence intervals 
are conditionally valid 
for the realized population, 
instead of marginalizing over many possible draws of populations.

We provide a statistical inference framework 
for parameters of finite populations 
with known attributes.   
Leveraging the attribute information, 
our estimators and confidence intervals 
closely target  
a specific finite population.
When the data is from the population of interest, 
our confidence intervals   
attain asymptotic conditional validity 
given the attributes, 
and are  shorter 
than those for super-population inference. 
In addition, 
we develop procedures to 
infer parameters of 
new populations with differing covariate distributions; the   
confidence intervals are also conditionally valid 
for the new populations
 under mild conditions. 
Our methods extend to situations where 
the fixed information has a weaker structure 
or is only partially observed.  
We demonstrate 
the validity and applicability of our methods 
using simulated and real-world data.

\end{abstract}
 

\section{Introduction}
  
Statistical inference targets populations of
various resolutions, 
from  
super-population to individuals.   
In causal inference for example,  
traditionally,  
average treatment effects 
describe properties of a 
hypothetical super-population. 
Driven by the need for individualization 
in domains like 
precision medicine~\citep{kosorok2019precision}, 
there is a surge of interest in 
heterogeneous treatment effects   
to provide unit-specific information that 
varies with individual characteristics.   

This paper studies  
a situation that lies 
between unit-specific and super-population inference.  
For instance, to decide  whether to deploy a 
novel treatment plan, 
it is sensible for a healthcare provider   
to focus on its own patients.   
The population of interest 
might be best described 
as a finite set of units, as opposed to  
one unit or a hypothetical super-population. 
Characterizing these patients as drawn from 
a super-population is reasonable  
if no other knowledge is available.  
However, if some information 
such as their demographics is given,  
averaging over
many draws of such information -- a super-population 
perspective -- becomes
inappropriate. 
Instead, to describe these patients,  
their demographics should be viewed 
as fixed or conditioned on the realized values.  

To model this scenario,  
we will allow practitioners to choose certain 
information  
about these units as fixed; 
{depending on the application, 
it may or may not 
be fully observed}.  
Our estimand is a parameter of 
the conditional distribution of data given such information, 
termed \emph{conditional parameter}.   
It is argued in~\cite{abadie2014inference} that 
sometimes
conditional parameters can 
be more relevant 
than super-population parameters.  
Let us continue with the healthcare example 
and consider 
two settings where
conditional parameters are of interest.  
 
The first setting is 
described in~\cite{abadie2014inference} 
as ``the sample 
is from the population of interest''.  
Suppose 
the healthcare provider is interested 
in its patients' health conditions 
after deploying the treatment plan 
and collecting the health data from them.  
The left panel of
Figure~\ref{fig:intro_visual} visualizes 
the super-population (blue) 
versus a realized sub-population (red), 
also showing other potential sub-populations (grey).  
Once certain attributes are fixed, 
the current population is  
represented by the red-shaded
conditional distribution from which the collected data are drawn. In contrast, 
the super-population  marginalizes over all potential sub-populations, 
including those less relevant for the current patients. 
The task is to infer 
parameters of the realized conditional distribution 
using data from it.  
We call this setting \emph{inference 
for the population at hand}.
In this case, 
quantifying the uncertainty for 
the (red) sub-population leads to distinct (shorter) 
confidence intervals.

\begin{figure}[ht]
\centering
\scalebox{0.8}{%
\begin{tikzpicture}
  \begin{axis}[
  no markers, domain=0:15, samples=100,
 axis x line=middle,
      axis y line=none,
  every axis y label/.style={at=(current axis.above origin),anchor=south},
  every axis x label/.style={at=(current axis.right of origin),anchor=west},
  height=4cm, width=12cm,
  xtick=\empty,
  ytick=\empty,
  enlargelimits=false, clip=false, axis on top,
  grid = major
  ]   
     \addplot [very thick,cyan!60!black, domain=5:13] {gauss(8.2,0.2,0.9)};   
\begin{scope}[scale=0.8]
   \addplot [very thick,red!80!black, domain=5.8:11.8, name path=C] {gauss(8.8,0,0.7)}; 
    \addplot [thick, white!60!black, dashed, domain=6.5:14.5] {gauss(10.5,0,0.7)};
    \addplot [thick,  white!60!black, dashed, domain=5.4:13.4] {gauss(9.4,0,0.7)};
    \addplot [thick,  white!60!black, dashed, domain=4.4:12.4] {gauss(8.4,0,0.7)};
    \path[name path=xaxis]
      (0,0) -- (\pgfkeysvalueof{/pgfplots/xmax},0);
    \addplot[white!80!red] fill between[of=xaxis and C, opacity=0, soft clip={domain=-1:15}];  
\end{scope} 
\end{axis} 
\draw [black] (0,0) -- (2, 1.2); 
\draw [black] (1.25, 0.75) -- (10, 0.75); 
\draw [dashed, thick, red!80!black] (3.65,0) -- (3.65,1.65);
\node[red!80!black] at (3.65,-0.3) {\footnotesize $ \theta_n^\cond$} ;

\begin{customlegend}[
legend cell align=left,
legend entries={ 
Super-pop.,
Realized sub-pop., 
Other sub-populations
},
legend style={at={(11,2.5)},font=\footnotesize}] 
    \addlegendimage{no markers,cyan!60!black}
    \addlegendimage{no markers,red }
    \addlegendimage{no markers,white!40!black,dashed}
\end{customlegend}

\end{tikzpicture}
}
\scalebox{0.8}{ 
\begin{tikzpicture}
  \begin{axis}[
  no markers, domain=0:15, samples=100,
 axis x line=middle,
      axis y line=none,
  every axis y label/.style={at=(current axis.above origin),anchor=south},
  every axis x label/.style={at=(current axis.right of origin),anchor=west},
  height=4cm, width=12cm,
  xtick=\empty,
  ytick=\empty,
  enlargelimits=false, clip=false, axis on top,
  grid = major
  ]  
    
     \addplot [very thick,green!50!black, domain=10:19] {gauss(14,0.16,1.2)};   

\begin{scope}[scale=0.8] 
    
    \addplot [very thick, magenta!70!black, domain=10:17, name path=T] {gauss(13.85,0,1.1)};
    \addplot [thick,  white!60!black, dashed, domain=9.4:16.4] {gauss(12.9,0,1.1)};
    \addplot [thick,  white!60!black, dashed, domain=11:19] {gauss(15,0,1.1)};
    \addplot [thick,  white!60!black, dashed, domain=12:20] {gauss(16,0,1.1)};
    \path[name path=xaxis2]
      (0,0) -- (15,0);
    \addplot[white!80!magenta] fill between[of=xaxis and T, opacity=0, soft clip={domain=-1:25}]; 
\end{scope} 
\end{axis}

\draw [black] (0,0) -- (2, 1.2); 
\draw [black] (1.25, 0.75) -- (10, 0.75);   
\draw [dashed, thick, magenta!70!black] (3.5,0) -- (3.5,1.45); 
\node[magenta!70!black] at (3.5,-0.3) {\footnotesize $ \theta_m^{\cond,\new}$} ;

\begin{customlegend}[
legend cell align=left,
legend entries={ 
New super-pop.,
New realized sub-pop., 
Other sub-populations
},
legend style={at={(11,2.5)},font=\footnotesize}] 
    \addlegendimage{no markers,thick,green!50!black}
    \addlegendimage{no markers,magenta }
    \addlegendimage{no markers,white!40!black,dashed}
\end{customlegend}
 
\end{tikzpicture}
}
\caption{Visualization of super-population inference (blue) 
versus conditional (red) 
and transductive inference (purple). 
Dashed lines
are other potential sub-populations.} 
\label{fig:intro_visual}
\end{figure}
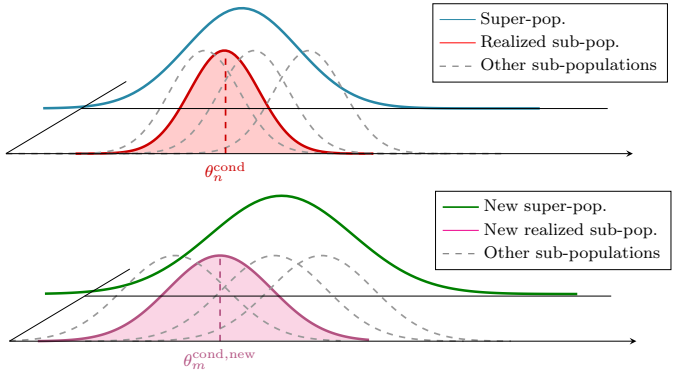

Another setting that is 
newer to the literature is  
where the sub-population of interest  
differs from the sample.  
After collecting data from a batch of patients, 
the healthcare provider may also want  
to predict the effect on 
a set of new patients 
\emph{before} deploying it 
on them, i.e.,   
without observing their responses, 
and the new patients may come from 
another super-population.  
{If the healthcare provider has observed
a few attributes of the new patients,  
then the conditional parameter for the new patients
informs their 
expected behavior  based on the available information.}
In the lower left of Fig.~\ref{fig:intro_visual}, the green curve represents the (shifted) super-population 
the new units are from. 
Newly observed attributes  give rise to
 the purple-shaded sub-population.
The goal is thus to 
transfer the knowledge to the new population, 
i.e., to infer parameters of the new
conditional distribution that potentially shifts. 
We call 
this setting \emph{transductive inference}.

Targeting specific sub-populations leads to 
estimators that differ from 
super-population ones. In Example~\ref{ex:gotv}, we give a sneak peek at a procedure that will be formally developed in Section~\ref{sec:trans_iid}.

\begin{example}[Super- versus targeted sub-population~estimation]
\label{ex:gotv}
Motivated by \cite{arceneaux2006comparing} who studied the effect of get-out-the-vote mails on voter turnout, 
we consider a scenario where a local politician is interested in using such mails in a particular region. 
In this setting, the estimand is the average treatment effect in the sub-population, conditionally on observed attributes. 
We repeatedly generate 
i.i.d.~training data 
and covariates of a small disjoint target population $(n=100)$, 
and 
compute covariate-adjusted estimators for
the average treatment effect of 
the super-population and 
the sub-population.  
We depict the results over $1000$ runs in Figure~\ref{fig:GOTV};
each point stands for 
one of the potential sub-populations in 
Fig.~\ref{fig:intro_visual}.   
\begin{figure}[ht]
\centering
\includegraphics[scale=.45]{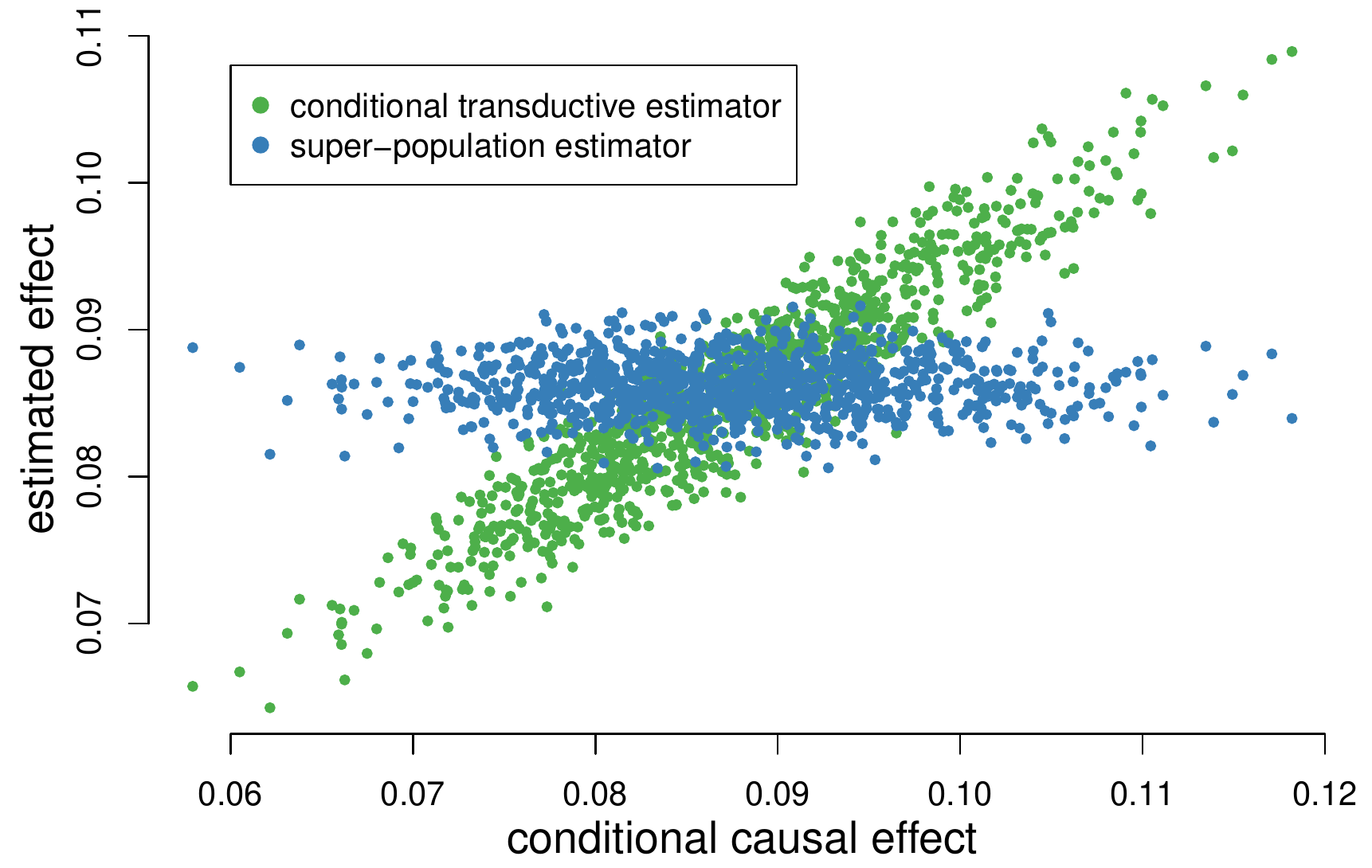}
\caption{Conditional versus super-population estimation in Example~\ref{ex:gotv}.  }\label{fig:GOTV} 
\end{figure}
Our transductive estimator uses covariates to target the sub-population 
and thus achieves much higher accuracy. 
Such targeted information, equipped with 
its reliable coverage guarantee we are to introduce, 
supports decision-making for specific sub-populations.  
\end{example} 
 
Fixing certain attributes  
motivates \emph{conditional} 
inference guarantees. 
A marginally valid confidence interval 
covers the target with a prescribed probability, 
averaged across many draws of the attributes. 
However, 
ideally, confidence intervals should be valid for the 
specific population we are interested in, 
that means, conditional on the attributes. 
In an illustrative example in Section~\ref{sec:super-vs-condit} 
of the supplementary material, 
we find that super-population inference 
lacks conditional validity 
even without transfer. 
In contrast, our inference 
is conditionally valid 
for the specific population of interest. 
For example, 
our method builds a confidence interval around 
\emph{each} conditional transductive estimator 
in Fig.~\ref{fig:GOTV} 
that is valid for the sub-population.

In this paper, 
we provide a framework for statistical 
inference tailored to finite populations 
with fixed attributes.  
%
%
In both {settings introduced above}, compared with 
super-population inference, 
our framework leads to shorter confidence intervals 
and more reliable inference with conditional coverage. 
To our knowledge, this is the first work to deal with both conditional inference and distribution shift in general estimation tasks.
To address all these problems in a single framework, 
throughout the main text, 
we use a sampling-based justification 
by assuming the attributes are i.i.d.~drawn and then 
conditioned on. 

{We provide R-packages for the first setup in 
\texttt{https://github.com/ying531/condinf}  
and for transductive inference in \texttt{https://github.com/ying531/transinf}. The packages
are easy to use, and 
allow transporting generalized linear models to a new (conditional) distribution with (conditionally) valid confidence intervals.} 



 
\section{Inferential targets}\label{sec:cond_para}

\subsection{Conditional parameters}
\label{subsec:cond_para}

We now formally 
define the \emph{conditional parameter} as 
our estimand,   
which characterizes a 
conditional distribution.  
Conditional parameters have a long history in statistics and econometrics; we give an overview of the literature in Section~\ref{sec:related-work}.

Let us begin with a recap on 
classical settings~\citep{Vaart1998,tsiatis2007semiparametric}.  
For a super-population $\mathbb{P}$ 
from which a random variable $D\in \cD$ is drawn, 
an unknown parameter $\theta_0 \in \Theta \subset \RR^p$ 
of dimension $p$ 
is defined  as a solution to 
\begin{equation}\label{eq:est}
   \EE\big\{ s(D,\theta)\big\} = 0 
\end{equation}
for some score function $s \colon \cD \times \Theta \rightarrow  \RR^p$, 
where $E$ denotes the expectation under $\mathbb{P}$.  
Here and in the following, we 
adopt the common assumption in the literature~\citep{Vaart1998} that the solution to equation~\eqref{eq:est} is unique.  
Inference for $\theta_0$ is often  
based on i.i.d.\ data
$\{D_i\}_{i=1}^n$ from $\mathbb{P}$. 
 
In situations where some attributes are fixed, 
as the inferential target  we consider  
a functional of a conditional distribution. 
Following \cite{abadie2014inference,Buja2016,Buja2019}, 
we 
suppose $(D_1,Z_1),\ldots,(D_n,Z_n)$ 
are i.i.d.~from an unknown distribution $\mathbb{P}$, 
where $D_i \in\mathcal{D}$ are the full data, 
and $Z_i \in\mathcal{Z} $ are the 
attributes we condition on.  
Given the attributes $Z_{1:n} = (Z_1,\dots,Z_n)$, 
the data 
$D_{1:n}=(D_1,\dots,D_n)$ are from  
the conditional distribution given $Z_{1:n}$.  
The i.i.d.~assumption on the attributes 
could be relaxed later on; for now, 
we keep 
it for consistency across all scenarios. 
We define the $Z_{1:n}$-conditional parameter 
$\theta_{n}^\cond = \theta_n^\cond(Z_{1:n})$ 
as the solution to
\#\label{eq:def_cond_para}
    \sum_{i=1}^n \EE\big\{  s(D_i,\theta) \biggiven Z_i\big\} = 0,
\#
which is assumed to be unique.   
From a marginal perspective, 
$\theta_n^\cond$ depends on $Z_{1:n}$ 
and is thus random.  
$\theta_n^\cond$ describes properties of $D_{1:n}$ 
given that $Z_{1:n}$  
are fixed at the realized values, 
as opposed to $\theta_0$ that describes the 
unconditional super-population. 
Conditional parameter can also 
be seen as a generalization of  super-population parameters in the sense that $\theta_0 = \theta_n^\cond(\emptyset)$.

\begin{example}[Patient health]
\label{ex:health}
In our motivating example 
of patient health,  
we let $Z$ be the attributes, 
$D=Y\in \RR$ be the measured health condition,
and assume one observes i.i.d.~data 
$(D_i,Z_i)_{i=1}^n$ from  $\mathbb{P}$. 
The super-population parameter is $\theta_0 = \EE(D_i)$, 
the solution to~\eqref{eq:est} 
where $s(D,\theta) = D-\theta$. 
If we view the 
attributes as fixed, the conditional parameter is 
$\theta_n^\cond = \frac{1}{n}\sum_{i=1}^n \EE(D_i\given Z_i)$, 
i.e., the average of conditional means 
for these patients,
marginalizing out other unexplained variations. 
If we view some (perhaps unobserved) 
underlying health condition $Y_{i}^*$
as fixed, then 
$\theta_n^\cond = \frac{1}{n}\sum_{i=1}^n \EE(D_i\given Y_{i}^*)$, 
marginalizing out the measurement error.  
\end{example}

Conditional parameters 
generalize several 
well-studied settings such as 
fixed-$X$ regression and finite-population causal effects, where 
fixing some  attributes 
leads to inference that is closely related to 
the populations at hand. 
{We defer detailed discussion on these examples in 
Section~\ref{app:subsec_causal} in the 
supplementary material, where conditional parameters can differ from unconditional ones.}

\subsection{Related work}\label{sec:related-work}

Several strands of literature have touched on conditional estimation or inference of a similar estimand as ours, usually with different guarantees or motivations from ours.  

\vspace{0.8em}

\noindent\textit{Conditional parameter with random covariates.}
There are several works 
that study the same  conditional parameters 
under similar assumptions yet with different guarantees, 
such as 
\citet{abadie2014inference} 
for conditional parameters for maximum likelihood and method of moments and
\citet{Buja2016,Buja2019} in the context of model misspecification. 
These works argue to treat the covariates as random 
and focus on marginal inference. 
Similar considerations also 
arise in the econometrics literature~\citep{manski1991,angrist1995estimating} that 
different sources
of variation may give different results.
We
substantially generalize  their framework 
by providing conditionally valid inference 
and studying transductive inference 
on new populations. 

\vspace{0.8em}

\noindent\textit{Asymptotics conditional on covariates.} 
More broadly, we connect  to a 
literature of asymptotics conditional on covariates. 
The major difference 
is that we study a new transductive inference 
setting. 
In fixed-design setting 
where the covariates are arbitrarily fixed,  
early works~\citep{white1980heteroskedasticity,goldberger1991course} study inference 
under well-specified models, 
and~\citet{fahrmexr1990maximum,kuchibhotla2018model,abadie2020sampling} study misspecified models.  
Among those for i.i.d.~attributes, 
the closest to ours is~\citet{andrews2019inference}, 
which derives conditionally valid confidence intervals for linear moment models.  
For inferring the population at 
hand,
assuming i.i.d.~attributes is not essential 
but does allow us to  consistently estimate the 
asymptotic variance, which is otherwise impossible~\citep{white1980heteroskedasticity,kuchibhotla2018model}. 
The i.i.d.\ assumption also guarantees that the proposed conditional confidence intervals are shorter than super-population confidence intervals.
In the new transductive inference problem, 
i.i.d.~attributes ensure  a shared structure 
between the two populations. 
Our transductive inference results potentially also generalize to 
fixed design, but for readability, we will state our results under an i.i.d.\ assumption.  

\vspace{0.8em}

\noindent\textit{Finite-population causal inference.}
In causal inference (see Example~\ref{ex:causal} 
in Section~\ref{app:subsec_causal} of the 
supplementary material for references), it is common to condition on potential outcomes and derive bounds for the asymptotic variance of estimators of causal effects. 
It usually does not rely on super-population assumptions, 
which is close to our fixed-design extension {(Remark~\ref{rmk:fix}); 
however, our framework implicitly assumes the treatment indicators
are mutually independent given the covariates, 
while in finite-population causal inference the treatment 
indicators can sometimes be dependent.}  
Furthermore, 
in this literature, 
conditional inference results are usually derived on a case-by-case basis, 
while we study a general class of estimators. 

\vspace{0.8em}

\noindent\textit{Distribution shift and missing data.} 
Our transductive inference is 
connected to a vast literature of inference under covariate shift, an important condition 
for transferring knowledge to new populations. 
In a general spirit, 
our method is similar to 
AIPW estimators~\citep{robins1994estimation}.  
\cite{rotnitzky2012improved} and~\cite{liu2020doubly} 
also study 
inference under unknown covariate shift  
with the doubly-robust property. 
The distinction  is that 
we provide conditional validity 
for conditional parameters instead  
of marginal validity for super-population quantities, 
leading to new targets and different variances. 
The estimands we study are also more general than theirs. 

\vspace{0.8em} 

\noindent\textit{Classical conditional inference.} A classical line of work \citep{hinkley1980likelihood,cox1987parameter} 
draw inference by conditioning on ancillary statistics or estimators of nuisance parameters (see e.g.,~a review in~\cite{CasellaConditional}), stemming from the ideas of~\citet{fisher1935statistical,fisher35inductive}. 
While we share the spirits 
of conducting inference that is closely related to the data at hand~\citep{fisher1937design,ernst2004permutation,edgington2007randomization}, 
compared to these works,  
we specifically condition on some attributes, leading to different parameters, different interpretations, and different inferential guarantees 
that rely on  the asymptotics of  
general semi-parametric and parametric estimators.


\section{Conditional inference}\label{sec:cond_inf}

\subsection{Conditional inference}
\label{subsec:guarantee}
 
Conditional parameters can 
describe a population at hand, or 
a new population with some observed attributes.  
It is desirable that 
confidence intervals are  conditionally valid 
for the specific sub-population we focus on,  
instead of marginalizing 
over all potential sub-populations.
We now formalize  
these two settings
and  
the conditional inference guarantees 
we aim to provide.

\vspace{0.4em}

\noindent\textit{Conditional inference for the population at hand.}  
As discussed earlier, the healthcare provider
might be interested in the health of 
its own patients, 
holding some intrinsic information as fixed 
and 
averaging over other variations. 
When inferring the population at hand,
we observe i.i.d.~data $\{(D_i,Z_i)\}_{i=1}^n$ 
from a super-population $\mathbb{P}$,  
where $Z_{1:n}=\{Z_i\}_{i=1}^n$ 
are the conditioning variables
(e.g., the attributes of the patients that are viewed as fixed), 
and $D_{1:n}=\{D_i\}_{i=1}^n$ are the observations 
(e.g., the observed health conditions). 
The conditional parameter $\theta_n^\cond = \theta_n^\cond(Z_{1:n})$ 
defined in~\eqref{eq:def_cond_para} provides 
a more precise characterization of the current patients
than the super-population 
quantity; the latter instead characterizes 
the overall  
health of a  hypothetical infinite patient base. 
In Section~\ref{sec:cond_inf_single}, 
we construct a confidence interval $\hat{C}(D_{1:n},Z_{1:n})$ 
obeying
\#\label{eq:cond_guarantee}
\PP\big\{\theta_n^\cond\in \hat{C}(D_{1:n},Z_{1:n})\biggiven Z_{1:n}\big\} \rightarrow  1-\alpha
\#
in probability 
as $n\to \infty$. 
Put another way,   
our inference on $\theta_n^\cond$ is 
valid conditional on any realized attributes. 
In our motivating example, 
the conditional guarantee~\eqref{eq:cond_guarantee} 
means the validity 
given the current patients. 
%
We will also extend
conditionally valid inference 
to situations where $Z_{1:n}$ is fixed at any value 
without being i.i.d., and 
where it is more reasonable to condition on 
some unobserved attributes 
$X_{1:n}$.  

\vspace{0.4em}

\noindent\textit{Transductive inference  for 
a new population.}
The healthcare provider might also be interested in 
estimating the health condition of another  subgroup 
of its patients, 
based on measurements of the first subgroup of 
patients. 
We formalize this problem as follows. 

We denote the target data as 
$\{(D_j^\new, Z_j^\new)\}_{j=1}^m$ from 
a super-population $\mathbb{Q}$, 
where $Z_{1:m}^\new=\{Z_j^\new\}_{j=1}^m$ 
are the new attributes we condition on, 
and $D_{1:m}^\new = \{D_j^\new\}_{j=1}^m$ 
are the unobserved data (e.g., 
the health measurements of the target units). 
The source units $\{(D_i,Z_i)\}_{i=1}^n $ 
are i.i.d.~from a super-population $\mathbb{P}$ 
(e.g., the health measurements and attributes 
of the source units). 
For transductive inference, we always impose the 
super-population assumption on the attributes to ensure  
sufficient structure.
The quantity of interest is 
$\theta_m^{\cond,\new} = \theta_m^{\cond}(Z_{1:m}^\new)$ 
as a functional of 
the conditional distribution of $D_{1:m}^\new$ given $Z_{1:m}^\new$. 
In Sections~\ref{sec:trans_iid} and~\ref{sec:known_shift}, 
we construct a confidence interval $\hat{C}(D_{1:n}, {Z}_{1:n},Z_{1:m}^\new)$ 
that obeys 
\$
\PP\big\{\theta_m^{\cond,\new}\in \hat{C}(D_{1:n},Z_{1:n},Z_{1:m}^\new)\biggiven Z_{1:m}^\new,Z_{1:n}\big\} \to  1-\alpha 
\$
in probability 
as $m,n\to \infty$. 
In particular, 
we allow $\mathbb{Q}$ to admit a covariate shift
$w(z) = d{\mathbb{Q}} / d\mathbb{P}(d,z)$ from the fully observed data, 
and the conditional distribution of $D$ given $Z$ 
is invariant.  
When $w(z)$ is unknown and needs to be estimated from data, 
our procedure yields valid inference even if nuisance components are estimated at slow rates.

\subsection{Conditional inference for the population at hand}
\label{sec:cond_inf_single}

Recall the motivating example 
where the healthcare provider is interested in 
a conditional 
parameter that is specific to its current patients.
In this part, we construct 
confidence interval 
with conditional validity. 
Our results imply that 
 inference for super-population quantities 
can be overly conservative, 
since it unnecessarily takes into account 
the variation in the attributes.

As introduced in Section~\ref{subsec:guarantee}, 
we assume access to i.i.d.~data $\{(D_i,Z_i)\}_{i=1}^n$ 
from a super-population $\mathbb{P}$. The conditional parameter is defined 
in equation~\eqref{eq:def_cond_para}. 
For simplicity of illustration, 
we present our theoretical results for $p=1$ 
throughout the rest of the paper, 
while all of them can be generalized to fixed-$p$ settings with 
variances replaced by covariance matrices.

Assume we are given an asymptotically linear estimator $\hat \theta_n = \hat\theta_n(D_{1:n}) \in \RR$, i.e.,
\begin{equation}\label{eq:linear_hat_n}
    \sqrt{n} ( \hat \theta_n - \theta_0 ) = \frac{1}{\sqrt{n}} \sum_{i=1}^n \phi(D_i) + o_{P}(1),
\end{equation}
for some $\phi \in L_2(\mathbb{P})$ with mean zero. 
Many parametric and semi-parametric estimators 
are asymptotically linear in standard asymptotics, 
see, e.g.,~\cite{Vaart1998} or \cite{tsiatis2007semiparametric}. 
Under regularity conditions, the conditional parameter  \eqref{eq:def_cond_para} satisfies
\begin{equation}\label{eq:exp}
    \sqrt{n} ( \theta_n^{\cond} - \theta_0 ) = \frac{1}{\sqrt{n}} \sum_{i=1}^n \EE\big\{\phi(D_i)\biggiven Z_i\big\} + o_P(1).
\end{equation} 
Note that~\eqref{eq:exp} 
implies the conditional parameters converge to the super-population 
parameters when the attributes are i.i.d. Their difference is of order $O_P(n^{-1/2})$, 
which is of the same order as the difference between 
the 
estimator and super-population parameters in~\eqref{eq:linear_hat_n}. 
Although such difference converges to zero, 
it is not negligible 
in standard statistical settings; 
here, it translates to 
 conditional confidence intervals that 
are shorter by a constant factor than marginal ones.

We establish sufficient conditions for~\eqref{eq:linear_hat_n}-\eqref{eq:exp} to hold for i.i.d.~$(D_i,Z_i)$. The proof of Proposition~\ref{prop:linear_exp} is in Section~\ref{app:linear_z} 
of the supplementary material.

\begin{proposition}[Asymptotic linearity of conditional parameters]\label{prop:linear_exp}
Suppose the following conditions hold:  
(i)  $\hat\theta_n$ is the unique solution to $\sum_{i=1}^n s(D_i,\theta)=0$, 
$\theta_0$ is the unique solution to~\eqref{eq:est} 
and $\theta_n^\cond$ is the unique solution to~\eqref{eq:def_cond_para}.   
(ii)  The parameter space $\Theta$ is compact. (iii) In a small neighborhood of $\theta_0$, 
$s(D,\theta)$ and $t(Z,\theta) = \EE\{s(D,\theta)\given Z\}$ are twice differentiable in $\theta$, 
with 
$\dot{s}(D,\theta)= \nabla_\theta s(D,\theta) \in \RR^{p\times p}$ 
the derivative matrix of $s(D,\theta)$ at $\theta$
and $\ddot{s}(D,\theta) = \nabla_\theta \dot{s}(D,\theta)$
the derivative tensor of $\dot{s}(D,\theta)$ at $\theta$. 
Additionally, 
$\dot{t}(Z,\theta) = \nabla_\theta t(Z,\theta) = \EE\{\dot{s}(D,\theta)\given Z\}$
and $\ddot{t}(Z,\theta) = \nabla_\theta \dot{t}(Z,\theta) = \EE\{\ddot{s}(D,\theta)\given Z\}$.
(iv) For each $j,k$, $\|\ddot{s}_{jk}(D,\theta)\|= \|\partial s(D,\theta)/\partial\theta_j \partial \theta_k\| \leq g(D)$ for some  $g$ with $\EE\{|g(D)|\}<\infty$.  
Also, the matrix $\EE\{\dot{s}(D,\theta_0)\}$ is  assumed to be non-singular. 
Then equations~\eqref{eq:linear_hat_n} and~\eqref{eq:exp} hold with influence function 
\#\label{eq:infl_form}
\phi(d) =  - \big[\EE\{\dot{s}(D ,\theta_0)\} \big]^{-1} s(d ,\theta_0),
\#
where all the expectations are induced by the joint distribution of $(D,Z)$. 
\end{proposition}

The conditions in Proposition~\ref{prop:linear_exp} 
resemble the well-established results for Z-estimators~\citep{Vaart1998}, 
and has been  informally stated in~\cite{Buja2016}. 
For the convenience of reference later, 
we impose the linear expansion as an assumption.

\begin{assumption} 
$\hat\theta_n$ and $\theta_n^\cond$ obey equations \eqref{eq:linear_hat_n} and \eqref{eq:exp}, respectively.
\label{assump:linear_main}
\end{assumption}

\begin{assumption} 
The influence function $\phi(\cdot)$ defined in~\eqref{eq:infl_form} satisfies $\EE\{\phi(D)^4\}<\infty$. 
\label{assump:moment_main}
\end{assumption}

Let $[a,b]$ denote 
the closed interval with endpoints $a,b\in \RR$, $a<b$. 
Theorem~\ref{thm:cond_intv}  constructs conditionally valid confidence intervals 
for conditional parameters, whose proof 
is deferred to Section \ref{app:cond_inf_single} 
in the supplementary material. 
 
\begin{theorem}[Asymptotic conditional validity]
Suppose Assumptions \ref{assump:linear_main} and \ref{assump:moment_main} hold. If 
    an estimator $\hat \sigma$ converges in probability to $\sigma>0$, 
    where 
\#\label{eq:def_asymp_var}
\sigma^2 = \EE\Big(\big[\phi(D) - \EE\{\phi(D)\given Z\}\big]^2\Big), 
\# 
then for any $\alpha\in (0,1)$, it holds that 
    the conditional coverage 
    \begin{equation}
        \PP \Big( \theta_n^{\cond} \in \big[\,\hat \theta_n -  z_{1-\alpha/2}  \hat \sigma/\sqrt{n},~ \hat \theta_n + z_{1-\alpha/2}  \hat \sigma/\sqrt{n}\,\big] \Biggiven Z_{1:n}\Big), 
    \end{equation}
    as a random variable measurable with respect to $Z_{1:n}$, 
    converges in probability to $1-\alpha$ as $n\to \infty$, where $z_{1-\alpha/2}$ is the $(1-\alpha/2)$ quantile of  standard Gaussian distribution. 
    \label{thm:cond_intv}
    \end{theorem}

The asymptotic conditional validity 
relies on 
the convergence of the conditional distribution 
of $\sqrt{n}(\hat\theta_n - \theta_n^\cond)$, 
derived from a conditional central limit theorem~\citep{dedecker2003conditional,grzenda2008conditional}; we include Lemma \ref{lem:cond_clt} in 
Section~\ref{app:subsec:cond_law} of 
the supplement for completeness.
As a clarification note, 
the conditional coverage 
converges in probability (with respect to 
the attributes) to the nominal level  
instead of uniformly over all possible values.  

It remains to construct a consistent estimator 
$\hat\sigma^2$ for the asymptotic variance~\eqref{eq:def_asymp_var}. 
In Section~\ref{sec:algo} of the supplementary
material, we describe 
a detailed stand-alone estimation procedure (c.f.~Algorithm~\ref{alg:sigma_est}) 
with consistency guarantees, 
relying on the formula~\eqref{eq:infl_form} 
and nonparametric regression for
$\varphi(Z)=\EE\{\phi(D)\given Z\}$. 
\cite{abadie2014inference} 
propose a matching-based algorithm
to estimate
the same asymptotic variance, 
whose proof relies on assuming compactness of $\cZ$ 
and smoothness of $\varphi(\cdot)$. 
In contrast, we prove that our estimator is consistent  
under generic consistency 
conditions on nonparametric regression. 
This relaxes the technical assumptions 
and overcomes the computational difficulty 
of matching in practice. 


Many results in the literature are close 
to Theorem~\ref{thm:cond_intv}, yet   
all providing marginal coverage guarantees~\citep{abadie2014inference},   
which
can be insufficient for reliable inference 
for a specific population (see Section~\ref{sec:super-vs-condit}). 
For fixed-design OLS,~\cite{kuchibhotla2018model} shows 
it is impossible to estimate the asymptotic variance without assumptions; instead, our sampling justification 
allows for consistent estimation of the variance 
and leads to a feasible conditional inference recipe.

\begin{remark}
Super-population inference carries out a similar protocol  
with an estimator of the (unconditional) asymptotic variance, 
usually of the form $\sigma_0^2 := \Var\{\phi(D)\}$. 
The variance 
for conditional inference is always no greater, as
$
\sigma^2  = \Var\{\phi(D)\} - \Var[ \EE\{\phi(D)\given Z\}]
\leq \Var\{\phi(D)\}.
$
Taking the OLS example,~\eqref{eq:linear_hat_n} and~\eqref{eq:exp} 
hold with  
$
\phi(D)= \{\EE (XX^\top ) \}^{-1} X (Y-X^\top \theta_0 ).
$  
If the linear model $Y=X^\top \theta_0 + \epsilon$ is well-specified, 
i.e., $\EE(\epsilon\given X) = 0$ a.s.,  
we have $\sigma^2=\sigma_0^2$ 
when $Z$ is contained in $X$.
With a mis-specified linear model,  
if $\EE(X\epsilon \given Z)$ is not a.s.~zero,  
our confidence interval is shorter 
than that for super-population inference. 
\end{remark}

We finally note two generalizations 
of the current framework of 
conditional inference. 
 
\begin{remark}[Non-i.i.d.~attributes]
\label{rmk:fix}
Conditional inference generalizes 
to fixed attributes $\{z_i\}_{i=1}^n$
without an i.i.d.~structure. 
In Section~\ref{app:fix} of 
the supplementary material, we 
provide a set of results that are 
parallel to this part, 
without any probabilistic assumption  
on the attributes. 
In that case, the asymptotic variance 
$\sigma_n^2$ for conditional inference
depends on $\{z_i\}_{i=1}^n$, 
whose estimation requires certain assumptions. 
For i.i.d.~attributes, one could also use $\sigma_n^2$ 
instead of $\sigma^2$  
for covariate-dependent uncertainty 
quantification  
while maintaining conditional validity; however,  
the difference is 
negligible. 
We discuss these issues in detail in the supplementary material. 
\end{remark}

\begin{remark}[Conditioning on unobserved variables]
In our motivating example, 
the practitioner may instead characterize the 
patients by fixing hidden intrinsic health $Y^*$. 
While this variable is unobserved, 
we could still conduct $\{Y_i^*\}_{i=1}^n$-conditionally 
valid inference 
with observed attributes $\{Z_i\}_{i=1}^n$. 
The resulting confidence intervals are 
shorter than super-population inference, 
yet perhaps unavoidably conservative for 
conditioning on unobserved variables. 
We provide formal results and detailed discussion  
in Section~\ref{app:cond_unobs} 
of the supplementary material. 
\end{remark}

\subsection{Transductive inference across data sets from the same super-population}  
\label{sec:trans_iid}

Prepared with the above conditional 
inference techniques, we now study 
transductive inference. 
This tackles situations 
where  
a healthcare provider has deployed 
a novel treatment plan for a subset of its patients, 
and would like to infer the effect 
on the remaining ones. 
To fix ideas, we first discuss the setting where 
the units in both populations 
are drawn from the same super-population. 
The case with different super-populations is discussed in Section~\ref{sec:known_shift}.

With access to i.i.d.~observations 
$\{(D_i,Z_i)\}_{i=1}^n \sim \mathbb{P}$, the 
i.i.d.~new units 
$\{(D_j^\new,Z_j^\new)\}_{j=1}^m\sim \mathbb{P}$ are
from the same distribution, 
where only $Z_{1:m}^\new:=\{ Z_j^\new \}_{j=1}^m$ 
are observed. 
We are interested in the finite population, 
which is from a conditional distribution 
given $Z_{1:m}^\new$.  
The 
new conditional parameter $\theta_m^{\cond,\new} := \theta_m^\cond(Z_{1:m}^\new)$ is 
the unique solution to
$
     \sum_{j=1}^m \mathbb{E} \{s( D_j^\new ,\theta)\given Z_j^\new \} = 0.
$

We assume  
an estimator $\hat\theta_n$ satisfies~\eqref{eq:linear_hat_n} 
(e.g., a Z-estimator given in Proposition~\ref{prop:linear_exp}), 
and similar to~\eqref{eq:exp},
$\theta_m^{\cond,\new }$ satisfies the asymptotic linearity 
$
\sqrt{m} (\theta_m^{\cond,\new }-\theta_0 )
= \frac{1}{\sqrt{m}}\sum_{j=1}^m \varphi(Z_j^\new)+o_P(1)
$
for $\varphi(\cdot):=\EE \{\phi(D_i)\given Z_i=\cdot\}$.  
We use $\hat \theta_n$ as a starting point and add a correction term to account for the fact that we target $\theta_{m}^{\text{cond},\text{new}}$. 
Specifically, we define 
\#\label{eq:hat_trans_iid}
\hat\theta_{m,n}^{\trans} = \hat\theta_n - \frac{1}{n}\sum_{i=1}^n \hat\varphi(Z_i) + \frac{1}{m}\sum_{j=1}^m \hat\varphi(Z_j^\new),
\#
where with a slight abuse of notation, we let 
$\hat\varphi(\cdot)$ 
be an estimator for 
$\varphi(\cdot)$ 
obtained  
from cross-fitting~\citep{chernozhukvo2018debiased}:  
we first randomly split $\cI=\{1,\dots,n\}$ 
into two equal-sized folds  $\cI_1$ and $\cI_2$, 
then use $\{(D_i,Z_i)\}_{i\in\cI_k}$ to obtain an estimator $\hat\varphi^{(k)}$ for $\varphi(\cdot)$ for each $k=1,2$ 
(a special case of Algorithm~\ref{alg:eta_est} in
the supplementary material 
by taking  weight $w(z)\equiv 1$ 
provides a detailed algorithm for estimating $\varphi$)
and then define
$\hat\varphi(Z_i)=\hat\varphi^{(k)}(Z_i)$ for $i\notin \cI_k$, 
and $\hat\varphi(Z_{j}^\new) = 
\{\hat\varphi^{(1)}(Z_j^\new)+ \hat\varphi^{(2)}(Z_j^\new)\}/2$ 
for all $j$. 

To gain some more intuition on the bias correction term, 
note 
the asymptotic expansion 
\$
\hat\theta_n - \theta_m^{\cond,\new} 
= \underbrace{\textstyle \frac{1}{n}\sum_{i=1}^n \phi(D_i)}_{\text{biased conditional on $Z_{1:n}$}} - 
\underbrace{\textstyle \frac{1}{m}\sum_{j=1}^m \varphi(Z_j^\new)}_{\text{  biased conditional on $Z_{1:m}^\new$}} + o_P(1/\sqrt{n}+1/\sqrt{m}). 
\$  
Conditional on $Z_{1:n}$ and $Z_{1:m}^\new$, 
the conditional mean of the first term is 
$\frac{1}{n}\sum_{i=1}^n\varphi(Z_i)$, 
and that of the second is 
$\frac{1}{m}\sum_{j=1}^m \varphi(Z_j^\new)$. 
These could be viewed 
as the \emph{conditional bias} of $\hat\theta_n$ for $\theta_m^{\cond,\new}$, and motivates 
our correction term in~\eqref{eq:hat_trans_iid}. 
Intuitively, 
correcting for this conditional bias
ensures that the resulting $\hat\theta_{m,n}^{\trans}$ 
centers around $\theta_m^{\cond,\new}$ 
conditional on $Z_{1:m}^\new$ and $Z_{1:n}$.

The following result shows the asymptotic 
conditional validity of confidence intervals based on 
our bias-corrected estimator. Its proof is in 
Section~\ref{app:thm_iid_simple} 
of the supplementary material. 

\begin{theorem}\label{thm:iid_simple}
Suppose $\hat\theta_n$ satisfies~\eqref{eq:linear_hat_n}, 
$\sqrt{m}\big(\theta_m^{\cond,\new }-\theta_0\big)
= \frac{1}{\sqrt{m}}\sum_{j=1}^m \varphi(Z_j^\new)+o_P(1)$, 
and Assumption~\ref{assump:moment_main} holds.
Assume an estimator $\hat\sigma^2$ converges in probability to 
$\sigma^2$ in~\eqref{eq:def_asymp_var}, 
and  
$
 \max_{k=1,2}\|\hat\varphi^{(k)}(\cdot)-\varphi(\cdot)\|_{L_2(\mathbb{P})}$ 
converges in probability to $0$. 
Let $\hat\theta_{m,n}^{\trans}$ 
be defined in~\eqref{eq:hat_trans_iid}.
Then 
\$
\PP\Big( \theta_m^{\cond,\new} \in \big[  
\hat\theta_{m,n}^{\trans} - z_{1-\alpha/2}\hat\sigma/\sqrt{n},~
\hat\theta_{m,n}^{\trans} + z_{1-\alpha/2}\hat\sigma/\sqrt{n}
\big] \Biggiven Z_{1:m}^\new, Z_{1:n} \Big)
\$
converges in probability to $1-\alpha$ as $n\to \infty$.
\end{theorem}

From a practical perspective, 
the above theorem enables targeted inference of sub-population parameters in settings where both groups follow the same distribution. For example, a company can run an experiment on 
a representative subset of the users, and then generalize the results to 
the other users based on 
their covariate information.

Let us discuss the mathematical consequences of this theorem.
First,
the length of the confidence interval 
is asymptotically the same as conditional inference, 
without any additional uncertainty from the new population. Thus, roughly speaking, in this setting, 
we do not pay any price for the parameter transfer in terms of asymptotic variance. 
Furthermore, 
such inference guarantee does not require any 
convergence rate of $\hat\varphi(\cdot)$; 
this is because the attributes in the two groups follow the same 
distribution, hence the cross-fitted $\hat\varphi(\cdot)$ 
is able to accurately cancel  out the 
variation in attributes. 
Similar ideas 
apply to settings with 
covariate shifts we study next.

\subsection{Transductive inference across distributions}\label{sec:known_shift}
 
In transductive inference, the 
first batch of patients and the target population 
might follow different distributions. 
In the following, we show that 
when the two distributions only differ in 
the covariate distribution, 
one could still  
reliably infer parameters for the new population. 
From now on, we assume the new i.i.d.~data
$\{(D_j^\new,Z_j^\new)\}_{j=1}^m\sim \mathbb{Q}$ with a  perhaps unknown 
covariate shift 
$
w(z)=  d\mathbb{Q}/ d \mathbb{P}(d,z) 
$, 
and assume 
$w(z)<\infty$ for $\mathbb{P}$-almost all $z$ 
to ensure transferability. 
The 
identical distribution setting 
is a special case with $w(z)\equiv 1$.

\begin{remark}
Covariate shift 
is a popular setting in 
machine learning~\citep{quinonero2008dataset} 
and 
social sciences~\citep{tipton2014sample,egami2021covariate}. 
In our context, it 
ensures identifiability 
of the new conditional parameter.   
It holds when the two populations are selected 
only based on the attributes, 
similar to the unconfoundedness assumption
in causal inference~\citep{Imbens2015}. 
In the context of Example~\ref{ex:health}, this assumption implies that for two patients from the two different distributions who have 
the same observed attributes,  the conditional distributions of 
their health outcomes must be identical. For example, if the observed attributes are age and gender, then men aged 75 must have the exact same distribution of health outcomes in the two populations. In addition, $w(z)<\infty$ 
resembles the overlap condition in causal inference, which 
rules out any sample space 
that is never observed under $\mathbb{P}$. 
\end{remark}

Recall that 
the new conditional parameter 
$\theta_m^{\cond,\new}$ is the (unique) solution to  
\#\label{eq:eq_new_cond}
     \sum_{j=1}^m \EE\big\{ s( D_j^\new ,\theta)\biggiven Z_j^\new\big\} = 0, 
\#
with the conditional expectation  induced by $\mathbb{Q}$. 
Given that  $\mathbb{P}_{D\given Z}=\mathbb{Q}_{D\given Z}$ 
are invariant, one might consider 
solving~\eqref{eq:eq_new_cond} for $\theta_{m}^{\cond,\new}$ 
by estimating $\EE\{s(D,\theta)\given Z=\cdot\}$ for every $\theta$. 
However, estimating infinitely many conditional expectations 
might be infeasible in general or lead to slow convergence rates that hinders statistical inference; 
we will briefly discuss other potential approaches in Remark~\ref{rm:est_eq}. 
Now, we describe a procedure that  estimates the new conditional parameter with $\sqrt{n}$-convergence rate even when the distribution shift is unknown. 

At a high level, our approach relies on 
the fact that  $\theta_m^{\cond,\new}$ 
is close to $\theta_0^\new$, the new super-population 
parameter, which is defined as 
the unique solution to 
\#\label{eq:theta_dag}
\EE\big\{ w(Z ) s(D ,\theta)\big\} =\EE_Q\big\{s(D^\new, \theta)\big\} = 0,
\#
with expectations over $(D,Z)\sim \mathbb{P}$ 
and $D^\new \sim \mathbb{Q}$, respectively. 
We will use the asymptotic linearity 
of $\theta_{m}^{\cond,\new}$
to 
correct for conditional bias and 
conduct conditionally valid inference.

 Let $\hat w (\cdot)$ be an estimator of $w(\cdot)$; 
 if $w$ is known, one can simply set $\hat w = w$. 
 We assume $\hat{w}$ is obtained from another independent set of data. 
 Alternatively, one could use cross-fitting~\citep{chernozhukvo2018debiased} 
 to yield
 the same guarantees under similar conditions 
 only using the data at hand. 
 However, since this increases the complexity of notation and exposition,  
 we defer the details to Appendix~\ref{app:subsec_cross_fitting_known}
 when $w(\cdot)$ is known, and Appendix~\ref{app:subsec_cross_fitting_est} 
 when $w(\cdot)$ is estimated.   

 To account for the covariate shift, 
we begin with a reweighted estimator 
 $\hat\theta_n^\trans$ 
 that is close to $\theta_0^\new$, 
 defined as the unique solution to 
\#\label{eq:hat_theta_dag}
\sum_{i=1}^n \hat{w}(Z_i) s(D_i,\theta) = 0.
\# 
It can be shown that $\hat\theta_n^\trans$ and $\theta_m^{\cond,\new}$ are asymptotically linear 
around $\theta_0^\new$ 
under mild assumptions. For simplicity of exposition, we will state these as assumptions; in Section~\ref{app:subsec_main_linear_trans}
of the supplementary material, 
we provide justifications  
under sup-norm consistency of $\hat{w}$ 
and mild regularity conditions 
that are 
similar to those in Proposition~\ref{prop:linear_exp}. 

\begin{assumption}
\label{assump:linear_expansion_known_shift}
As $m,n\to \infty$, 
$\sup_z|\hat{w}(z)-w(z)|\to 0$ in probability, and  
\# 
\sqrt{n}(\hat\theta_n^\trans - \theta_0^\new) &= \frac{1}{\sqrt{n}} \sum_{i=1}^n \psi(D_i)\hat w(Z_i) + o_P(1), \label{eq:new_hat_lin_exp}\\
\sqrt{m}(\theta_m^{\cond,\new} - \theta_0^\new )&= \frac{1}{\sqrt{m}}\sum_{j=1}^m \eta(Z_j^\new)
 + o_P(1),\label{eq:new_cond_lin_exp}
\# 
where 
$ 
\psi(d) = - \big(\EE_Q[\dot{s}(D^\new, \theta_0^\new)]\big)^{-1} s(d, \theta_0^\new)
$, 
and $\eta(z) =\EE[\psi(D_j^\new)\given Z_j^\new=z]$. 
\end{assumption}
  
Similar to the preceding subsection, 
we add a bias correction term to $\hat\theta_n^\trans$ 
and construct
\#\label{eq:trans_new_simple}
 \hat\theta_{m,n}^{\trans } = \hat\theta_n^\trans - \hat{c}^\trans, 
 \quad \hat c^\trans  :=   \frac{1}{n} \sum_{i =1}^n  \hat \eta (Z_i) \hat w(Z_i)  - \frac{1}{m} \sum_{j=1}^m   \hat\eta(Z_j^\new) .
\# 
Again, for ease of illustration, 
we assume $\hat\eta(\cdot)$ 
is an estimator for $\eta(\cdot)$ obtained elsewhere, 
such that it  
is independent of all the data we have. 
A rigorous treatment without referring to external datasets 
is in Sections~\ref{app:subsec_cross_fitting_known} 
and~\ref{app:subsec_cross_fitting_est} 
in the supplementary material. 
 
\begin{remark}\label{rm:est_eq}
The bias correction technique in~\eqref{eq:trans_new_simple} can be seen as 
generalizing the ideas in 
the missing data literature, e.g., 
the AIPW estimator~\citep{robins1994estimation}, 
where outcomes (score functions in our setting) 
in the original group are used for correcting 
for the bias in the estimated nuisance components. 
While our method directly focuses on 
the asymptotic linear expansion of our 
estimands and estimators, 
{there may be alternative estimators that yield similar guarantees.} 
In particular, aggregate the observations 
into $\{(T_i,T_iD_i,Z_i)\}_{i=1}^{m+n}$ 
and model them as i.i.d.~from a joint distribution, where $T_i\in\{0,1\}$ indicates 
whether unit $i$ belongs to the original group.
One may view $\theta_m^{\cond,\new}$ 
as solving the following estimation equation: 
\$
\sum_{i=1}^{m+n} \EE\bigg[ \frac{T_i}{pr(T_i=1)}w(Z_i)\big\{ s(D_i,\theta) - \eta(Z_i,\theta) \big\} + \frac{1-T_i}{pr(T_i=0)} \eta(Z_i,\theta) \bigggiven Z_i \bigg] = 0,
\$
where $\eta(Z_i,\theta):=\EE\{s(D_i,\theta)\given Z_i\}$ for any $\theta\in \RR^p$.   
This representation may also motivate 
other estimation approaches such as one-step estimators and under-smoothing~\citep{newey1994large}.
However, 
developing concrete procedures 
and theoretical guarantees for conditionally valid inference based on such an approach is 
beyond the scope of this work; 
we conjecture that such methods may 
lead to similar guarantees as those derived from our approach. 
\end{remark} 
 
Theorem~\ref{thm:est_cov_shift_simple} 
establishes conditional inference guarantee 
that is robust  
to estimation error: we obtain $n^{-1/2}$-rate inference,  
as long as the product of the errors 
in the estimation of $\hat{w}$ and $\hat\eta$ 
is no greater than $o_P(n^{-1/2})$.  
As a special case, 
when $w(\cdot)$ is known, 
we achieve conditionally valid inference 
under $L_2$-consistency of $\hat\eta$ 
similar to Theorem~\ref{thm:iid_simple}.
In the rigorous treatment with cross-fitting, 
the same result holds under similar 
convergence rates of the estimated covariate shift 
and influence functions;   
these are shown to be achievable 
under generic conditions for nonparametric regression,  
see Proposition~\ref{prop:consist_eta} 
of Section~\ref{subsec:est_consist} 
in the supplementary material. 
The proof of Theorem~\ref{thm:est_cov_shift_simple} 
is in Section~\ref{app:thm_cov_shift_simple} 
of the supplementary material. 

\begin{theorem}\label{thm:est_cov_shift_simple}
Under Assumption~\ref{assump:linear_expansion_known_shift}, 
suppose $m\geq \epsilon n$ for a fixed $\epsilon>0$, $\| \hat\eta (\cdot) - \eta(\cdot) \|_{L_2(\mathbb{Q})}=o_P(1)$, 
$\|w(\cdot)\{\hat\eta (\cdot) - \eta(\cdot)\}\|_{L_2(\mathbb{P})}=o_P(1)$,  
$\EE_{\mathbb{P}}[w(Z_i)^4\psi(D_i)^4]<\infty$, and 
$
\big\|\hat{w} (\cdot) - w(\cdot) \big\|_{L_2(\mathbb{P})} \cdot \big\|  \hat\eta (\cdot) - \eta(\cdot)\big\|_{L_2(\mathbb{P})} = o_P(1/\sqrt{n})
$.  
If an estimator $\hat\sigma_{\shift}^2$ 
converges in probability
to 
\#\label{eq:def_shift_var}
\sigma^2_{\shift} = \Var\big[ w(Z_i) \big\{\psi(D_i) - \eta(Z_i)\big\}   \big],
\#
where the variance is induced by $(D_i,Z_i)\sim \mathbb{P}$,
then the random variable 
\$
\PP\Big( \theta_m^{\cond,\new} \in \big[ \hat\theta_{m,n}^{\trans } - \hat\sigma_{\shift} \cdot  z_{1-\alpha/2}/\sqrt{n},    \hat\theta_{m,n}^{\trans } + \hat\sigma_{\shift} \cdot z_{1-\alpha/2}/\sqrt{n}    \big] \Biggiven Z_{1:m}^\new, Z_{1:n}  \Big)
\$
converges in probability to $1-\alpha$ as $n\to \infty$, 
where $\hat\theta_{m,n}^{\trans }$ is defined 
in equation~\eqref{eq:trans_new_simple}. 
\end{theorem}

Theorem~\ref{thm:est_cov_shift_simple} 
shows 
how to conduct estimation and inference of sub-population parameters under distribution shift.  
For instance, after running an experiment 
on a set of patients, 
a hospital could infer its effect  on a 
new set of patients who have a different covariate distribution.

Perhaps surprisingly, 
even with distribution shift,  
the asymptotic variance $\sigma_{\shift}^2/n$ 
does not depend on $m$; this is due to the fact that bias correction is statistically an easy task. 
We pay some price for the transfer to the new super-population 
since the variance term is weighted with $w(\cdot)$; 
however, we do not pay any price 
in efficiency for the transfer to the new sub-population.

To complete the picture, 
it remains to construct a consistent estimator 
for $\sigma_{\shift}^2$ defined in~\eqref{eq:def_shift_var}.  
In Section~\ref{sec:algo} 
of the supplementary material, 
we detail a stand-alone estimation procedure 
for $\sigma_{\shift}^2$ 
(Algorithm~\ref{alg:sigma_shift_est}) 
that does not rely on external data. 
An intermediate step 
relies on estimating $\hat\eta(\cdot)$; 
we offer a detailed procedure (Algorithm~\ref{alg:eta_est})
in Section~\ref{sec:algo}
with rigorous guarantee, 
Theoretical analysis for these algorithms 
is in Section~\ref{subsec:est_consist} 
of the supplementary material.

\begin{remark}[Transfer to super-populations] 
We have described how to conduct inference for the sub-population parameter $\theta_m^{\cond,\new}$, the parameter for the new distribution conditionally on $(Z_1^\new,\ldots,Z_m^\new)$. 
Practitioners may also be interested in the super-population parameter of the new distribution or may want to condition on a different set of variables. The proposed approach can be extended to this setting by adjusting the confidence intervals appropriately. 
We discuss these issues in more detail in 
Section~\ref{app:subsec_trans_subset} of 
the supplementary material.  
\end{remark}


\section{Simulations}\label{sec:simu}

\subsection{Conditional inference}
\label{subsec:simu_cond_inf}
In this part, 
we evaluate the conditional inference 
procedure in Section~\ref{sec:cond_inf_single} 
with simulations.
The results 
validate the conditional coverage 
and show the robustness 
to estimation error. 

We generate data $D_i=(X_i,Y_i)$ with covariates $X\in \RR^{10}$ and response $Y\in \RR$ according to
\$
&X_1,X_2,X_5,\dots,X_{10}\iid N(0,1),~X_3=X_1+\varepsilon_1,~X_4=X_1+\varepsilon_2,\\
&(\varepsilon_1,\varepsilon_2)^\top \sim N(0,\Sigma),~\Sigma_{11}=\Sigma_{22}=1,~\Sigma_{12}=\Sigma_{21}=1/2,\\ 
& Y = X_1+|X_1|+X_3+ \varepsilon',~~\varepsilon'\sim N(0,\nu^2).
\$
Here the linear model is misspecified but the OLS projection coefficient is still well-defined. 
We focus on two
conditional parameters: 
the first two entries  of the ordinary least square coefficient
$
\theta_n^\cond = \argmin_{\beta\in \RR^p} \sum_{i=1}^n\EE\{ (Y_i - \beta^\top X_i)^2\given Z_i\}, 
$ 
where the conditioning set is $Z=(X_1,X_2)$. 
The super-population estimands are $\theta_1 = 1$ and $\theta_2=0$.
The influence function is 
\#\label{eq:simu_inf_fct_form}
\phi(d;\theta) = \big\{\EE(X X^\top)\big\}^{-1} x(y-\theta^\top x),\quad \text{where }~ d=(x,y)\in \RR^p\times \RR.
\#

The procedure in Section~\ref{sec:cond_inf_single} 
is carried out for sample sizes $n\in\{200,1000,2000, 5000\}$ 
and $\nu\in\{0.1,0.2,0.5\}$ with $\alpha=0.05$. 
We first generate i.i.d.~observations $\{Z_i\}_{i=1}^n = \{(X_{i1},X_{i2})\}_{i=1}^n$; 
then we repeatedly sample $\{D_i\}_{i=1}^n$ 
conditional on $\{Z_i\}_{i=1}^n$ 
for $N_Y=10000$ times. 
We construct 
confidence intervals 
and evaluate the coverage of 
the two conditional parameters over $N_Y$ times. 
The asymptotic variance 
is estimated with Algorithm~\ref{alg:sigma_est} 
in the supplementary material,
where we use  \texttt{loess} function in R
for  the nonparametric regression. 
The procedure is repeated 
for $N_X=2000$ draws of 
the conditioning set.
 
We summarize the $N_X$ conditional coverage for $\theta_n^{\cond}$  
in Figure~\ref{fig:cond_cov}; 
each subplot corresponds to a configuration of $\nu$. 
Both figures confirm the conditional validity of our procedure 
(the boxplots mark the median and quarter quantiles of the conditional coverage). 
In particular, the estimation error of variance for the second entry 
with smaller sample sizes 
leads to overcoverage 
on the right-hand side of Fig.~\ref{fig:cond_cov}. 
It shows the robustness of our procedure 
to the estimation error of $\varphi(\cdot)$: 
in cases where the estimation of $\varphi(\cdot)$ is inaccurate, the algorithm tends to overestimate the variance, 
so that the procedure still provides valid coverage. 
This is because
using Algorithm~\ref{alg:sigma_est} (see details in Section~\ref{sec:algo} of the supplementary material), when $\hat\varphi(\cdot)$ 
converges to a function  $\varphi'(\cdot)$, 
our output $\hat\sigma^2$ 
converges to $\EE[\{\phi(D)-\varphi'(Z)\}^2] \geq \EE[\{\phi(D)-\varphi (Z)\}^2]=\sigma^2$, as $\varphi(Z)$ 
is the least-square projection of $\phi(D)$ onto 
the space of measurable functions of $Z$.

\begin{figure}[h]
  \centering
  \begin{minipage}{.5\textwidth}
    \centering
    \includegraphics[width=1\linewidth]{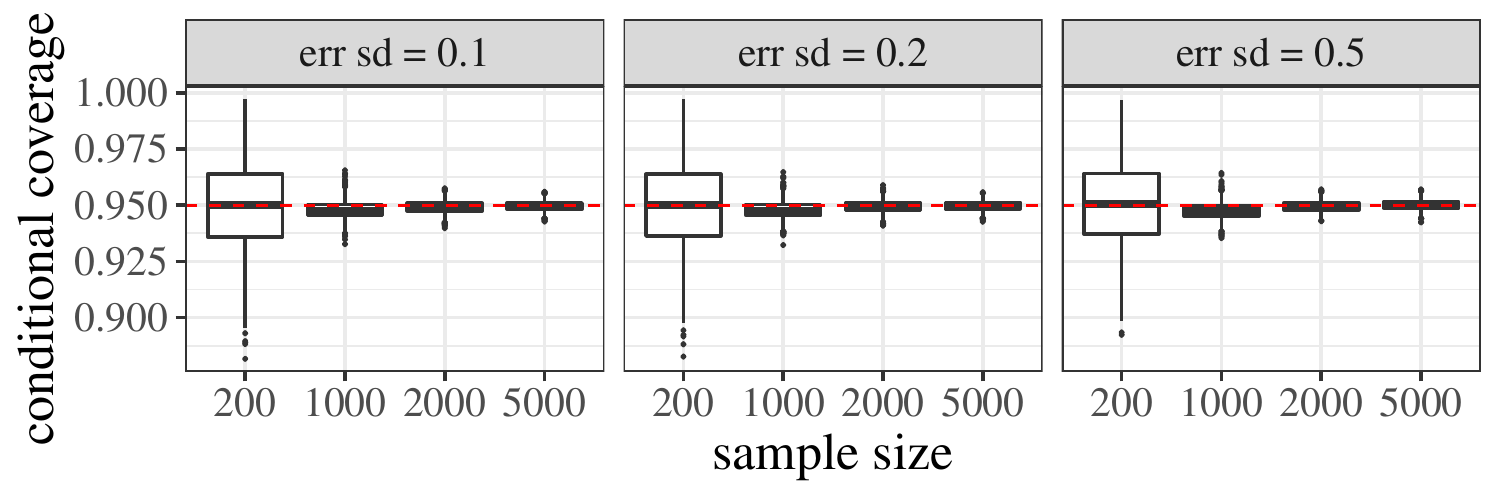}
  \end{minipage}%
  \begin{minipage}{.5\textwidth}
    \centering
    \includegraphics[width=1\linewidth]{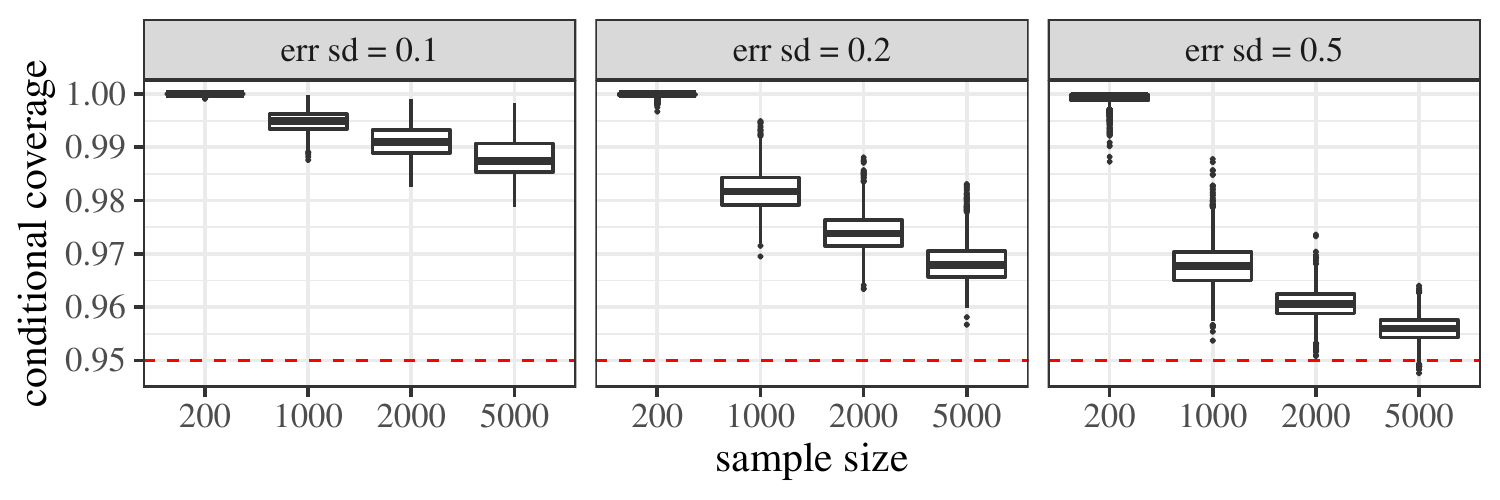}
  \end{minipage}
  \caption{Conditional coverage of $\theta_n^{\cond}$  for the first (left) and second (right) entry. 
  Red dashed lines are the nominal level $1-\alpha=0.95$.
  }
  \label{fig:cond_cov}
\end{figure}

Furthermore, 
we see that  
conditional inference leads to shorter 
confidence intervals once the estimation error 
is reasonably small; see Figure~\ref{fig:cond_len} in Section~\ref{app:subsec_simu}
of the supplementary material.

\subsection{Transductive inference under covariate shift}
\label{subsec:simu_trans}  
In this part, 
we evaluate the
transductive inference procedures.  
Our results show that 
the conditional coverage  
is close to the nominal level
even with estimated covariate shift. 
 
The data-generating process and 
parameters of interest are the same as Section~\ref{subsec:simu_cond_inf}, 
while we set the conditioning set as $Z=X_1$ 
and  the covariate shift as $w(z)=0.5 + \ind\{z>0\}$.
We set sample sizes $n \in \{200,1000,2000,5000\}$ and $m=n\cdot \epsilon$, where $\epsilon \in \{0.5, 1, 2\}$. 
We independently draw $N_X=2000$ times 
of i.i.d.~attributes $Z^\new=(Z_j^\new)_{1\leq j\leq m}$. 
Each time, we
fix the new attributes  
and repeatedly draw $\{D_i,Z_i\}_{1\leq i\leq n}$,  
then apply the procedures in Section~\ref{sec:known_shift} 
for $N_Y=10000$ times. 
We follow algorithms in Section~\ref{sec:algo} 
in the supplementary material 
to construct $\hat\sigma_{\shift}^2$, 
$\hat\theta_{m,n}^{\trans}$ and the confidence intervals, 
where the meta algorithm~\ref{alg:meta_cond_reg} 
uses the \texttt{loess} function in R. 
When covariate shift is estimated, 
we let $\hat{w}(\cdot) = \frac{\hat{e}(\cdot) }{1-\hat{e}(\cdot)}\cdot\frac{1-\hat p}{p}$, 
where $T_i=\ind\{i~\text{is in the new dataset}\}$, 
and $\hat{e}(x)$ (resp.~$\hat{p}$) estimates 
$\PP(T_i =1\given X_i=x)$  
(resp.~$\PP(T_i=1)$)
by pooling the two datasets, 
and $\hat{e}(x)$ is obtained by \texttt{randomForest} function in R.  
 
Given $\alpha=0.05$, we evaluate 
the  conditional coverage of the two procedures 
given each draw of new attributes
by empirical coverage among the $N_Y=10000$ replicates. 
Coverage for $\theta_m^{\cond,\new}$ 
associated with the first (left) and second (right) entries
is in Figure \ref{fig:simu_trans_cov_cond}. 
The conditional coverage is close to the nominal level $95\%$ 
with 
both ground truth (blue) and estimated (yellow) covariate shift. 
The proposed procedure works slightly better with larger noise $\nu$; 
it is due to over-estimation of asymptotic variance. 
Also, the coverage is higher for large proportion of $m/n$. 
This might be due to smaller approximation error 
of asymptotic linear expansion. 

\begin{figure}[h]
  \centering
  \begin{minipage}{.5\textwidth}
    \centering
    \includegraphics[width=1\linewidth]{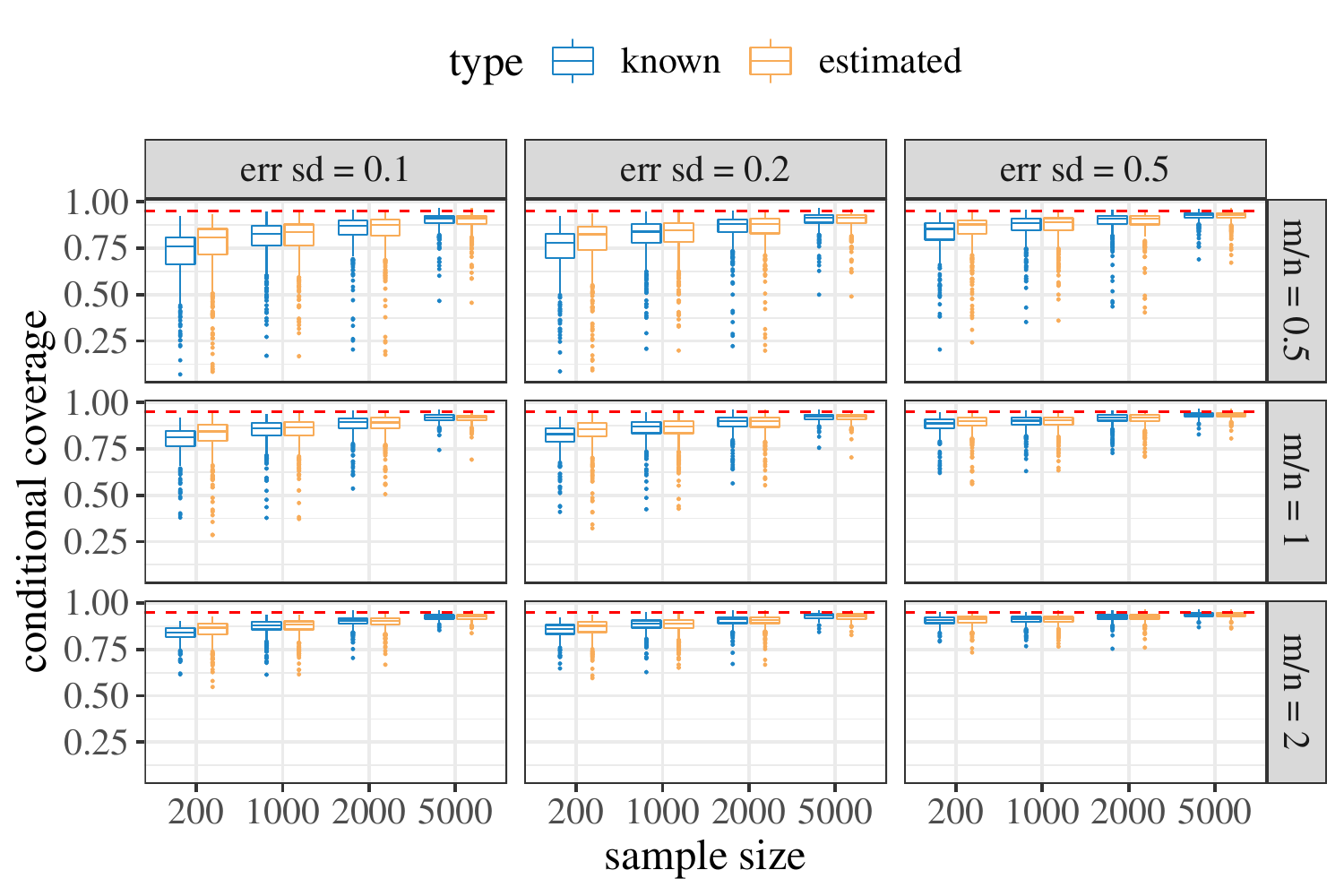}
  \end{minipage}%
  \begin{minipage}{.5\textwidth}
    \centering
    \includegraphics[width=1\linewidth]{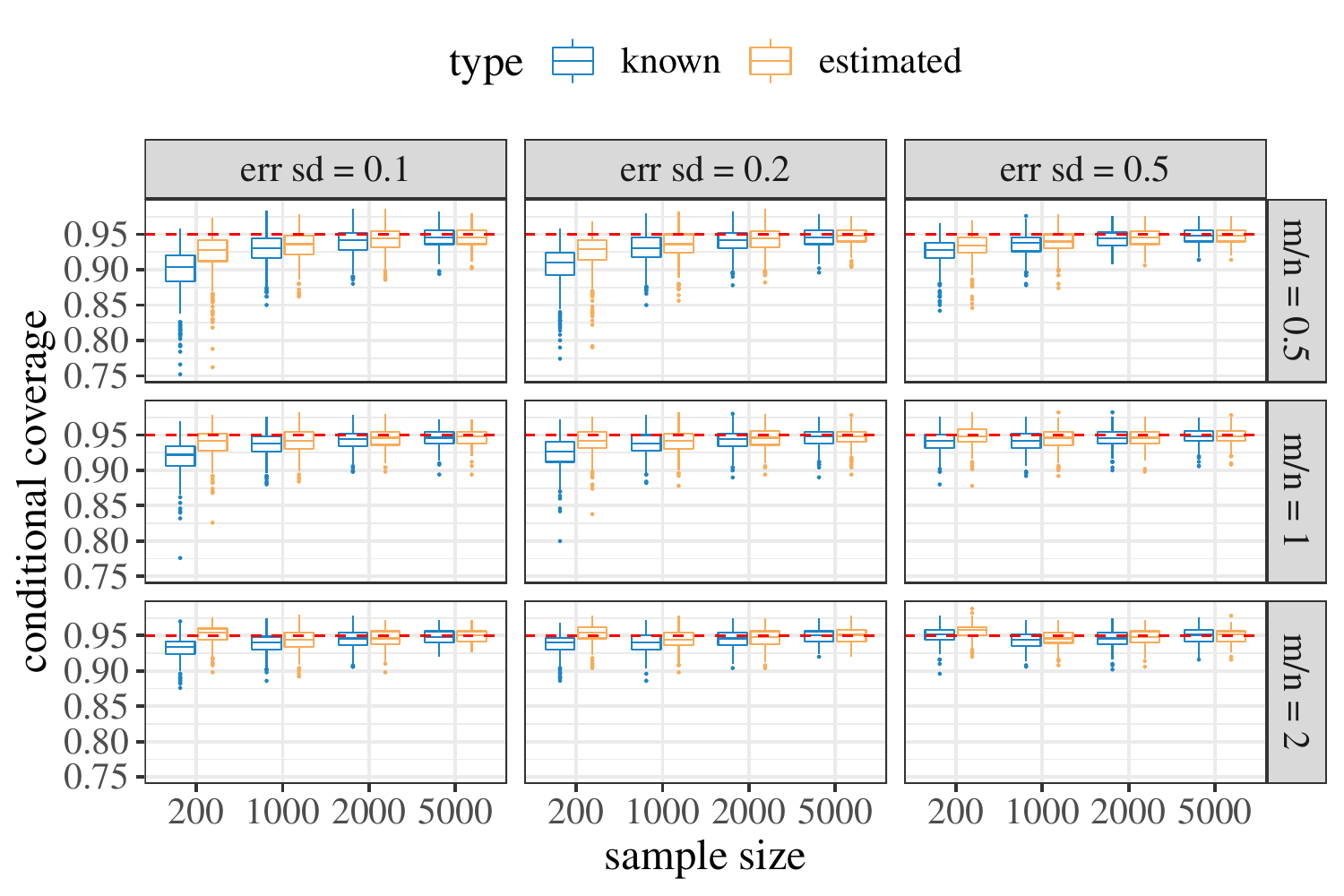}
  \end{minipage}
  \caption{Conditional coverage of $\theta_m^{\cond,\new}$  
  for the first (left) and second (right) entry. 
  Red dashed lines are the nominal level.}
  \label{fig:simu_trans_cov_cond}
\end{figure}

\vspace{-1em}

\section{Real data analysis}\label{sec:realdata}

Besides the real data analysis we show in the introduction, 
we also apply the 
transductive inference procedure in Section \ref{sec:known_shift} 
to a real-world dataset for predicting car prices. 
The dataset is from Ebay-Kleinanzeigen and consists of around 50,000 observations. Features include continuous ones like registration year and discrete ones like brand and make. 
The dataset has been studied in \cite{kuenzel2019heterogeneous}, where reliable prediction of car prices is found to be challenging. 
In particular, it is difficult to predict the individual prices of some `unsual' cars, 
such as old cars (registered before 2000), vintage cars and race cars. 

Our framework constructs 
conditionally valid confidence intervals 
for the mean price of a subset of cars. 
This is suitable when a dealer or agent is interested in 
whether to sell its own cars through this platform. 
This situation, as we introduced at the beginning of this paper, 
is in between predicting an individual price or 
inferring the overall mean price of cars. 
In the following, we conduct conditional inference for 
the mean of a sub-population of old cars and evaluate the performance by the conditional coverage. 

We first generate a semi-synthetic dataset for evaluation.  
We fit a random forest model $\hat{m}(\cdot)$ for the conditional mean $m(x)=\EE(Y_i\given X_i=x)$ 
on the whole dataset, and view the fitted values $\hat{m}(X_i)$ as the conditional mean, 
then compute the residuals $\epsilon_i = Y_i - \hat{m}(X_i)$. 
To create the synthetic dataset, 
we randomly sample (without replacement) 
a population of size $N\in \{2, 5, 10, 20, 50\}\times 10^3$ 
from the original dataset.  
We focus on the particularly difficult task 
of inferring the price of old cars~\citep{kuenzel2019heterogeneous}. 
We choose the old cars with registration year earlier than 2000, 
and take a subsample of proportion $r\in \{0.1,0.2,\dots,0.9\}$ 
as the new (shifted) dataset $\{(Y_j^{*\new},X_j^\new)\}_{j=1}^m$; 
The original dataset $\{(Y_i^*,X_i)\}_{i=1}^n$  
consists of the rest of the old cars 
and all newer cars, so that $m+n=N$. 
In particular, we fix the covariates and 
randomly resample the errors 
to generate $\{Y_i^*\}_{i=1}^n$ and $\{Y_j^{*\new}\}_{j=1}^m$, 
and evaluate conditional coverage. 

The transductive inference procedure discussed in Section \ref{sec:known_shift} is applied to the synthetic dataset, where the confidence interval is constructed as
\$
\big[\, \hat\theta_{m,n}^{\trans,\shift} + z_{0.025} \cdot \hat\sigma_{\shift}/\sqrt{n} ,~\hat\theta_{m,n}^{\trans,\shift} + z_{0.975} \cdot \hat\sigma_{\shift}/\sqrt{n}   \,  \big].
\$
Specifically, with 
$T_i=1$ indicating $(X_i,Y_i^*)$ is in the new (shifted) dataset, 
the weight function 
is obtained by $\hat{w}(\cdot) = \frac{\hat{e}(\cdot) }{1-\hat{e}(\cdot)}\cdot\frac{1-\hat p}{p}$, 
where $\hat{e}(x)$ estimates 
$\PP(T_i =1\given X_i=x)$  
and $\hat{p}$ estimates $\PP(T_i=1)$
by pooling the two datasets. 
The coverage for the conditional parameter 
$
\theta^{\cond,\new} = \frac{1}{m}\sum_{i=1}^N T_i \cdot \hat{m}(X_i)
$
is evaluated over 1000 replicates; 
the results are summarized in Figure \ref{fig:realdata_car}. 

\begin{figure}[h]
    \centering
    \includegraphics[width=6in]{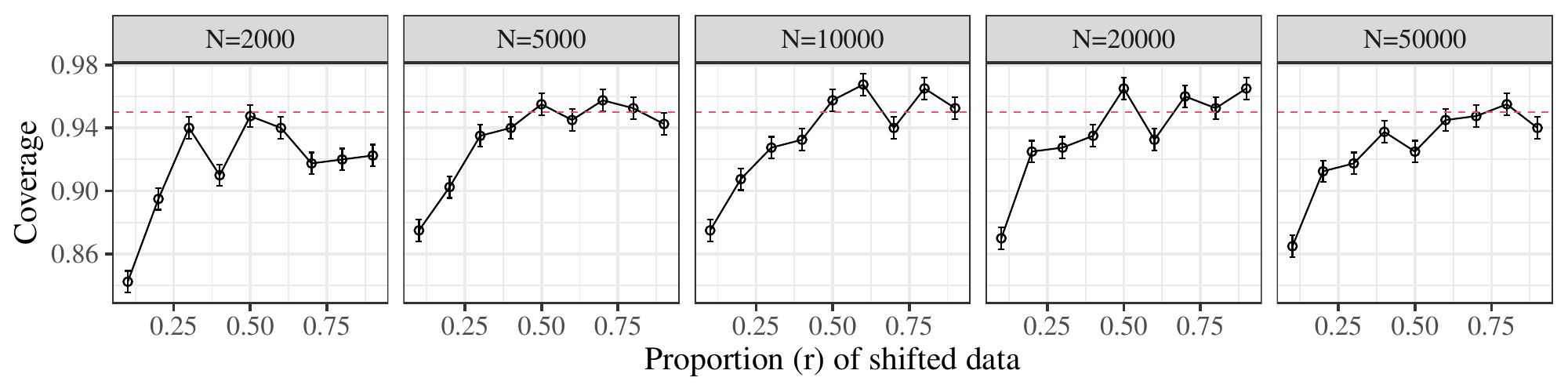}
    \caption{Conditional coverage versus proportions $r$ of shifted data; 
    each subplot corresponds to a sample size $N$.
    The red dashed lines indicate the nominal level 0.95.
    }
    \label{fig:realdata_car}
\end{figure}

Our procedure works well
especially if the original data set is reasonably large and the proportion $r$ of shifted data is moderate.  
The coverage improves as the sample size gets larger, especially when the proportion of shifted data is not too large or too small 
(so that old cars appear reasonably often in both datasets). 
We observe that the coverage might be deteriorated when the proportion of shifted data is large (like $r=0.9$), in which case there are fewer representative  observations 
of old cars in the original data, 
so that training a model for those conditional means gets harder. 
On the other hand, when the sample size (for example $N=2000$) and 
the proportion $r$ (such as $r=0.1$) is relatively small, 
we observe undercoverage  in the first plot in Figure \ref{fig:realdata_car}. This is because the sample 
size of the new data $m$ is relatively small, and the 
normal approximation of $\theta_m^{\cond,\new}$ 
such as imposed in Theorem~\ref{thm:iid_simple} is less accurate.
Meanwhile, the outliers in this car dataset potentially
make  inference for small subsets less stable. 

\section*{Acknowledgement}
We are grateful to three anonymous referees, the associate editor, and the editor for valuable comments and suggestions.
We thank Peng Ding, Kevin Guo, Guido Imbens,
and Zhimei Ren for helpful discussions and feedback. 

\bibliographystyle{apalike}
\bibliography{reference}

\begin{thebibliography}{}

\bibitem[Abadie et~al., 2020]{abadie2020sampling}
Abadie, A., Athey, S., Imbens, G., and Wooldridge, J. (2020).
\newblock Sampling-based versus design-based uncertainty in regression
  analysis.
\newblock {\em Econometrica}, 88(1):265--296.

\bibitem[Abadie et~al., 2014]{abadie2014inference}
Abadie, A., Imbens, G.~W., and Zheng, F. (2014).
\newblock Inference for misspecified models with fixed regressors.
\newblock {\em Journal of the American Statistical Association},
  109(508):1601--1614.

\bibitem[Andrews et~al., 2019]{andrews2019inference}
Andrews, I., Roth, J., and Pakes, A. (2019).
\newblock Inference for linear conditional moment inequalities.
\newblock Technical report, National Bureau of Economic Research.

\bibitem[Angrist, 1995]{angrist1995estimating}
Angrist, J. (1995).
\newblock Estimating the labor market impact of voluntary military service
  using social security data on military applicants.

\bibitem[Arceneaux et~al., 2006]{arceneaux2006comparing}
Arceneaux, K., Gerber, A.~S., and Green, D.~P. (2006).
\newblock Comparing experimental and matching methods using a large-scale voter
  mobilization experiment.
\newblock {\em Political Analysis}, 14(1):37--62.

\bibitem[Billingsley, 1995]{billingsley1995probability}
Billingsley, P. (1995).
\newblock {\em Probability and Measure}.
\newblock Wiley Series in Probability and Statistics. Wiley.

\bibitem[Breiman et~al., 1984]{Breiman:decisionTree}
Breiman, L., Friedman, J.~H., Olshen, R.~A., and Stone, C.~J. (1984).
\newblock {\em {Classification and regression trees}}.
\newblock Wadsworth \& Brooks/Cole Advanced Books \& Software.

\bibitem[Buja et~al., 2019]{Buja2019}
Buja, A., Brown, L., Berk, R., George, E., Pitkin, E., Traskin, M., Zhang, K.,
  and Zhao, L. (2019).
\newblock Models as approximations {I}: Consequences illustrated with linear
  regression.
\newblock {\em Statistical Science}, 34:523--544.

\bibitem[Buja et~al., 2016]{Buja2016}
Buja, A., Brown, L., Berk, R., George, E., Pitkin, E., Traskin, M., Zhang, K.,
  Zhao, L., et~al. (2016).
\newblock Models as approximations {II}: A model-free theory of parametric
  regression.
\newblock {\em Statistical Science}, 34:545--565.

\bibitem[Casella, 1992]{CasellaConditional}
Casella, G. (1992).
\newblock Conditional inference from confidence sets.
\newblock {\em Lecture Notes-Monograph Series}, 17:1--12.

\bibitem[Chernozhukov et~al., 2018]{chernozhukvo2018debiased}
Chernozhukov, V., Chetverikov, D., Demirer, M., Duflo, E., Hansen, C., Newey,
  W., and Robins, J. (2018).
\newblock {Double/debiased machine learning for treatment and structural
  parameters}.
\newblock {\em The Econometrics Journal}, 21(1):C1--C68.

\bibitem[Cleveland, 1979]{cleveland1979robust}
Cleveland, W.~S. (1979).
\newblock Robust locally weighted regression and smoothing scatterplots.
\newblock {\em Journal of the American Statistical Association},
  74(368):829--836.

\bibitem[Cleveland and Devlin, 1988]{cleveland1988locally}
Cleveland, W.~S. and Devlin, S.~J. (1988).
\newblock Locally weighted regression: an approach to regression analysis by
  local fitting.
\newblock {\em Journal of the American Statistical Association},
  83(403):596--610.

\bibitem[Cox and Reid, 1987]{cox1987parameter}
Cox, D.~R. and Reid, N. (1987).
\newblock Parameter orthogonality and approximate conditional inference.
\newblock {\em Journal of the Royal Statistical Society: Series B},
  49(1):1--18.

\bibitem[Dedecker and Merlev{\`e}de, 2003]{dedecker2003conditional}
Dedecker, J. and Merlev{\`e}de, F. (2003).
\newblock The conditional central limit theorem in {H}ilbert spaces.
\newblock {\em Stochastic Processes and Their Applications}, 108(2):229--262.

\bibitem[Duchi, 2021]{mestimate}
Duchi, J. (2021).
\newblock Exercises for theory of statistics (stats300b).
\newblock https://web.stanford.edu/class/stats300b/Exercises/all-exercises.pdf.

\bibitem[Edgington and Onghena, 2007]{edgington2007randomization}
Edgington, E. and Onghena, P. (2007).
\newblock {\em Randomization tests}.
\newblock CRC press.

\bibitem[Egami and Hartman, 2021]{egami2021covariate}
Egami, N. and Hartman, E. (2021).
\newblock Covariate selection for generalizing experimental results:
  Application to a large-scale development program in uganda.
\newblock {\em Journal of the Royal Statistical Society: Series A (Statistics
  in Society)}.

\bibitem[Ernst et~al., 2004]{ernst2004permutation}
Ernst, M.~D. et~al. (2004).
\newblock Permutation methods: a basis for exact inference.
\newblock {\em Statistical Science}, 19(4):676--685.

\bibitem[Fahrmexr, 1990]{fahrmexr1990maximum}
Fahrmexr, L. (1990).
\newblock Maximum likelihood estimation in misspecified generalized linear
  models.
\newblock {\em Statistics}, 21(4):487--502.

\bibitem[Fisher, 1935a]{fisher1935statistical}
Fisher, K. (1935a).
\newblock Statistical tests.
\newblock {\em Nature}, 136(3438):474--474.

\bibitem[Fisher, 1935b]{fisher35inductive}
Fisher, R.~A. (1935b).
\newblock The logic of inductive inference.
\newblock {\em Journal of the Royal Statistical Society}, 98(1):39--82.

\bibitem[Fisher et~al., 1937]{fisher1937design}
Fisher, R.~A. et~al. (1937).
\newblock {\em The Design of Experiments.}
\newblock Oliver \& Boyd, Edinburgh \& London.

\bibitem[Freedman et~al., 2008]{freedman2008regression}
Freedman, D.~A. et~al. (2008).
\newblock On regression adjustments in experiments with several treatments.
\newblock {\em The Annals of Applied Statistics}, 2(1):176--196.

\bibitem[Goldberger, 1991]{goldberger1991course}
Goldberger, A.~S. (1991).
\newblock {\em A course in econometrics}.
\newblock Harvard University Press.

\bibitem[Green and Silverman, 1993]{green1993nonparametric}
Green, P.~J. and Silverman, B.~W. (1993).
\newblock {\em Nonparametric regression and generalized linear models: a
  roughness penalty approach}.
\newblock {CRC} Press.

\bibitem[Grzenda and Zieba, 2008]{grzenda2008conditional}
Grzenda, W. and Zieba, W. (2008).
\newblock Conditional central limit theorem.
\newblock In {\em Int. Math. Forum}, volume~3, pages 1521--1528.

\bibitem[Hinkelmann and Kempthorne, 1994]{hinkelmann1994design}
Hinkelmann, K. and Kempthorne, O. (1994).
\newblock {\em Design and analysis of experiments}.
\newblock Wiley Online Library.

\bibitem[Hinkley, 1980]{hinkley1980likelihood}
Hinkley, D.~V. (1980).
\newblock Likelihood.
\newblock {\em The Canadian Journal of Statistics / La Revue Canadienne de
  Statistique}, 8(2):151--163.

\bibitem[Ho, 1995]{ho1995random}
Ho, T.~K. (1995).
\newblock Random decision forests.
\newblock In {\em Proceedings of 3rd international conference on document
  analysis and recognition}, volume~1, pages 278--282. IEEE.

\bibitem[Imbens and Rubin, 2015]{Imbens2015}
Imbens, G.~W. and Rubin, D.~B. (2015).
\newblock {\em Causal Inference for Statistics, Social, and Biomedical
  Sciences: An Introduction}.
\newblock Cambridge University Press.

\bibitem[Kosorok and Laber, 2019]{kosorok2019precision}
Kosorok, M.~R. and Laber, E.~B. (2019).
\newblock Precision medicine.
\newblock {\em Annual review of statistics and its application}, 6:263.

\bibitem[Kuchibhotla et~al., 2018]{kuchibhotla2018model}
Kuchibhotla, A.~K., Brown, L.~D., and Buja, A. (2018).
\newblock Model-free study of ordinary least squares linear regression.
\newblock {\em arXiv preprint arXiv:1809.10538}.

\bibitem[Kuenzel, 2019]{kuenzel2019heterogeneous}
Kuenzel, S.~R. (2019).
\newblock {\em Heterogeneous Treatment Effect Estimation Using Machine
  Learning}.
\newblock PhD thesis, UC Berkeley.

\bibitem[Liu et~al., 2020]{liu2020doubly}
Liu, M., Zhang, Y., and Cai, T. (2020).
\newblock Doubly robust covariate shift regression with semi-nonparametric
  nuisance models.
\newblock {\em arXiv preprint arXiv:2010.02521}.

\bibitem[Manski, 1991]{manski1991}
Manski, C.~F. (1991).
\newblock Regression.
\newblock {\em Journal of Economic Literature}, 29(1):34--50.

\bibitem[Nadaraya, 1964]{nadaraya1964estimating}
Nadaraya, E.~A. (1964).
\newblock On estimating regression.
\newblock {\em Theory of Probability \& Its Applications}, 9(1):141--142.

\bibitem[Newey and McFadden, 1994]{newey1994large}
Newey, W.~K. and McFadden, D. (1994).
\newblock Large sample estimation and hypothesis testing.
\newblock {\em Handbook of econometrics}, 4:2111--2245.

\bibitem[Quinonero-Candela et~al., 2008]{quinonero2008dataset}
Quinonero-Candela, J., Sugiyama, M., Schwaighofer, A., and Lawrence, N.~D.
  (2008).
\newblock {\em Dataset shift in machine learning}.
\newblock Mit Press.

\bibitem[Robins et~al., 1994]{robins1994estimation}
Robins, J.~M., Rotnitzky, A., and Zhao, L.~P. (1994).
\newblock Estimation of regression coefficients when some regressors are not
  always observed.
\newblock {\em Journal of the American Statistical Association},
  89(427):846--866.

\bibitem[Rosenbaum, 2010]{rosenbaum2010design}
Rosenbaum, P.~R. (2010).
\newblock {\em Design of observational studies}.
\newblock Springer.

\bibitem[Rotnitzky et~al., 2012]{rotnitzky2012improved}
Rotnitzky, A., Lei, Q., Sued, M., and Robins, J.~M. (2012).
\newblock Improved double-robust estimation in missing data and causal
  inference models.
\newblock {\em Biometrika}, 99(2):439--456.

\bibitem[Splawa-Neyman et~al., 1990]{splawa1990application}
Splawa-Neyman, J., Dabrowska, D.~M., and Speed, T. (1990).
\newblock On the application of probability theory to agricultural experiments.
  {E}ssay on principles.
\newblock {\em Statistical Science}, pages 465--472.

\bibitem[Sugiyama et~al., 2012]{sugiyama2012density}
Sugiyama, M., Suzuki, T., and Kanamori, T. (2012).
\newblock {\em Density ratio estimation in machine learning}.
\newblock Cambridge University Press.

\bibitem[Tipton et~al., 2014]{tipton2014sample}
Tipton, E., Hedges, L., Vaden-Kiernan, M., Borman, G., Sullivan, K., and
  Caverly, S. (2014).
\newblock Sample selection in randomized experiments: A new method using
  propensity score stratified sampling.
\newblock {\em Journal of Research on Educational Effectiveness},
  7(1):114--135.

\bibitem[Tsiatis, 2007]{tsiatis2007semiparametric}
Tsiatis, A. (2007).
\newblock {\em Semiparametric theory and missing data}.
\newblock Springer.

\bibitem[van~der Vaart, 1998]{Vaart1998}
van~der Vaart, A.~W. (1998).
\newblock {\em Asymptotic Statistics}.
\newblock Cambridge University Press.

\bibitem[Watson, 1964]{watson1964smooth}
Watson, G.~S. (1964).
\newblock Smooth regression analysis.
\newblock {\em Sankhy{\=a}: The Indian Journal of Statistics, Series A}, pages
  359--372.

\bibitem[White, 1980]{white1980heteroskedasticity}
White, H. (1980).
\newblock A heteroskedasticity-consistent covariance matrix estimator and a
  direct test for heteroskedasticity.
\newblock {\em Econometrica}, pages 817--838.

\end{thebibliography}

\newpage 
\thispagestyle{empty}

\appendix

\begin{center}
{\Large\bf Supplementary material for ``Tailored inference for finite populations: conditional validity 
and transfer across distributions''} 
\end{center}

\medskip


\appendix

\section{Deferred discussion and results}
\subsection{Conditional versus marginal inference}\label{sec:super-vs-condit}

Continuing the  
example of a healthcare provider (for simplicity, let us say they are hospitals 
in a city) estimating 
health conditions discussed in the introduction, 
suppose there are $N=1000$ hospitals $j=1,\ldots,N$, 
each having $n=10000$ fixed patients  
with i.i.d.~attributes $Z_{ij} \sim \mathbb{P}_Z$, 
$i=1,\ldots,n$. 
For simplicity, 
we assume that  the fixed population is 
defined by fixing their attributes, so that  
the observations are $Y_{ij} =  f_j(Z_{ij}) + \epsilon_{ij}$,
where $\epsilon_{ij}\sim N(0,1)$ is i.i.d.~measurement noise and unexplained variation, 
and $f_j(z)$ is the
average health of a patient with attributes $Z=z$, 
which can vary with $j$. 
We also assume $f_j(Z_{ij})$ and 
$\epsilon_{ij}$ have finite second moments. 
In our simulation, the attributes $Z_{ij}$ are 
fixed at their observed values, 
while $\epsilon_{ij}$ 
are repeatedly drawn. 
For simplicity, we assume the marginal variance 
$\sigma_{\text{m}}=\text{sd}(Y_{ij})$ 
and measurement noise $\sigma_{\text{e}}=\text{sd}(\epsilon_{ij})$ 
are known. 

\begin{figure}[ht]
  \centering
    \includegraphics[width=1\linewidth]{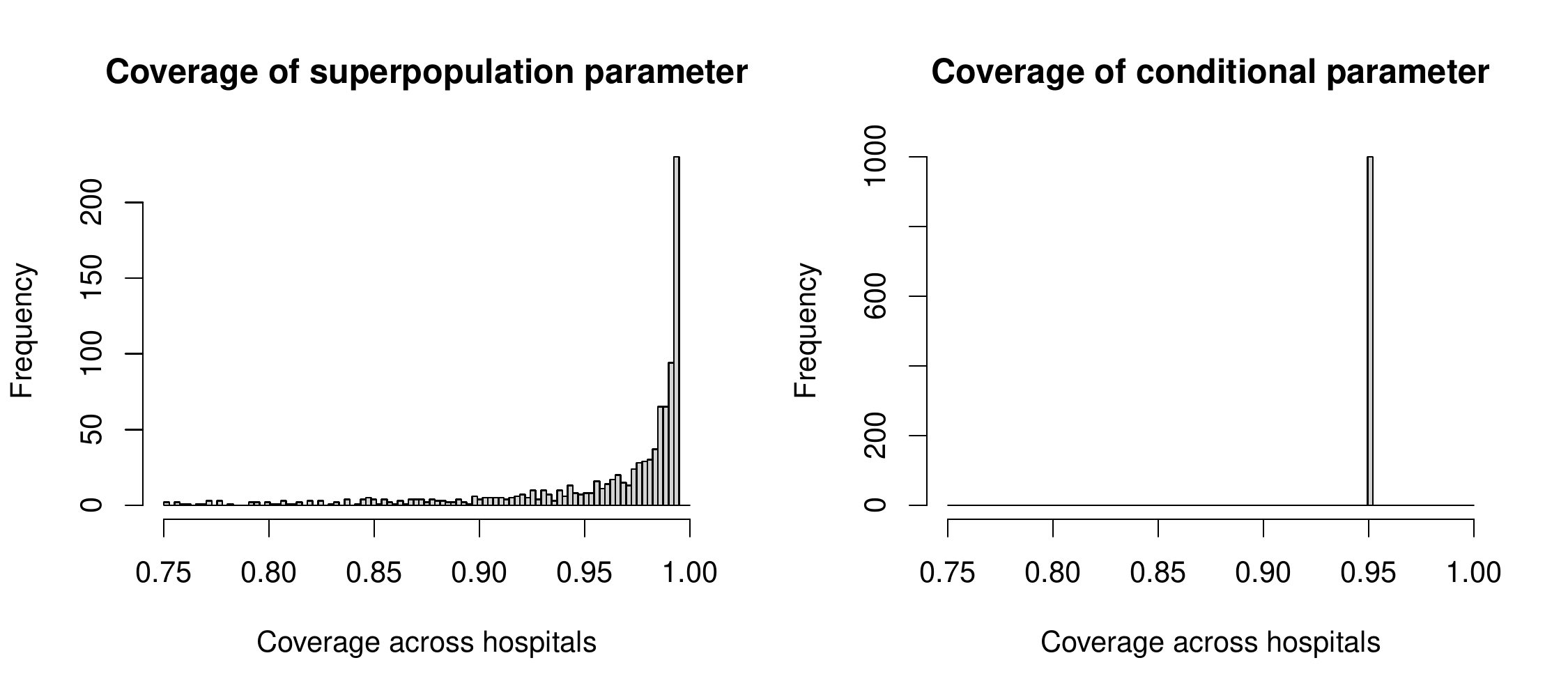}
    \caption{Left: coverage of super-population confidence intervals across hospitals ($37$ hospitals with coverage $< .75$ are not shown). Right: coverage of conditional confidence intervals across hospitals. In both cases, the marginal coverage is $.95$. Details of the simulation are in Section~\ref{sec:super-vs-condit}}
  \label{fig:super-vs-condit}
\end{figure}

In super-population inference,  
each hospital can construct  
$95 \%$ confidence intervals 
for the super-population parameter 
$  \EE(Y_{ij})$ via 
$\frac{1}{n} \sum_{i=1}^n Y_{ij}  \pm 1.96 \sigma_{\text{m}} /\sqrt{n}$. 
We show the histogram of coverage across the hospitals
in the left-hand side of Figure~\ref{fig:super-vs-condit}, 
where 
we observe \emph{under-coverage} for some hospitals. 
Indeed, $26\%$ of hospitals have coverage below $.95$, and the average coverage among these hospitals is only $.85$. 
Furthermore, if these hospitals would 
repeat similar examinations many times, 
their confidence intervals would 
consistently suffer from under-coverage.  
On the other hand, the confidence 
intervals of some other hospitals 
will consistently over-cover  if similar 
examinations are repeated many times.   
Such lack of \emph{conditional} coverage 
illustrates the risk  of 
super-population inference, 
which is especially pressing in high-stakes 
applications such as healthcare.

Mathematically, the issue is that for each hospital $j$, 
the customers defined by $\{Z_{ij}\}_{i=1}^n$ are fixed and 
only the remaining variation in $\{\epsilon_{ij}\}_{i=1}^n$ 
is drawn repeatedly.  
Super-population inference 
that accounts for 
the randomness of both $Z_{ij}$ and $\epsilon_{ij}$ 
is marginally valid (the coverage is .95 
averaged over the hospitals). 
In this situation, 
however,  
it would be more desirable 
to have coverage close to $.95$ for each hospital, 
i.e., 
conditional on  
$\{Z_{ij}\}_{i=1}^n$ for each  $j$. 

As discussed above, 
the data scientist might find conditional parameters 
more relevant. 
One can conduct inference for the conditional parameter $\frac{1}{n} \sum_{i=1}^n f_j(Z_{ij})$, 
the average intrinsic health risk of the fixed 
patients in hospital $j$, ruling out the measurement error. 
In our framework, 
$95 \%$ confidence intervals can be constructed via $\frac{1}{n} \sum_{i=1}^n Y_{ij}  \pm 1.96 \sigma_{\text{e}} /\sqrt{n}$. 
%
The histogram of coverage of these confidence intervals 
are shown on the right-hand side of 
Figure~\ref{fig:super-vs-condit}, 
where we observe coverage consistently close to $.95$ for all hospitals. 
By switching to conditional parameters, 
the confidence intervals are more relevant 
and more reliable for the 
patients in each hospital. 
We also note that conditional confidence intervals 
are shorter than those for super-population inference. 

We finally remark  
a few over-simplified aspects in 
this stylized example.   
Firstly,  
replacing 
$\sigma_{\text{e}}$  with a consistent or 
conservative estimator 
preserves 
conditional validity under mild conditions. 
Secondly, one may want to condition on unobserved variables. To be more precise,  
one may want to infer some health indicator $Y^* = \mathbb{E}(Y|Z^*)$ 
(where $Z^*$ is an unobserved variable that is finer than $Z$), 
we still allow for $Y^*$-conditionally valid 
(yet conservative) inference for 
the corresponding conditional parameter
$\frac{1}{n}\sum_{i=1}^n Y_{ij}^*$. 
More details are discussed in Section~\ref{sec:cond_inf_single}. 

\subsection{Connection to fixed-design regression and finite-population causal inference}
\label{app:subsec_causal}

\begin{example}[Linear regression]
For linear regression with misspecified model~\citep{Buja2016}, 
conditional parameters are defined by conditional 
ordinary least-squares (OLS).  
Assume $D=(X,Y)$ for a response $Y\in \RR$ 
and predictors $X \in \mathbb{R}^p$. 
The OLS parameter is 
$\theta_0 = \argmin_\theta E\{(Y-X ^\top \theta)^2\}$
with 
$s(D,\theta) = 2 X (Y-X^\top \theta)$. 
The conditional parameter is 
$
    \theta_n^{\cond} = \argmin_b \sum_{i=1}^n \EE\big\{ (Y_i- X_i^\top  b)^2 \biggiven Z_i\big\},
$ 
the least-square projection of $Y$ on $X$ when the observations are 
drawn conditional on $(Z_1,\dots,Z_n)$.  
If $Z_i = X_i$, then $\theta_n^\cond$ can be 
viewed as the regression coefficient 
for a set of subjects with fixed regressors, 
averaging over measurement noise. 
In model-based inference, 
if $Z_i=X_i$ and $Y_i=X_i^\top \theta_0 + \epsilon_i$ 
for $\EE(\epsilon_i\given X_i)=0$, i.e., well-specified model,
we would have $\theta_n^\cond = \theta_0$. 
In practice, however, this will usually not hold, and 
the conditional parameter 
may vary with the realization  of $X_i$. 
More generally, $Z$ can also be a variable 
outside  the set of predictors; 
conditioning on $Z$ can change the parameters 
if it is correlated with both the predictors and the residuals.  
\end{example}

\begin{example}[Finite-population causal inference]
\label{ex:causal} 
Finite-population treatment effects are 
a common target in causal inference.
In social sciences, for example, it is expected that individuals react differently to treatments. 
In this case, conditional inference can be used to understand 
the reaction of 
a specific population.  
Suppose $(T,X,Y(1),Y(0))$ are sampled from a super-population $\mathbb{P}$, 
where $T\in\{0,1\}$ is the treatment indicator, $X$ is the covariates, 
$Y(1),Y(0)$ are the potential outcomes 
if the treatment is received ($T=1$) and not ($T=0$). 
Under SUTVA and consistency~\citep{Imbens2015}, 
for each unit we observe $D=(T,X,Y)$,  where
$
    Y = T Y(1) + (1-T) Y(0).
$ 
The (super-population) average treatment effect $\theta_0 = \mathbb{E}\{Y(1)- Y(0)\}$ is the solution to~\eqref{eq:est}  
where
$
    s(D,\theta) = Y(1) - Y(0) - \theta.
$
There are many choices of conditioning variables $Z$. 
Conditioning on the (unobserved) potential outcomes $Z_i = (Y_i(1),Y_i(0))$, 
$
    \theta_n^{\cond} = \frac{1}{n} \sum_{i=1}^n  \{Y_i(1) - Y_i(0) \},
$ 
characterizes the population where 
potential outcomes of the subjects 
are fixed, a 
common  target in finite-population causal inference \citep{splawa1990application,hinkelmann1994design,freedman2008regression,rosenbaum2010design,Imbens2015}, where 
only the randomness in treatment assignment is accounted for. 
It can also be sensible to condition on covariates and average over measurement noise, leading to 
$
\theta_n^{\cond} = \frac{1}{n} \sum_{i=1}^n \EE \{ Y_i(1) - Y_i(0)\given X_i \} = \frac{1}{n}\sum_{i=1}^n \tau(X_i),
$ 
which is
the best prediction for the treatment effects 
of the population given $\{X_i\}_{i=1}^n$. 
Here $\tau(x)=\EE\{Y(1)-Y(0)\given X=x\}$ is the 
conditional average treatment effect (CATE) that indicates
treatment effect heterogeneity on the covariate level. 
Finally, conditioning on the empty set gives $\theta_0$
for the super-population.  
\end{example}

\subsection{Deferred simulation results}
\label{app:subsec_simu}

In our simulation studies in Section~\ref{sec:simu},
we compute the ratio of estimated standard deviation 
(i.e., that of confidence interval lengths) 
for conditional inference and super-population inference 
in Figure~\ref{fig:cond_len}. 
\begin{figure}[ht]
    \centering
    \includegraphics[width=4in]{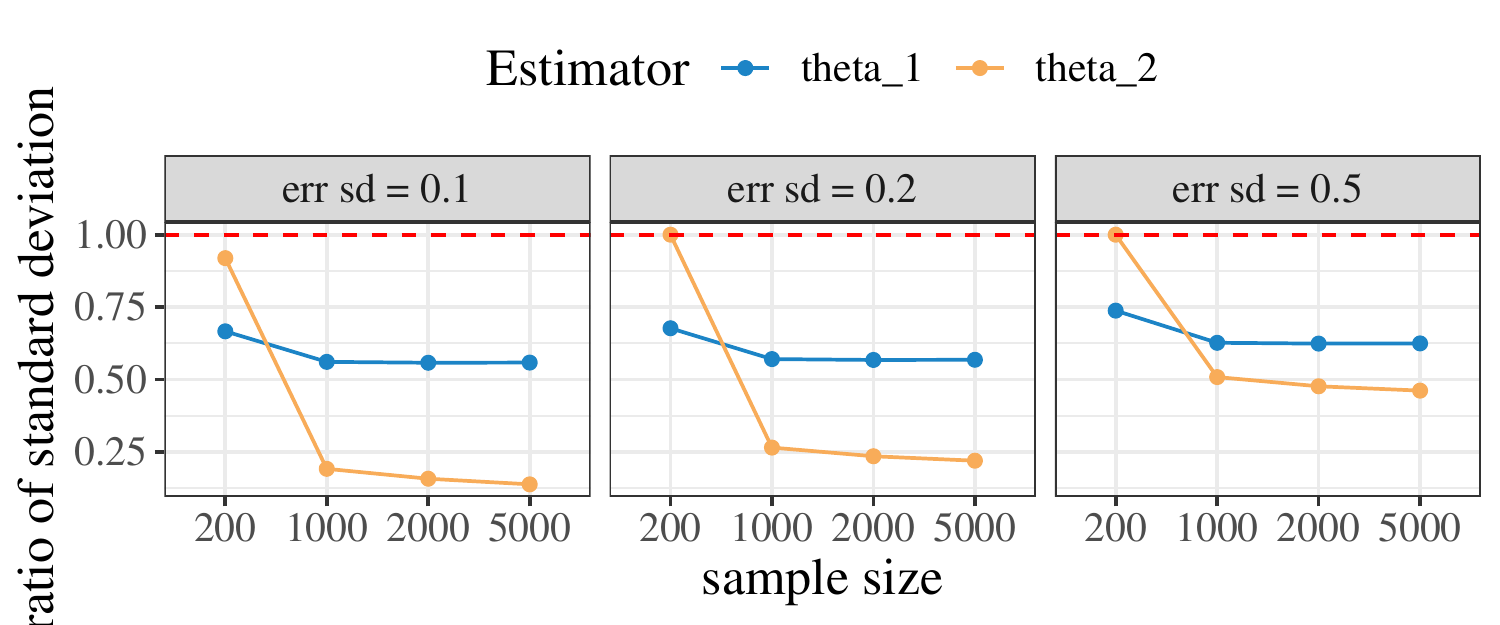}
    \caption{Ratio of estimated standard deviation
    of conditional v.s.~super-population inference. 
    Red dashed lines equal $1$.}
    \label{fig:cond_len}
\end{figure}

\section{Algorithms and convergence guarantees}

\subsection{Algorithms for inference procedures}\label{sec:algo} 
In this section, we describe concrete algorithms for 
estimating 
$\varphi(\cdot)$, $\eta(\cdot)$, $\sigma^2$ and $\sigma_{\shift}^2$. 
Corresponding theory can be found in Section~\ref{subsec:est_consist}. 
Similar to the main text, 
the estimation of variances is discussed for  
the one-dimensional parameters, 
while the arguments 
naturally carry over to the estimation of covariance matrix 
for multi-dimensional influence functions. 
Other quantities like conditional mean functions  
are discussed in the general case 
for multivariate covariates (attributes).

Note that 
the influence functions $\phi(\cdot)$, $\psi(\cdot)$ 
and their corresponding 
conditional mean functions $\varphi(\cdot)$, $\eta(\cdot)$ 
all admit the generic form
\#\label{eq:infl_generic}
f(d) = M(s,w,\theta) s(d,\theta),\quad 
g(z) = M(s,w,\theta) \EE\{s(D_i,\theta)\given Z_i=z\}
\#
for 
some weight function $w(\cdot)$, 
$\theta\in \RR^p$, 
score function $s\colon \mathbb{D}\times \Theta \to \RR^p$ and 
\#\label{eq:M_swtheta}
M(s,w,\theta) = -[\EE\{w(Z_i)\dot{s}(D_i,\theta)\}]^{-1}\in \RR^{p\times p},
\quad (D_i,Z_i)\sim \mathbb{P}.
\# 

Our general recipe is to estimate $M(s,w,\theta)$ 
and $\EE\{s(D_i,\theta)\given Z_i=z\}$ separately 
with plug-in nuisance components. 
We build our procedures 
upon the following two meta algorithms.

\begin{algorithm}
\caption{Meta Algorithm: Estimation of $\mathbb{E}[h(D)|Z=\cdot]$. }\label{alg:meta_cond_reg}
\begin{algorithmic}[1]
  \REQUIRE Function $h(\cdot)\colon \mathbb{D}\to \RR^p$, dataset $\{(D_i,Z_i)\}_{i\in \cI}$ independent of $h(\cdot)$. 
\ENSURE  function $\mathcal{G}(h,\cI)(\cdot)\colon \mathbb{Z}\to \RR^p$.
\end{algorithmic}
\end{algorithm}

\begin{algorithm}
  \caption{Meta Algorithm: Matrix Estimation.}\label{alg:meta_matrix}
\begin{algorithmic}[1]
  \REQUIRE  Score function $s\colon \mathbb{D}\times\Theta \to \RR^p$, weight function $w\colon \mathbb{Z}\to \RR$, $\theta\in \Theta$, data $\{(Z_i,D_i)\}_{i\in \cI}$. 
  \ENSURE  Matrix $\hat{M}(s,w,\theta,\cI) = -\big(\frac{1}{|\cI|}\sum_{i\in \cI} w(Z_i) \dot{s}(D_i,\theta)\big)^{-1} \in \RR^{p\times p}$.  
\end{algorithmic}
\end{algorithm}

\noindent\textit{Estimation for conditional inference
in Section~\ref{sec:cond_inf_single}.}
Recall that in Theorem~\ref{thm:cond_intv}, the only quantity  needed for constructing confidence intervals (besides $\hat \theta$) is
a consistent estimator $\hat\sigma^2$ 
for $\sigma^2=\Var((\phi(D)-\varphi(Z))$.
The estimation with data 
$\{(D_i,Z_i)\}_{i\in \cI}$
is detailed in Algorithm~\ref{alg:sigma_est}. 
Roughly speaking, we first obtain estimators for $\phi(D_i)$, $i\in\cI_2$; 
then  we estimate $\varphi(\cdot) = \mathbb{E}(\phi \given Z = \cdot)$   using only the data in one fold $\cI_1$ 
and apply to  another independent fold $\cI_2$, 
which are finally used to estimate $\sigma^2$. 
The sub-routine of 
estimating $\varphi(\cdot)$ 
is detailed in Algorithm~\ref{alg:varphi_est} 
below. 

\begin{algorithm}
  \caption{Estimate $\sigma^2$.}\label{alg:sigma_est}
\begin{algorithmic}[1]
  \REQUIRE  Dataset $\{(D_i,Z_i)\}_{i\in \cI}$, score function $s\colon \mathbb{D}\times\Theta \to \RR^p$. 
  \STATE  Split indices $\cI$ into equally-sized $\cI_1$ and $\cI_2$.   
  \STATE  Set $\hat{\theta}$ as solution to $\sum_{i\in\cI_1}s(D_i,\theta)=0$. 
  \hfill  \textcolor{gray}{\sl // Estimate $\phi(D_i)$ for $i\in\cI_2$} 
  \STATE  Obtain $\hat{M}:=\hat{M}(s,\mathbf{1},\hat\theta,\cI_2)$ using Algorithm~\ref{alg:meta_matrix}.  
  \STATE  Set $\hat\phi_i = \hat{M} s(D_i,\hat\theta)$ for all $i\in \cI_2$. 
\STATE   Obtain $\hat\varphi(\cdot) = \varphi(s,\cI_1)(\cdot)$ from Algorithm~\ref{alg:varphi_est}.
 \hfill   \textcolor{gray}{\sl // Estimate $\varphi(\cdot)$ with only $\cI_1$}  
  \STATE   Set $\hat\varphi_i = \hat\varphi (Z_i)$ for all $i\in \cI_2$. 
  \hfill \textcolor{gray}{\sl // Apply to $\cI_2$} 
\ENSURE $\hat\sigma^2 = \frac{1}{|\cI_2|}\sum_{i\in \cI_2} (\hat\phi_i - \hat\varphi_i)^2$. 
\end{algorithmic}
\end{algorithm}

\noindent\textit{Estimation for transductive inference in Sections~\ref{sec:trans_iid} and~\ref{sec:known_shift}.}
Transductive inference 
requires a consistent estimator 
for $\sigma_{\shift}^2$ defined in equation~\eqref{eq:def_shift_var}, 
and an estimator 
for $\eta(\cdot)$ only using one fold $\cI_k$. 
For preparation, we describe in Algorithm~\ref{alg:est_w} 
a generic method 
to estimate the covariate shift $w(\cdot)$ 
when it is unknown. 
It is not the only choice; 
there have been a rich literature 
on estimating density ratios, see, e.g.,~\cite{sugiyama2012density}
for a comprehensive review.

\begin{algorithm}
  \caption{Estimate $w(\cdot)$.}\label{alg:est_w}
  \begin{algorithmic}[1]
    \REQUIRE Datasets $\{ Z_i \}_{i\in \cI}$, $\{ Z_j^\new \}_{j\in\cI^\new}$. 
    \STATE 
    Pool $\cI,\cI^\new$ together, 
    and set $T_i=0$ for $i\in \cI$ and $T_j=1$ for $j\in \cI^\new$. 
    \STATE
    Estimate $\hat{e}(z) = \hat\PP(T=1 \given Z=z)$ using pooled data 
    by any regression or classification algorithm.  
  \ENSURE function $\hat{w}(\cdot) = \frac{\hat{e}(\cdot)}{1-\hat{e}(\cdot)}\frac{|\cI|}{|\cI^\new|}\colon \mathbb{Z}\to \RR$. 
  \end{algorithmic}
  \end{algorithm}

In Algorithm~\ref{alg:eta_est}, 
we describe in details the 
estimation of $\eta(\cdot)$ 
using any fold $\cI$ and $\cI^\new$. 
Note that when $w(\cdot)$ is known, 
it is directly used 
to construct $\hat\theta_{m,n}^\trans$ 
for
Theorem~\ref{thm:transfer}, 
so that  $\cI^\new$ is in fact not used. 
Otherwise, 
we set aside a part  of $\cI$ to estimate it 
and construct $\hat\theta_{m,n}^{\trans,\shift}$ 
in 
Theorem~\ref{thm:est_cov_shift}.  

\begin{algorithm}
  \caption{Estimate $\eta(\cdot)$.}\label{alg:eta_est}
  \begin{algorithmic}[1]
    \REQUIRE Datasets $\{(D_i,Z_i)\}_{i\in \cI}$, $\{(D_j^\new,Z_j^\new)\}_{j\in\cI^\new}$,  score $s\colon \mathbb{D}\times\Theta \to \RR^p$, weight $w\colon \mathbb{Z}\to \RR$. 
    \STATE  Split indices $\cI$ into equally-sized $\cI_1$, $\cI_2$ and $\cI_3$.    
    \IF{  $w$ is given}
    \STATE  Set $\hat{w}=w$;  
    \hfill \textcolor{gray}{\sl // Obtain weight function} 
    \ELSE 
    \STATE Estimate weight function $\hat{w}(\cdot)\colon \mathbb{Z} \to \RR$ with $\cI_1$ and $\cI^\new$. 
    \ENDIF
    \STATE  Set $\hat{\theta}$ as solution to $\sum_{i\in\cI_2}\hat{w}(Z_i)s(D_i,\theta)=0$.  
    \hfill \textcolor{gray}{\sl // Estimate $ \theta_0^\new $} 
    \STATE  Obtain $\hat{M} := \hat{M}(s,\hat{w},\hat\theta,\cI_3)$ using Algorithm~\ref{alg:meta_matrix}.  
    \hfill \textcolor{gray}{\sl // Estimate $M(s,w,\theta_0^\new)$}  
    \STATE 
    Set $\hat{s}(\cdot) = s(\cdot,\hat\theta)\colon \mathbb{D} \to \RR^p$.
    \hfill \textcolor{gray}{\sl // Estimate $\EE[s(D,\theta_0^\new)\given Z=\cdot]$} 
    \STATE Obtain
    $\hat{t}(\cdot) := \cG(\hat{s},\cI_3)(\cdot)\colon \mathbb{Z} \to \RR^p$ using Algorithm~\ref{alg:meta_cond_reg}.  
  \ENSURE function $\eta(s,w,\cI,\cI^\new )(\cdot)\colon \mathbb{Z} \to \RR^p$, 
  where $\hat\eta(z) = \hat{M}\hat{t}(z)$.
  \end{algorithmic}
  \end{algorithm}

The estimation of $\sigma_{\shift}^2=\Var(w(Z)(\phi(D)-\varphi(Z))$ 
is described in Algorithm~\ref{alg:sigma_shift_est}. 
After sample splitting, we first estimate $\psi(D_i)$ for $i\in \cI_3$; 
then we 
use only $\cI_1,\cI_2$ to estimate $\eta(\cdot)$
and apply to estimate $\eta(Z_i)$, $i\in \cI_3$, 
which are  used to estimate $\sigma_{\shift}^2$.

\begin{algorithm}
  \caption{Estimate $\sigma_{\shift}^2$.}\label{alg:sigma_shift_est}
\begin{algorithmic}[1]
  \REQUIRE Datasets $\{(D_i,Z_i)\}_{i\in \cI}$, $\{(D_j^\new,Z_j^\new)\}_{j\in\cI^\new}$,  score function $s\colon \cD\times\Theta \to \RR^p$, weight function $w\colon \cZ\to \RR$. 
  \STATE  Split indices $\cI$ into equally-sized $\cI_1$, $\cI_2$ and $\cI_3$.  
  \IF{  $w$ is given}
  \STATE set $\hat{w}=w$; 
  \hfill \textcolor{gray}{\sl // Obtain weight function} 
  \ELSE 
  \STATE Estimate weight function $\hat{w}(\cdot)\colon \cZ\to \RR$ with $\cI_1$ and $\cI^\new$.  
  \ENDIF
  \STATE  Set $\hat{\theta}$ as solution to $\sum_{i\in\cI_2}\hat{w}(Z_i)s(D_i,\theta)=0$.   
  \hfill \textcolor{gray}{\sl // Estimate $\psi(D_i)$ for $i\in \cI_3$} 
\STATE  Obtain $\hat{M}:=\hat{M}(s,\hat{w},\hat\theta, \cI_3)$ using Algorithm~\ref{alg:meta_matrix}. 
\STATE  Set $\hat\psi_i = \hat{M}  s(D_i,\hat\theta)$ for all $i\in  \cI_3$.  
\STATE  Obtain $\hat\eta = \eta(s,\hat{w},\cI_2,\emptyset)(\cdot)$ from Algorithm~\ref{alg:eta_est}. 
\hfill  \textcolor{gray}{\sl // Estimate $\eta(\cdot)$ using only $\hat{w}$ and $\cI_2$}  
   \STATE Set $\hat\eta_i = \hat\eta (Z_i)$ for all $i\in \cI_3$. 
  \hfill \textcolor{gray}{\sl // Apply to $\cI_3$} 
\ENSURE $\hat\sigma_{\shift}^2 = \frac{1}{| \cI_3|}\sum_{i\in  \cI_3} \hat{w}(Z_i)^2(\hat\psi_i - \hat\eta_i)^2$.
\end{algorithmic}
\end{algorithm}


We note that the estimation of $\varphi(\cdot)=- \big[\EE\{\dot{s}(D ,\theta_0)\} \big]^{-1} \EE\{s(D ,\theta_0)\given Z=\cdot\}$
with data $\{(D_i,Z_i)\}_{i\in \cI}$ 
could be viewed as a special case 
of Algorithm~\ref{alg:eta_est} by taking 
$w(z)\equiv 1$. Nevertheless, 
we include a stand-alone algorithm here 
for convenience of reference.  
For notational convenience, 
we define $\mathbf{1}(z)\equiv 1$. 

\begin{algorithm}
  \caption{Estimate $\varphi(\cdot)$.}\label{alg:varphi_est}
\begin{algorithmic}[1]
  \REQUIRE Dataset $\{(D_i,Z_i)\}_{i\in \cI}$, score function $s\colon \mathbb{D}\times\Theta \to \RR^p$. 
  \STATE  Split indices $\cI$ into equally-sized $\cI_1$ and $\cI_2$.    
  \STATE  Set $\hat{\theta}$ as solution to $\sum_{i\in\cI_1}s(D_i,\theta)=0$.   
  \STATE  Obtain $\hat{M} := \hat{M}(s,\mathbf{1},\hat\theta,\cI_2)$ using Algorithm~\ref{alg:meta_matrix}.    
  \STATE   
  Set $\hat{s}(\cdot) = s(\cdot,\hat\theta)\colon \mathbb{D}\to \RR^p$. 
  \STATE Obtain
  $\hat{t}(\cdot) := \cG(\hat{s},\cI_2)(\cdot)\colon \mathbb{Z}\to \RR^p$ using Algorithm~\ref{alg:meta_cond_reg}.  
\qquad 
\ENSURE Function $ \varphi(s,\cI)(\cdot)=\hat{M}\hat{t}(\cdot)\colon \mathbb{Z}\to \RR^p$. 
\end{algorithmic}
\end{algorithm}


\subsection{Estimation guarantees}
\label{subsec:est_consist}

In this section, 
we provide estimation guarantees 
for algorithms in Section~\ref{sec:algo} 
with explicit and detailed conditions.
We state conditions  
for parameters in $\RR^p$ for generality;  
the targets are still the variance estimation 
for one-dimensional parameters,  
which can be generalized to covariance matrix estimation 
under the same conditions we state.

We begin with generic assumptions on  
the meta 
Algorithms~\ref{alg:meta_cond_reg} and~\ref{alg:meta_matrix}. 
For any function $f(\cdot)$, we let $\cG(f)(z)=\EE[f(D)\given Z=z]$ 
be the conditional mean function, viewing $f$ as fixed; 
also, recall that $\cG(f,\cI)$ is the output 
of Algorithm~\ref{alg:meta_cond_reg} 
using data $\cI$. 
Also, recall that $\hat{M}(s,w,\theta,\cI)$ is the 
output of Algorithm~\ref{alg:meta_matrix} using data $\cI$ 
and $M(s,w,\theta)$ in~\eqref{eq:M_swtheta}  
is its estimation target.

\begin{assumption}\label{assump:meta_reg}
For any fixed input function $f$ and dataset $\cI$, 
the output of Algorithm~\ref{alg:meta_cond_reg} satisfies  that 
$\|\cG(f,\cI)(\cdot) - \cG(f)(\cdot)\|_{L_2(\mathbb{P})} =O_P\{\cR_{r}(|\cI|)\}$
for some rate function $\cR_{r}(\cdot)\colon \mathbb{N}\to \RR^+$. 
\end{assumption}

\begin{assumption}\label{assump:meta_matrix}
For any fixed input $w,\theta$ and dataset $\cI$,  
the output of Algorithm~\ref{alg:meta_matrix} satisfies  
$\|\hat{M}(s,w,\theta,\cI) - M(s,w,\theta)\|_{\infty} =O_P\{\cR_{m}(|\cI|)\}$
for some rate function $\cR_{m}(\cdot)\colon \mathbb{N}\to \RR^+$, 
where $\|\cdot\|_\infty$ is the entry-wise maximum. 
\end{assumption}


The above assumption on the convergence rate holds for a couple of nonparametric regression methods if the input $f(\cdot)$, viewed as a fixed function, is sufficiently smooth. 
For example, localized nonparametric methods like kernel regression \citep{nadaraya1964estimating,watson1964smooth}, local polynomial regression \citep{cleveland1979robust,cleveland1988locally}, smoothing spline \citep{green1993nonparametric} and modern machine learning methods including regression trees \citep{Breiman:decisionTree} and random forests \citep{ho1995random}, to name a few.

To show consistency of $\hat\sigma^2$ 
from Algorithm~\ref{alg:sigma_est}, 
we additionally assume the targets are stable. 
\begin{assumption}\label{assump:stable}
The matrix $M(s,w,\theta)$
satisfies that 
$\| M(s,\mathbf{1},\theta)-M(s,\mathbf{1},\theta')\|_{\infty} = O(\|\theta-\theta'\|_2)$
and 
$\|M(s,w,\theta)-M(s,w',\theta )\|_{\infty} = O(\|w(z)-w'(z)\|_{L_2(\mathbb{P})})$ 
for any weight functions $w,w'$ and 
any $\theta,\theta'\in \Theta$. 
Also, 
$\|s(\cdot,\theta)-s(\cdot,\theta)\|_{L_2(\mathbb{P})} = O(\|\theta-\theta'\|_2)$ 
for any $\theta,\theta'\in \Theta$.
\end{assumption}

We show that $\hat\sigma^2$, the output of Algorithm~\ref{alg:sigma_est},
is consistent  
if the two generic meta algorithms  
have diminishing estimation error 
and  the target functions are stable. 
The proof of Proposition~\ref{prop:consist_sigma} is in 
Appendix~\ref{app:est_sigma}. 

\begin{proposition}[Consistency of $\hat\sigma^2$]
\label{prop:consist_sigma}
Suppose Assumptions~\ref{assump:meta_reg},~\ref{assump:meta_matrix}
and~\ref{assump:stable} hold, and 
the regularity conditions of Proposition~\ref{prop:linear_exp} 
hold for $\hat\theta$ in Algorithm~\ref{alg:sigma_est}. 
Also assume  
$\cR_m(n)$, $\cR_r(n)\to 0$ as $n\to \infty$. 
Then the output of Algorithm~\ref{alg:sigma_est} 
satisfies $\hat\sigma^2 \to \sigma^2$ in probability as $|\cI|\to \infty$. 
\end{proposition}

In Theorem~\ref{thm:transfer}, 
we only need the $L_2$-consistency for 
the estimation of $\eta(\cdot)$, 
as well as a consistent estimator 
for $\sigma_\shift^2$. 
In Theorem~\ref{thm:est_cov_shift}, 
we further need the convergence rate for 
estimating $\eta(\cdot)$. 
We analyze $\hat\eta(\cdot)$, 
the output of Algorithm~\ref{alg:eta_est}, 
under 
generic rates of the meta algorithms as follows. 
The proof of Proposition~\ref{prop:consist_eta} 
is in Appendix~\ref{app:est_eta}. 
\begin{proposition}[Convergence rate of $\hat\eta$]
\label{prop:consist_eta}
Suppose Assumptions~\ref{assump:meta_reg},~\ref{assump:meta_matrix},~\ref{assump:stable}  
and the regularity conditions in Propositions~\ref{prop:main_new_lin_exp}
and~\ref{prop:linear_hat_theta_est} hold.  
Let $\cI$, $\cI^\new$ be any inputs of Algorithm~\ref{alg:eta_est}. 
If $w(\cdot)$ is known, the output of Algorithm~\ref{alg:eta_est} satisfies 
\#\label{eq:consist_eta_w}
\big\|\hat\eta(\cdot) - \eta(\cdot) \big\|_{L_2(\mathbb{P})}\leq
p\cdot O_P\big\{\cR_m(|\cI|) + \mathcal{R}_r(|\cI|) + |\cI|^{-1/2} \big\}.
\#
If $\hat{w}(\cdot)$ is estimated, assume 
$\sup_z|\hat{w}(z)-w(z)| = o_P(1)$ and 
the regularity conditions 
in Proposition~\ref{prop:linear_hat_theta_est} 
also hold for $\hat\theta$. 
Then the output of Algorithm~\ref{alg:eta_est} satisfies 
\$
\big\|\hat\eta(\cdot) - \eta(\cdot) \big\|_{L_2(\mathbb{P})}\leq
p\cdot O_P\big\{\|\hat{w}(\cdot)-w(\cdot )\|_{L_2(\mathbb{P})} + \cR_m(|\cI|) + \mathcal{R}_r(|\cI|) + |\cI|^{-1/2} \big\}.
\$
As a direct implication, 
Assumption~\ref{assump:cov_est_rate}
holds if
\$
\|\hat{w}(\cdot)-w(\cdot )\|_{L_2(\mathbb{P})} = O_P(n^{-1/4}) \quad 
\text{and} \quad \cR_m(n) + \mathcal{R}_r(n) = O_P(n^{-1/4}).
\$
\end{proposition}
Note that in Algorithm~\ref{alg:eta_est}, $\cI^\new$ 
is only possibly used to estimate $w(\cdot)$. 
Consequently, 
the convergence rate of $\hat\eta$ depends on $\cI^\new$ only
through $\|\hat{w}(\cdot)-w(\cdot)\|_{L_2(\mathbb{P})}$.

The output $\hat\sigma_{\shift}^2$ 
of Algorithm~\ref{alg:sigma_shift_est}
is analyzed as follows, 
whose proof is in 
Appendix~\ref{app:est_sigma_shift}.  

\begin{proposition}[Consistency of $\hat\sigma_{\shift}^2$]\label{prop:est_sigma_shift}
Let $\cI$, $\cI^\new$ be any inputs of Algorithm~\ref{alg:sigma_shift_est}. 
Suppose Assumptions~\ref{assump:meta_reg},~\ref{assump:meta_matrix} 
and~\ref{assump:stable} hold, and 
the regularity conditions in Proposition~\ref{prop:linear_hat_theta_est}
hold for $\hat\theta$ in Algorithm~\ref{alg:sigma_shift_est}. 
Assume
$\cR_m(n)\to 0$ and $\cR_r(n)\to 0$ as $n\to \infty$. 
If $\sup_{z}|\hat{w}(z)-w(z)|=o_P(1)$, $\sup_z|w(z)|<\infty$, 
then 
the output of Algorithm~\ref{alg:sigma_shift_est} obeys 
$\hat\sigma_{\shift}^2 \to \sigma_{\shift}^2$ in probability 
as $|\cI|\to \infty$. 
\end{proposition}

Finally, we note that $\varphi(\cdot)$ 
can be estimated with Algorithm~\ref{alg:eta_est} 
by setting $w(z)\equiv 1$. 
For completeness, we include 
the following consistency result for $\hat\varphi(\cdot)$, 
whose proof is in Appendix~\ref{app:est_varphi}. 

\begin{proposition}[Consistency of $\hat\varphi$]
\label{prop:consist_hat_varphi}
Suppose Assumptions~\ref{assump:meta_reg} and~\ref{assump:meta_matrix} hold, 
and the regularity conditions in Proposition~\ref{prop:linear_exp} 
hold for $\hat\theta$. 
Assume $M(s,\mathbf{1},\theta)-M(s,\mathbf{1},\theta')\|_{\infty} = O(\|\theta-\theta'\|_2)$ 
and $\|s(\cdot,\theta)-s(\cdot,\theta')\|_{L_2(\mathbb{P})} = O(\|\theta-\theta'\|_2)$ 
for any $\theta,\theta'\in \Theta$. 
Then the output 
of Algorithm~\ref{alg:eta_est}
with $w(z)\equiv 1$, denoted as 
$\hat\varphi(\cdot)$, 
satisfies 
\$
\big\| \hat\varphi(\cdot) - \varphi(\cdot) \big\|_{L_2(\mathbb{P})}\leq
p\cdot O_P\big\{\cR_m(|\cI|) + \mathcal{R}_r(|\cI|) + |\cI|^{-1/2} \big\}.
\$
\end{proposition}

\section{Extensions}

\subsection{Fixed-attributes results}
\label{app:fix}

In this section, we provide a set of 
results for fixed attributes, i.e., 
the attributes $\{z_i\}_{i=1}^n$ are fixed 
a priori and not drawn i.i.d.~from a super-population. 
These results generalize the counterparts in 
Section~\ref{sec:cond_inf_single}. 

We start by describing the setup. 
Let $\{z_i\}_{i=1}^n$ be the fixed attributes of $n$ 
units. We assume the observed dataset $\{D_i\}_{i=1}^n$ 
are mutually independent with 
$D_i\sim \mathbb{P}_{D\given Z=z_i}$. 
Our target is still the parameter 
that characterizes 
the (conditional) distribution of these fixed units. 
For ease of illustration, we 
use the language of M-estimators. 
Let $\ell\colon \mathbb{D}\times \Theta\to \RR$ be a loss function, 
and 
define the conditional parameter as 
\$
 \theta_n^\cond = \argmin_{\theta\in \Theta\subset \RR^{p}}\, L_n(\theta), \quad  {L}_n(\theta) := \frac{1}{n}\sum_{i=1}^n \EE\big\{\ell(D_i,\theta)\biggiven z_i\big\}.
\$
Then $\theta_n^\cond$ only depends on $\{z_i\}_{i=1}^n$ 
and is fixed. 
Parallel to Section~\ref{sec:cond_inf_single}, 
we assume access to an estimator 
\$
\hat\theta_n =  \argmin_{\theta\in \Theta\subset \RR^{p}}\, \hat{L}_n(\theta), \quad \hat{L}_n(\theta) := \frac{1}{n} \sum_{i=1}^n \ell(D_i,\theta).
\$
%
We first establish  the asymptotic linearity 
for the deviation of $\hat\theta_n$ from $\theta_n^\cond$, 
and a few conditions for the loss function 
and observations are needed. 
For simplicity, since $\{z_i\}_{i=1}^n$ 
are fixed, all probabilities and expectations 
are then implicitly conditional on $\{z_i\}_{i=1}^n$. 
The proof of Proposition~\ref{prop:linear_fixed}
is in Appendix~\ref{app:subsec_fix}. 
\begin{assumption}\label{assump:fixed}
(i) $\hat\theta_n$ and $\theta_n^\cond$ 
are both unique minimizers for their targets; 
(ii) $\ell(d,\cdot)$ is convex and twice 
continuously differentiable for every $d$; 
(iii) $\nabla \ell(d,\cdot)$ and 
$\nabla^2 \ell(d,\cdot)$ are $m_n(d)$-Lipschitz 
on $\Theta$ and 
$\frac{1}{n}\sum_{i=1}^n \EE[m_n(D_i)^2 ] \leq M$ for constant $M<\infty$; 
(iv) $\Var\{ \nabla \hat{L}_n(\theta_n^\cond) \} \succeq c_2 \mathbf{I}_{p\times p}$ 
and $\nabla^2   L_n(\theta_n^\cond) \succeq c_2 \mathbf{I}_{p\times p}$ for constant $c_2>0$; 
(v) $\frac{1}{n}\sum_{i=1}^n \EE\{\|\nabla \ell(D_i,\theta_n^\cond)\|^q \} <\infty$ for $q>2$. 
\end{assumption}

\begin{proposition}\label{prop:linear_fixed}
Suppose Assumption~\ref{assump:fixed} holds. 
Then $\hat\theta_n - \theta_n^\cond = o_P(1)$ 
and it holds that 
\$
\sqrt{n}(\hat\theta_n - \theta_n^\cond) 
= \frac{1}{\sqrt{n}}\sum_{i=1}^n - \big\{ \nabla^2 L_n(\theta_n^\cond)  \big\}^{-1} \nabla \ell(D_i,\theta_n^\cond) + o_P(1).
\$
\end{proposition}

The following theorem then establishes the 
asymptotics for $\theta_n^\cond$, 
whose proof is also in Appendix~\ref{app:subsec_fix}.  

\begin{theorem}\label{thm:fix}
Define $\Sigma_n^{1/2} = \{\nabla^2 L_n(\theta_n^\cond)\}^{-1}   \Var\{\sqrt{n} \nabla \hat{L}_n(\theta_n^\cond) \}^{1/2}$ and suppose Assumption~\ref{assump:fixed} holds. Then 
$
\Sigma_n^{-1/2} \sqrt{n}(\hat\theta_n - \theta_n^\cond) \stackrel{d}{\to } N(0, \mathbf{I}_{p\times p}) 
$
as $n\to \infty$. 
\end{theorem}

To form (conditional) confidence intervals, 
it remains to construct a consistent 
estimator for $\Sigma_n^{1/2}$. 
For simplicity, we show a concrete approach 
for $p=1$, while the multi-dimensional case follows 
similar ideas. 
Denote $\dot{\ell}(d,\theta) = \nabla_\theta \ell(d,\theta)$ 
and $\ddot{\ell}(d,\theta) = \nabla_{\theta}^2 \ell(d,\theta)$.
The asymptotics in Theorem~\ref{thm:fix} reduce to
\#\label{eq:fix_sigma}
\sigma_n := \Sigma_{n}^{1/2} = \bigg[  \frac{1}{n}\sum_{i=1}^n \EE\big\{\ddot{\ell}(D_i,\theta_n^\cond)\big\} \bigg]^{-1}
\bigg[\frac{1}{n}\sum_{i=1}^n \Var\big\{\dot{\ell}(D_i,\theta_n^\cond) \given z_i\big\} \bigg]^{1/2}.
\#
The following algorithm 
returns an estimator for $\sigma_n$ 
by running a nonparametric regression on $\{z_i\}_{i=1}^n$. 

\begin{algorithm}
  \caption{Estimate $\sigma_n$.}\label{alg:sigma_fix}
  \begin{algorithmic}[1]
    \REQUIRE Dataset $\{(D_i,z_i)\}_{i=1}^n$, loss function $\ell \colon \mathbb{D}\times\Theta \to \RR $. 
    \STATE Set $\hat{\theta}_n$ as solution to $\sum_{i=1}^n\dot\ell(D_i,\theta)=0$.    
  \STATE Compute $\hat M = \frac{1}{n}\sum_{i=1}^n \ddot{\ell}(D_i,\hat\theta_n)$. 
  \STATE Obtain
    $\hat{t}(\cdot) := \cG(\hat{s},\cI )(\cdot)\colon \mathbb{Z}\to \RR^p$ using Algorithm~\ref{alg:meta_cond_reg} for $\hat{s}(\cdot) = \dot{\ell}(\cdot,\hat\theta_n)\colon \mathbb{D}\to \RR $.  
  \ENSURE estimator $\hat\sigma_n = \hat{M}^{-1} 
  \frac{1}{n}\sum_{i=1}^n [\hat{s}(D_i)  - \hat{t}(z_i) ]^2 $. 
  \end{algorithmic}
  \end{algorithm} 

We establish
the consistency of Algorithm~\ref{alg:sigma_fix} 
under mild conditions, and its robustness 
if such conditions fail. 
The proof of Proposition~\ref{prop:sigma_consist} 
is in Appendix~\ref{app:subsec_fix}. 

\begin{proposition}\label{prop:sigma_consist}
Suppose Assumption~\ref{assump:fixed} holds, 
and $\frac{1}{n}\sum_{i=1}^n \{\hat{t}(z_i)-\mu(z_i)\}^2 = o_P(1)$ for some fixed function $\mu\colon \mathbb{Z}\to \RR$. 
Also, suppose $\EE\{\dot{\ell}(D_i,\theta_n^\cond)\given z_i\}=\mu^*(z_i)$ for some function $\mu^*\colon \mathbb{Z}\to \RR$. Then the output of Algorithm~\ref{alg:sigma_fix} 
satisfies that (i) $\hat\sigma_n - \sigma_n = o_P(1)$ 
when $\mu^* = \mu$, and (ii) otherwise, 
$\hat\sigma_n - \tilde\sigma_n = o_P(1)$
for some $\tilde\sigma_n \geq \sigma_n$ 
which also only depends on $\{z_i\}_{i=1}^n$. 
\end{proposition}

In words, the above proposition shows the 
consistency of $\hat\sigma_n$ if 
the nonparametric regression of $\dot{\ell}(D_i,\hat\theta_n)$ 
on $z_i$ is consistent for the truth 
$\mu^*(z_i)=\EE\{\dot{\ell}(D_i,\theta_n^\cond)\given z_i\}$. 
If the regression is instead run with 
$\dot{\ell}(D_i, \theta_n^\cond)$, 
the well-established theory of nonparametric regression 
such as kernel regression or smoothing splines 
guarantees the diminishing $L_2$ error 
if the underlying function $\mu^*$ is well-behaved. 
Considering the order $O(n^{-1/2})$ deviation 
of $\dot{\ell}(D_i,\hat\theta_n)$ 
from $\dot{\ell}(D_i, \theta_n^\cond)$, 
the consistency requirement of Proposition~\ref{prop:consist_sigma} is mild. 
In addition, even though the regression is not consistent 
but converges to some deterministic function (an even 
more mild condition), 
Algorithm~\ref{alg:sigma_fix} returns an upper bound 
for $\sigma_n$, which would lead to a conservative 
yet valid confidence interval. 

Finally, as we mentioned in Remark~\ref{rmk:fix}, 
one could also use $\sigma_n^2$ derived here 
instead of $\sigma^2$ for the i.i.d.~setting~\eqref{eq:def_asymp_var}, 
as the regularity conditions above hold with high probability 
for i.i.d.~drawn attributes. 
Compared to $\sigma^2$, 
$\sigma_n^2$ might provide attribute-dependent 
characterization for the statistical uncertainty. 

However, we do note 
the similarity between $\sigma^2$ 
and $\sigma_n^2$ for i.i.d.~attributes. 
In this case, we have 
$\sigma_n^2=\sigma^2 + O_P(1/\sqrt{n})$; 
their contributions 
to constructing the confidence interval differ 
by a magnitude  
that is of the same order as, hence indistinguishable 
from, the error in
the asymptotic linearity (Assumption~\ref{assump:linear_main}) 
we rely on. 
Furthermore, in practice, the estimation of 
$\sigma_n^2$ in Algorithm~\ref{alg:sigma_fix} 
and of $\sigma^2$ in Algorithm~\ref{alg:sigma_est} 
does not make much difference.  
The main distinction is that 
the i.i.d.~assumption allows for sample splitting 
in the estimation of $\sigma^2$, 
which simplifies theoretical analysis and  
our theoretical guarantee only relies on generic 
consistency conditions of nonparameteric regression. 
Instead, Algorithm~\ref{alg:sigma_fix} applies 
all the data in 
every step; hence, it needs slightly stronger 
and less generic conditions on the regression outputs. 
The practical choice 
between $\sigma_n^2$ or $\sigma^2$ 
could also be viewed as a tradeoff 
between confidence in the i.i.d.\ assumption
and confidence in regression accuracy. 

\subsection{Conditioning on unobserved variables}
\label{app:cond_unobs}

In this part, we generalize the conditional inference 
framework to situations where some unobserved variables 
are fixed. 
Again, we assume $\{(D_i,Z_i)\}_{i=1}^n$ are i.i.d.~from 
a super-population $\mathbb{P}$. 
Suppose a data scientist would like to view some 
unobserved variable $X_{1:n} = \{X_i\}_{i=1}^n$ as fixed, 
which are also i.i.d.~from the super-population 
and then conditioned on. 
While we could also relax the i.i.d.~assumption 
to fixed-attributes settings, such extension 
follows similar ideas as Appendix~\ref{app:fix} hence 
we omit here for brevity. 

Following 
Section~\ref{subsec:cond_para}, the 
conditional parameter $\theta_n^\cond = \theta_n^\cond(X_{1:n})$ 
is the unique solution to 
\$
\sum_{i=1}^n \EE\big\{ s(D_i,\theta)\biggiven X_i   \big\} = 0.
\$
As we only observe attributes $\{Z_i\}_{i=1}^n$, 
we could use the procedures proposed in Sections~\ref{sec:cond_inf_single} and~\ref{sec:algo} 
to obtain an estimator $\hat\sigma_Z^2$ 
for the (observed) asymptotic variance 
$\sigma_Z^2 = \EE\{(\phi(D)-\varphi(Z))^2\}$. 
The following theorem states the 
conditional validity of inference based on $\hat\sigma_Z^2$, 
whose proof is in Appendix~\ref{app:subsec_cond_unobs}. 

\begin{theorem}\label{thm:cond_unobs}
Suppose $\hat\sigma_Z^2\stackrel{P}{\to}\sigma_Z^2$, 
and Assumptions~\ref{assump:linear_main} 
and ~\ref{assump:moment_main} 
hold with $Z$ replaced by $X$. 
If $\EE([\EE\{\phi(D)\given X\}]^2)\geq \EE([\EE\{\phi(D)\given Z\}]^2)$, 
then for any $\alpha\in(0,1)$, it holds that 
\$ 
    \PP \Big( \theta_n^{\cond}(X_{1:n}) \in \big[\,\hat \theta_n -  z_{1-\alpha/2}  \hat \sigma_Z/\sqrt{n},~ \hat \theta_n + z_{1-\alpha/2}  \hat \sigma_Z/\sqrt{n}\,\big] \Biggiven X_{1:n}\Big)
\$
converges in probability to $1-\beta$ for some fixed $\beta\leq \alpha$ as $n\to \infty$. 
\end{theorem}

The only additional requirement for conditionally valid 
(and perhaps conservative) inference given unobserved attributes 
is that $\EE([\EE\{\phi(D)\given X\}]^2)\geq \EE([\EE\{\phi(D)\given Z\}]^2)$. Roughly speaking, it requires 
the covariates $Z$ to explain away less variation in $\phi(D)$ 
than the unobserved variables $X$. 
For instance, 
one might want to condition on more information 
than the observed, and view the observed attributes $Z$
as partially defining the fixed population. 
In this situation (which might be the only case 
where one would like to condition on unobserved attributes), 
this condition is naturally satisfied. 
In fact, this situation is related  
Example~\ref{ex:causal} when we condition on 
the partially observed $\{ Y_i(0),Y_i(1) \}_{i=1}^n$ 
for finite-population treatment effect 
$\frac{1}{n}\sum_{i=1}^n \{Y_i(1)-Y_i(0)\}$; 
methods in the literature often proceed with 
conservative estimators for the asymptotic variance, 
which is similar to our setting here.

\subsection{Transferring a subset of observed attributes}
\label{app:subsec_trans_subset}
 
 The proposed transductive inference procedures 
generalize  to settings where 
the conditioning set is smaller than 
that for the 
covariate shift. 
That is, the covariate shift holds 
for the 
whole set $Z$ of observed attributes,
while one might think a subset $X$ should 
be viewed as fixed to characterize the new population. 
We discuss the extension of our framework to 
this setting in this part. 

Formally, we assume $X\subset Z$ where 
$Z$ is observable, 
and there is a (possibly unknown) 
covariate shift 
$d \mathbb{Q} / d \mathbb{P}(d,z) = w(z)$. 
As usual, we denote $X_{1:m}^\new=\{X_j^\new\}_{j=1}^m$ 
as the new conditioning attributes. 
The target is the new conditional parameter 
$\theta_m^\cond(X_{1:m}^\new)$ with respect 
to a subset of observed attributes. 
This setting is challenging because 
the invariance of conditional distribution 
does not necessarily hold for $X$. 
However, as we will see, 
the proposed estimator 
still allows for $X_{1:m}^\new$-conditionally 
valid transductive inference. 

Let $\hat\theta_{m,n}^\trans$ be 
the estimator defined in~\eqref{eq:trans_new_simple} 
of the main text, which is obtained from 
$Z_{1:n}\cup Z_{1:m}^\new$. 
The following theorem shows that 
slightly modifying the asymptotic variance 
leads to valid inference; 
for completeness, we discuss both the simplified 
exposition in the main text and 
the cross-fitted procedures in Appendix~\ref{app:subsec_cross_fitting_known}  and~\ref{app:subsec_cross_fitting_est}. 
The proof is in Appendix~\ref{app:proof_trans_subset}. 

\begin{theorem}
\label{thm:trans_subset}
Under the setup of Section~\ref{sec:known_shift}, 
suppose all conditions in Theorem~\ref{thm:est_cov_shift_simple} hold; 
under the setup of Appendix~\ref{app:subsec_cross_fitting_known}, 
suppose the conditions in 
Theorem~\ref{thm:transfer} hold; 
under the setup of 
Appendix~\ref{app:subsec_cross_fitting_known}, 
suppose the conditions in 
Theorem~\ref{thm:est_cov_shift} hold. 
Let $\hat\theta_{m,n}^{\trans}$ 
be the estimator built with 
$Z_{1:n}$ and $Z_{1:m}^\new$ in any of these cases. 
Suppose $(\hat\sigma_{\shift}')^2$ 
converges in probability to 
\$
(\sigma_{\shift}')^2 = \Var\big[ w(Z_i) \big\{\psi(D_i) - \eta(Z_i)\big\}   \big] 
+ n/m \cdot \Var\big[  \eta(Z_j^\new) - \EE\{\eta(Z_j^\new)\given X_j^\new\}   \big].
\$
Then for any fixed $\alpha\in(0,1)$, it holds that 
\$
\PP\Big\{ \theta_m^{\cond}(X_{1:m}^\new) \in \big[\,\hat\theta_{m,n}^{\trans} -  z_{1-\alpha/2} \hat\sigma_{\shift}'/\sqrt{n},~ \hat\theta_{m,n}^{\trans} + z_{1-\alpha/2}  \hat\sigma_{\shift}'/\sqrt{n}\,\big] \Biggiven X_{1:m}^\new \Big\}
\$
converges in probability to $1-\alpha$ 
as $m,n\to \infty$. 
\end{theorem}

\section{Details of cross-fitting for 
transductive inference}

\subsection{Details of cross-fitting for 
transductive inference with known covariate shift} 
\label{app:subsec_cross_fitting_known} 

This section contains details 
of cross-fitting for transductive inference 
under \emph{known} covariate shift $w(\cdot)$ 
that we omit for clarity in Section~\ref{sec:known_shift}. 
Instead of referring to external datasets, 
we split the data to decouple 
the estimation of nuisance components, 
and reuse the folds to achieve the same 
statistical efficiency.

We first split 
the index set $\cI=\{1,\dots,n\}$ of 
$\{(D_i,Z_i)\}_{i=1}^n$ into equally-sized halves $\cI_1$ and $\cI_2$. 
Then for $k=1,2$, 
we obtain estimator $\hat\eta^{\cI_{k}}(\cdot)$ for $\eta (\cdot)= \EE\{\psi(D_j^\new)\given Z_j^\new = \cdot\}$, 
using only the data in $\cI_k$.
\footnote{One can set $\hat\eta^{\cI_k}$ as the output of Algorithm~\ref{alg:eta_est} 
(c.f.~Section~\ref{sec:algo})
with inputs $w$ and $\cI_k$. 
Since in Algorithm~\ref{alg:eta_est}, 
the new attributes are only 
used to estimate the weight function (if it is unknown), 
here we
do not need them as input 
for estimating $\eta(\cdot)$.}  
We then define the estimator 
\#\label{eq:trans_new}
\hat\theta_{m,n}^{\trans } = \hat\theta_n^\trans - \hat{c}^\trans , 
\#
where $\hat\theta_n^\trans$ 
is the unique solution to 
\$
\sum_{i=1}^n w(Z_i) s(D_i,\theta)=0,
\$
i.e., setting $\hat{w}(\cdot)=w(\cdot)$ in~\eqref{eq:hat_theta_dag}. 
The correction term is defined as 
\#\label{eq:cor_new}
  \hat c^\trans  :=  \frac{1}{2|\cI_1|} \sum_{i \in \cI_1} \hat \eta^{\cI_2}(Z_i) w(Z_i)  + \frac{1}{2|\cI_2|} \sum_{i \in \cI_2}  \hat \eta^{\cI_1}(Z_i) w(Z_i)  - \frac{1}{2m} \sum_{k=1}^2\sum_{j=1}^m   \hat\eta^{\cI_k}(Z_j^\new) . 
\# 
We construct a confidence interval centered around $\hat\theta_{m,n}^\trans$ 
in Theorem~\ref{thm:transfer}. 
We assume
the $L_2(\mathbb{Q})$ consistency of $\hat\eta^{\cI_k}$;  
note that 
similar to the i.i.d.~setting in Section~\ref{sec:trans_iid}, 
we do not require any convergence rates of $\hat\eta^{\cI_k}$. 

\begin{assumption}\label{assump:cov_shift}
$\|[\hat\eta^{\cI_k}(\cdot) - \eta (\cdot)]w(\cdot)\|_{L_2(\mathbb{P})}$ 
and $\|\hat\eta^{\cI_k}(\cdot) - \eta (\cdot) \|_{L_2(\mathbb{Q})}$ converges 
in probability to zero for $k=1,2$. 
\end{assumption}
 
The following theorem states the asymptotic 
conditional validity of our cross-fitting procedure, 
whose  
proof is deferred to Section~\ref{app:thm_transfer} 
in this supplementary material.

\begin{theorem}
\label{thm:transfer}
Suppose 
Assumption~\ref{assump:linear_expansion_known_shift} 
in the main text holds for $\hat{w}=w$,  
Assumption~\ref{assump:cov_shift} holds, and $m\geq \epsilon n$ for some 
constant $\epsilon >0$. If an estimator $\hat\sigma_{\shift}$ converges in probability to $\sigma_{\shift}>0$ for 
$\sigma^2$ defined in~\eqref{eq:def_shift_var}. 
Then the random variable 
\$
\PP\Big( \theta_m^{\cond,\new} \in \big[ \hat\theta_{m,n}^{\trans} - \hat\sigma_{\shift} \cdot  z_{1-\alpha/2}/\sqrt{n},    \hat\theta_{m,n}^{\trans} + \hat\sigma_{\shift} \cdot z_{1-\alpha/2}/\sqrt{n}    \big] \Biggiven Z_{1:m}^\new, Z_{1:n} \Big) 
\$ 
converges in probability to $1-\alpha$ as $n\to \infty$, 
where $\hat\theta_{m,n}^{\trans}$ is defined in equation~\eqref{eq:trans_new}. 
\end{theorem}

In Theorem~\ref{thm:transfer}, 
the asymptotic linearity with $\hat{w}=w$ 
in the main text has been justified 
in Proposition~\ref{prop:main_new_lin_exp}. 
Similar to Theorem~\ref{thm:iid_simple}, 
the asymptotic variance does not depend on $m$, 
and the result only depends on the $L_2$-consistency of 
$\hat\eta^{\cI_k}$. 
Finally, one could estimate $\hat\sigma_{\shift}^2$ 
using Algorithm~\ref{alg:sigma_shift_est} without 
referring to external datasets. 


\subsection{Details of cross-fitting for 
transductive inference with estimated covariate shift}
\label{app:subsec_cross_fitting_est} 

In this section, we provide 
details for cross-fitting 
in transductive inference when the covariate shift 
$w(\cdot)$ in Section~\ref{sec:known_shift} is unknown. 
Procedures here do not refer to any external datasets.

We employ cross-fitting~\citep{chernozhukvo2018debiased} 
to decouple the estimation of $w(\cdot)$ and other quantities. 
The index set $\cI=\{1,\dots,n\}$ 
of the original dataset $\{(D_i,Z_i)\}_{i=1}^n$ 
is
randomly split into 
three equally-sized folds, denoted as 
$\cI_1$, $\cI_2$ and $\cI_3$. 
The index set 
$\cI^\new=\{1,\dots,m\}$ 
of the new dataset $\cZ_m^\new = \{Z_j^\new\}_{j=1}^m$ is 
randomly split into three equally-sized folds 
$\cI_1^\new$, $\cI_2^\new$ and $\cI_3^\new$. 
We then carry out a three-fold estimation: 
for each $\ell=1,2,3$, 
we first use $\cI_\ell$ and $\cI_\ell^\new$ 
to obtain 
an estimator $\hat{w}_\ell(\cdot)$ of the covariate shift. 
(We give an example in~Algorithm~\ref{alg:est_w} for 
estimating $w(\cdot)$ using any fold $\cI$
of original data and any fold $\cI^\new$ of new covariates.)
Then we use all remaining data $\cI \backslash \cI_\ell$
to obtain $\hat\theta_n^{\new,(\ell)}$, which is a unique solution to 
\#\label{eq:def_new_hat_theta_1}
\sum_{i\notin \cI_\ell} \hat{w}_\ell(Z_i) s(D_i, \theta) =0 .
\#
Next, 
for each $k\neq \ell$, 
we obtain an estimator $\hat\eta^{\cI_k}(\cdot)$ for 
$\eta(\cdot)$ using only $\cI_k$ and $\cI_k^\new$.
(To be specific, $\hat\eta^{\cI_k}(\cdot)$ 
is the output $\eta(s,\emptyset,\cI_k,\cI_k^\new )(\cdot)$ from Algorithm~\ref{alg:eta_est} that only depends on $\cI_k$ and $\cI_k^\new$.)
We define the $\ell$-th correction term as 
\$
\hat{c}^{(\ell)} &= \sum_{k\neq \ell}\frac{3}{2n} \sum_{i\notin \cI_\ell\cup \cI_{k}} \hat{w}_\ell(Z_i) \hat\eta^{\cI_k}(Z_i) 
- \sum_{k\neq \ell}  \frac{3}{2m}\sum_{j\notin \cI_\ell^\new\cup \cI_k^\new}  \hat\eta^{\cI_k}(Z_j^\new).
\$ 
We note that the high-level idea of 
the above correction term is similar to 
our simplified expression in Section~\ref{sec:known_shift}; 
the only difference is that 
we carefully split and reuse the data 
to achieve good statistical property. 
Finally, we define the transductive estimator as 
\#\label{eq:new_trans_est}
\hat\theta_{m,n}^{\trans,\shift} = \frac{1}{3}\sum_{\ell=1}^3 \big( \hat\theta_n^{\new,(\ell)} - \hat{c}^{(\ell)}\big).
\#

Without loss of generality, we assume $n_0= n/3$, $m_0=m/3$ 
are integers, so that 
the split folds are of exactly the same size; 
otherwise the induced bias is of a negligible order $O(1/m+1/n)$. 
Similar to Assumption~\ref{assump:linear_expansion_known_shift} in the main text,
we assume 
consistency of $\hat{w}_\ell$ as follows. 

\begin{assumption}\label{assump:consist_w}
For $\ell=1,2,3$, 
$\sup_{z}|\hat{w}_\ell(z)-w(z)|\to 0$ in probability as $n\to \infty$. 
\end{assumption}

For ease of exposition, 
we impose the linear expansion of $\hat\theta_n^{\new,(\ell)}$ 
as follows. 
In Proposition~\ref{prop:linear_hat_theta_est} in Section~\ref{app:subsec_linear_est} of this
supplementary material, we show that 
Assumption~\ref{assump:linear_est_shift} holds under 
Assumption~\ref{assump:consist_w} 
and regularity conditions similar to previous cases. 

\begin{assumption}\label{assump:linear_est_shift}
For $\ell=1,2,3$, 
$\hat\theta_n^{\new,(\ell)}$ is 
the unique solution to~\eqref{eq:def_new_hat_theta_1}. 
Also, letting 
$\theta_0^\new$ be the unique solution to~\eqref{eq:theta_dag}, 
assume the following asymptotic linearity holds:  
\#\label{eq:linear_hat_theta_est}
\sqrt{2n/3}(\hat\theta_n^{\new,(\ell)} - \theta_0^\new)
&= \frac{1}{\sqrt{2n/3}}\sum_{i \notin \cI_\ell}  \hat{w}_{\ell}(Z_i) \psi(D_i) + o_P(1), \\ 
\sqrt{m}(\theta_m^{\cond,\new} - \theta_0^\new )&= \frac{1}{\sqrt{m}}\sum_{j=1}^m \eta(Z_j^\new)
 + o_P(1), \notag 
\#
where 
$
\psi(d) = - \big[\EE_Q\{\dot{s}(D^\new, \theta_0^\new)\}\big]^{-1} s(d, \theta_0^\new)
$, 
and $\eta(z) =\EE\{\psi(D_j^\new)\given Z_j^\new=z\}$. 
\end{assumption}

In the linear expansion~\eqref{eq:linear_hat_theta_est}, $\sqrt{2n/3}$ is due to sample splitting 
where $\hat\theta_n^{\new,(\ell)}$ only 
uses a fold of cardinality $2n/3$; 
we still obtain $\sqrt{n}$ order 
for inference 
by reusing all folds.

We also assume (slow) 
convergence rates of the estimated covariate shift 
and influence functions. 
Detailed conditions for it to hold can be found in 
the analysis of our estimation procedures, 
see Proposition~\ref{prop:consist_eta} 
of Section~\ref{subsec:est_consist}.

\begin{assumption}\label{assump:cov_est_rate} 
$\|w(\cdot)\{\hat\eta^{\cI_k}(\cdot) - \eta(\cdot)\}\|_{L_2(\mathbb{P})}\to 0$  
in probability, 
$\EE_{\mathbb{P}}\{w(Z_i)^4\psi(D_i)^4\}<\infty$ 
and 
$
\big\|\hat{w}_\ell(\cdot) - w(\cdot) \big\|_{L_2(\mathbb{P})} \cdot \big\|  \hat\eta^{\cI_k}(\cdot) - \eta(\cdot)\big\|_{L_2(\mathbb{P})} = o_P(1/\sqrt{n})
$
for $k=1,2$ and $\ell=1,2,3$. 
\end{assumption}

The following theorem 
proved in Appendix~\ref{app:subsec_proof_est_cov_shift} 
provides inference that is robust 
to estimation error---we obtain $n^{-1/2}$-rate inference 
with the same asymptotic variance 
as the case of known covariate shift, 
as long as the product of the errors is no greater than $O(n^{-1/2})$.  

\begin{theorem}\label{thm:est_cov_shift}
Suppose Assumptions~\ref{assump:consist_w},~\ref{assump:linear_est_shift} and~\ref{assump:cov_est_rate} 
hold, and $m\geq \epsilon n$ for some fixed $\epsilon>0$. 
If an estimator $\hat\sigma_{\shift}\to\sigma_{\shift}$ in probability
for the variance $\sigma_{\shift}^2$ defined in~\eqref{eq:def_shift_var}, then  
\$
\PP\Big( \theta_m^{\cond,\new} \in \big[ \hat\theta_{m,n}^{\trans,\shift} - \hat\sigma_{\shift} \cdot  z_{1-\alpha/2}/\sqrt{n},    \hat\theta_{m,n}^{\trans,\shift} + \hat\sigma_{\shift} \cdot z_{1-\alpha/2}/\sqrt{n}    \big] \Biggiven Z_{1:m}^\new, Z_{1:n}  \Big),
\$
as a random variable measurable with respect to $Z_{1:m}^\new $, 
converges in probability to $1-\alpha$ as $n\to \infty$, 
where $\hat\theta_{m,n}^{\trans,\shift}$ is defined 
in equation~\eqref{eq:new_trans_est}. 
\end{theorem}

As before, a noteworthy feature of this result is that 
asymptotically the variance does not depend on $m$. 
The 
inference procedure 
in Theorem~\ref{thm:est_cov_shift} 
relies on 
the construction of $\hat\sigma_{\shift}$, 
$\hat{w}_\ell(\cdot)$  
and $\hat\eta^{\cI_k}(\cdot)$. 
In Section~\ref{sec:algo}, 
we provide 
a stand-alone procedure to 
obtain $\hat\eta^{\cI_k}(\cdot)$ 
with a single fold  
${\cI_k}$  (c.f.~Algorithm~\ref{alg:eta_est}) 
and a detailed procedure  
to estimate $\sigma_{\shift}^2$ (c.f.~Algorithm~\ref{alg:sigma_shift_est}).

\section{Details for asymptotic linearity} 


\subsection{Asymptotic linearity 
for conditional inference}\label{app:linear_z}
\begin{proof}[Proof of Proposition~\ref{prop:linear_exp}]
We first show the consistency of $\hat\theta_n \stackrel{P}{\to}\theta^0$ 
and $\theta_n^\cond\stackrel{P}{\to} \theta^0$, 
where the convergence in probability is in each entry. 
The consistency of $\hat\theta_n$ follows directly from the classical results
~\citep[Theorem 5.9]{Vaart1998}. Similarly, we note that $\theta_n^\cond$ 
is the unique solution to~\eqref{eq:est} with the score function replaced by $t(Z_i,\theta)$.
Thus, under the given conditions we have the consistency of $\theta_n^\cond$  
following~\cite[Theorem 5.9]{Vaart1998}. 

We now employ the Taylor expansion argument to obtain 
the asymptotic linearity~\eqref{eq:linear_hat_n} and~\eqref{eq:exp}. 
Recall that $s(D,\theta)\colon \cD\times \Omega \to \RR^p$. 
Expanding $\sum_{i=1}^n s(D_i, \hat\theta_n)$ at $\theta^0$ yields 
\$
0 = \frac{1}{\sqrt{n}} \sum_{i=1}^n s(D_i,\theta^0) +   \frac{1}{\sqrt{n}} \sum_{i=1}^n \dot{s}(D_i,\theta^0) (\hat\theta_n - \theta^0) + \frac{1}{2\sqrt{n}}\sum_{i=1}^n (\hat\theta_n - \theta^0)^\top \ddot{s}(D_i,\tilde\theta_n) (\hat\theta_n - \theta^0),
\$
where the random vector $\tilde\theta_n$ 
lies within the segment between $\theta^0$ and $\hat\theta_n$. 
Rearranging the terms, we have 
\$
- \frac{1}{\sqrt{n}} \sum_{i=1}^n s(D_i,\theta^0)
=  \bigg\{ \frac{1}{n}\sum_{i=1}^n s(D_i,\theta^0) + \frac{1}{2 n}\sum_{i=1}^n (\hat\theta_n - \theta^0)^\top \ddot{s}(D_i,\tilde\theta_n) \bigg\} \cdot \sqrt{n}(\hat\theta_n - \theta^0).
\$
The law of large numbers implies 
$\frac{1}{n} \sum_{i=1}^n \dot{s}(D_i,\theta^0) = \EE\{\dot{s}(D ,\theta^0)\} + o_P(1)$, 
where $E\{\dot{s}(D,\theta^0)\}$ is non-singular according to (iv). 
Meanwhile, 
Condition (iv) implies $\|\frac{1}{2n}\sum_{i=1}^n (\hat\theta_n - \theta^0)^\top \ddot{s}(D_i,\tilde\theta_n)  \|_1 \leq \|\hat\theta_n-\theta^0\|_1  \cdot \frac{1}{2n}\sum_{i=1}^n g(D_i) = o_P(1 )$ since $ \hat\theta_n  $ converges in probability to $\theta^0$. Hence 
\$
\big[ \EE\{\dot{s}(D ,\theta^0)\} + o_P(1) \big] \cdot \sqrt{n}(\hat\theta_n - \theta^0) = - \frac{1}{\sqrt{n}} \sum_{i=1}^n s(D_i,\theta^0) . 
\$
On the left-handed side,  $o_P(1)$ 
means a random matrix where each entry converges in probability to zero. 
Thus we have 
\$
\sqrt{n}(\hat\theta_n - \theta^0) = - \frac{1}{\sqrt{n}} \sum_{i=1}^n \big[\EE\{\dot{s}(D ,\theta^0)\} \big]^{-1} s(D_i,\theta^0) + o_P(1). 
\$
That is, the asymptotic linearity~\eqref{eq:linear_hat_n} holds with 
\$
\phi(D) =  - \big[\EE\{\dot{s}(D ,\theta^0)\} \big]^{-1} s(D ,\theta^0).
\$
On the other hand, recall the observation that $\theta_n^\cond$ 
is the unique solution to~\eqref{eq:est} with the score function replaced by $t(Z_i,\theta)$.
Meanwhile, condition (iv) implies also $\|\ddot{s}_{jk}(Z,\theta)\|\leq \EE\{g(D)\given Z\}$ 
due to Jensen's inequality; and 
$E\{\dot{t}(Z,\theta^0)\} = E\{\dot{s}(D,\theta^0)\}$ 
due to the tower property 
of conditional expectations and 
the exchangeability of expectation and derivative in (iv). 
Following exactly the same arguments, we have 
\$
\sqrt{n}( \theta_n^\cond - \theta^0) &= - \frac{1}{\sqrt{n}} \sum_{i=1}^n \big[\EE\{\dot{t}(Z ,\theta^0)\} \big]^{-1} t(Z_i,\theta^0) + o_P(1) \\
&= - \frac{1}{\sqrt{n}} \sum_{i=1}^n \big[\EE\{\dot{s}(D ,\theta^0)\} \big]^{-1} E\big\{s(D_i,\theta^0)\biggiven Z_i\big\} + o_P(1).
\$
That is, the asymptotic linearity~\eqref{eq:exp} holds with the same $\phi(D)$. 
Therefore, we complete the proof of Proposition~\ref{prop:linear_exp}.
\end{proof}

\subsection{Justification for asymptotic linearity in main text}

\label{app:subsec_main_linear_trans}

In this part, we 
show that the asymptotic linearity 
in Assumption~\ref{assump:linear_expansion_known_shift} 
of the main text holds 
under the consistency of $\hat{w}$ and 
regularity conditions that are similar to 
Proposition~\ref{prop:linear_exp} in the main text. 

\begin{proposition}\label{prop:main_new_lin_exp}
Suppose conditions (ii), (iii) in Proposition~\ref{prop:linear_exp} 
hold also at $\theta=\theta_0^\new$
and the following two conditions hold: 
(i')  
$\theta_0^\new$ is the unique solution to~\eqref{eq:theta_dag}, 
$\theta_m^{\cond,\new}$ is the unique solution to~\eqref{eq:eq_new_cond}
and $\hat\theta_n^\trans$ is the unique solution to~\eqref{eq:hat_theta_dag}. 
(iv') For each $j,k$, $\|\ddot{s}_{jk}(D,\theta)\|= \|\partial s(D,\theta)/\partial\theta_j \partial \theta_k\| \leq g(D)$, 
where $g(D)$ and 
$g(D)w(Z)$ are both integrable. 
Also, both $\EE\{\dot{s}(D,\theta_0^\new)\}$ and $\EE\{w(Z)\dot{s}(D,\theta_0^\new)\}$ are non-singular matrices. 
Then Assumption~\ref{assump:linear_expansion_known_shift} holds if $\sup_z|\hat{w}(z)-w(z)|$ converges to zero in probability.
\end{proposition}

\begin{proof}[Proof of Proposition~\ref{prop:main_new_lin_exp}]
Since $\hat{w}$ is obtained from an external dataset, 
we condition on the training process of $\hat{w}$; 
in this way, we view $\hat{w}(\cdot)$ 
as fixed. 
Also, without loss of generality we suppose 
$\sup_z|\hat{w}(z)-w(z)|$ converges to zero 
as $n\to \infty$. 

We first show  $\hat\theta_n^{\trans}\stackrel{P}{\to} \theta_0^\new$.  
To this end, we utilize Theorem 5.9 of~\cite{Vaart1998} and define  
\$
\hat{S}(\theta) = \frac{1}{n}\sum_{i=1}^n \hat{w} (Z_i) s(D_i,\theta),\quad S(\theta)=\EE\big\{w(Z)s(D,\theta)\big\}, 
\$ 
so that it suffices to show (a) $\sup_{\theta\in \Theta}|\hat{S}(\theta) - S(\theta)| \to 0$ in probability, and 
(b) for any $\epsilon>0$, there exists some $\delta>0$ 
such that $\inf_{\|\theta-\theta_0^\new\|_2>\delta} |S(\theta)-S(\theta_0^\new)|>\epsilon$. 
Firstly, for any fixed $\theta\in \Theta$, we have 
\$
\big|\hat{S}(\theta) - S(\theta)\big| \leq \frac{1}{ n}\sum_{i =1}^n \big|\hat{w} (Z_i) - w(Z_i)\big| \cdot\big| s(D_i,\theta)\big|
+ \bigg| \frac{1}{ n}\sum_{i =1}^n w(Z_i)s(D_i,\theta) - S(\theta)\bigg|,
\$
where 
\$
\frac{1}{ n}\sum_{i =1}^n \big|\hat{w} (Z_i) - w(Z_i)\big| \cdot\big| s(D_i,\theta)\big|
\leq \sup_{z}\big|\hat{w} (z)-w(z)| \cdot \frac{1}{ n}\sum_{i =1}^n \big| s(D_i,\theta)\big| = o_P(1)
\$
by the consistency assumption 
and the integrability of $s(D,\theta)$. 
The second term also converges to zero by the law of large numbers. Hence $|\hat{S}(\theta)-S(\theta)|= o_P(1)$ 
for any fixed $\theta\in \Theta$. 
By compactness of $\Theta$ in condition (ii) of Proposition~\ref{prop:linear_exp} as well as the continuity 
of $\hat{S}(\theta)$ and $S(\theta)$, we know that 
the uniform convergence in (a) holds. 
The compactness of $\Theta$ and the uniqueness of solution $\theta_0^\new$ implies the well-separatedness condition (b)
(c.f.~Theorem 5.9 of~\cite{Vaart1998}). 
Thus we have $\hat\theta_n^{\trans}\to \theta_0^\new$ 
in probability as $n\to \infty$. 

We now employ a Taylor expansion argument to show the 
asymptotic linearity. Expanding $\hat{S}(\hat\theta_n^{\trans})$ 
at $\theta_0^\new$ yields
\$
0 &= \sum_{i =1}^n \hat{w} (Z_i) s(D_i,\theta_0^\new) 
+ \sum_{i =1}^n \hat{w} (Z_i) \dot{s}(D_i,\theta_0^\new) (\hat\theta_n^{\trans} - \theta_0^\new) \\
&\qquad + 
\frac{1}{2}\sum_{i =1}^n \hat{w} (Z_i) (\hat\theta_n^{\trans} - \theta_0^\new)^\top  \ddot{s}(D_i,\theta_0^\new) (\hat\theta_n^{\trans} - \theta_0^\new).
\$
Utilizing the fact that each entry of 
$\ddot{s}(D_i,\theta_0^\new)$ is 
controlled by an integrable $g(D_i)$, the random variable
\$
\frac{1}{ n}\sum_{i=1}^n \hat{w} (Z_i) \ddot{s}(D_i,\theta_0^\new) = \frac{1}{ n}\sum_{i=1}^n \big\{\hat{w} (Z_i)-w(Z_i) \big\} \ddot{s}(D_i,\theta_0^\new)
+ \frac{1}{ n}\sum_{i=1}^n {w}(Z_i) \ddot{s}(D_i,\theta_0^\new) 
\$
is of order $O_P(1)$, hence 
\$
\frac{1}{2n}\sum_{i=1}^n \hat{w} (Z_i) (\hat\theta_n^{\trans} - \theta_0^\new)^\top  \ddot{s}(D_i,\theta_0^\new) = o_P(1).
\$
The above $O_P(1)$ and $o_P(1)$ are both in the entry-wise sense. 
Reorganizing the Taylor expansion,  
\$
- \frac{1}{ \sqrt{n}}\sum_{i=1}^n \hat{w} (Z_i) s(D_i,\theta_0^\new) 
= \bigg\{ o_P(1) +\frac{1}{ n}\sum_{i=1}^n \hat{w} (Z_i) \dot{s}(D_i,\theta_0^\new)\bigg\}\cdot \sqrt{n} (\hat\theta_n^{\trans} - \theta_0^\new) .
\$
Following similar arguments as before, we also have 
\$
\frac{1}{ n}\sum_{i =1}^n \hat{w} (Z_i) \dot{s}(D_i,\theta_0^\new) =  \EE\big\{ w(Z_i) \dot{s}(D_i,\theta_0^\new)  \big\}  + o_P(1).
\$
Since the expected matrix is invertible by condition (iv') in Proposition~\ref{prop:main_new_lin_exp}, we know 
\$
\sqrt{n}(\hat\theta_n^{\trans} - \theta_0^\new)
=- \frac{1}{\sqrt{n}}\sum_{i=1}^n \Big[\EE\big\{ w(Z_i) \dot{s}(D_i,\theta_0^\new)  \big\}\Big]^{-1} \hat{w} (Z_i) s(D_i,\theta_0^\new) + o_P(1),
\$
which is equivalent to the linear expansion 
for $\hat\theta_n^\trans$ in 
Assumption~\ref{assump:linear_expansion_known_shift} 
by the definition of $\psi(\cdot)$. 

We now show the linear expansion of $\theta_m^{\cond,\new}$. 
Note that the conditions in Proposition~\ref{prop:main_new_lin_exp}
imply the same conditions as Proposition~\ref{prop:linear_exp} 
when we substitute $(D_i,Z_i)\sim \mathbb{P}$ 
with $(D_j^\new,Z_j^\new) \sim \mathbb{Q}$. 
Thus, applying the same arguments as 
those in the proof of Proposition~\ref{prop:linear_exp} 
leads to the linear expansion of $\theta_m^{\cond,\new}$ 
in Assumption~\ref{assump:linear_expansion_known_shift} 
of the main text. 
Therefore, we complete the proof of Proposition~\ref{prop:main_new_lin_exp}. 
\end{proof}

\subsection{Justification for asymptotic linearity with cross-fitting in Section~\ref{app:subsec_cross_fitting_est} }
\label{app:subsec_linear_est}

In this part, we justify the asymptotic linearity 
for cross-fitted estimators in 
Assumption~\ref{assump:linear_est_shift} of
Section~\ref{app:subsec_cross_fitting_est}; 
we show that it holds 
under consistency of $\hat{w}_\ell$ (Assumption~\ref{assump:consist_w}) 
and mild regularity conditions. 

\begin{proposition}\label{prop:linear_hat_theta_est}
Suppose Assumption~\ref{assump:consist_w}, 
conditions (ii), (iii) in Proposition~\ref{prop:linear_exp} 
and conditions  (iv') in Proposition~\ref{prop:main_new_lin_exp} hold. Also, suppose 
the following  condition  hold: 
(i'')  
$\theta_0^\new$ is the unique solution to~\eqref{eq:theta_dag}, 
$\theta_m^{\cond,\new}$ is the unique solution to~\eqref{eq:eq_new_cond}
and $\hat\theta_n^{\new,(\ell)}$ is the unique solution to~\eqref{eq:def_new_hat_theta_1}. 
Then Assumption~\ref{assump:linear_est_shift} holds.  
\end{proposition}

\begin{proof}[Proof of Proposition~\ref{prop:linear_hat_theta_est}]
The proof is quite similar to that of 
Proposition~\ref{prop:main_new_lin_exp}, 
except for the notational complexity 
and the change in sample size due to 
cross-fitting. 
The asymptotic linearity of $\theta_m^{\cond,\new}$ 
has been proved in Proposition~\ref{prop:main_new_lin_exp}. 
We thus only need to prove~\eqref{eq:linear_hat_theta_est} 
for $\hat\theta_n^{\new,(\ell)}$. 

We first show the consistency of  $\hat\theta_n^{\new,(\ell)}\stackrel{P}{\to} \theta_0^\new$. 
Without loss of generality,  
the sample size is $ |\cI\backslash \cI_\ell| = 2n/3$.
To this end, we utilize Theorem 5.9 of~\cite{Vaart1998}, 
so that it suffices to show (a) $\sup_{\theta\in \Theta}|\hat{S}(\theta) - S(\theta)| \to 0$ in probability, and 
(b) for any $\epsilon>0$, there exists some $\delta>0$ 
such that $\inf_{\|\theta-\theta_0^\new\|_2>\delta} |S(\theta)-S(\theta_0^\new)|>\epsilon$, where we define 
\$
\hat{S}(\theta) = \frac{3}{2n}\sum_{i \notin \cI_\ell} \hat{w}_\ell(Z_i) s(D_i,\theta),\quad S(\theta)=\EE\big\{ w(Z)s(D,\theta)\big\}.
\$
Firstly, for any fixed $\theta\in \Theta$, we have 
\$
\big|\hat{S}(\theta) - S(\theta)\big| \leq \frac{3}{2n}\sum_{i \notin \cI_\ell} \big|\hat{w}_\ell(Z_i) - w(Z_i)\big| \cdot\big| s(D_i,\theta)\big|
+ \bigg| \frac{3}{2n}\sum_{i \notin \cI_\ell} w(Z_i)s(D_i,\theta) - S(\theta)\bigg|,
\$
where 
\$
\frac{3}{2n}\sum_{i \notin \cI_\ell} \big|\hat{w}_\ell(Z_i) - w(Z_i)\big| \cdot\big| s(D_i,\theta)\big|
\leq \sup_{z}\big|\hat{w}_\ell(z)-w(z)| \cdot \frac{3}{2n}\sum_{i \notin \cI_\ell} \big| s(D_i,\theta)\big| = o_P(1)
\$
by Assumption~\ref{assump:consist_w} and the integrability of $s(D,\theta)$. The second term also converges to zero by the law of large numbers. Hence $|\hat{S}(\theta)-S(\theta)|= o_P(1)$ 
for any fixed $\theta\in \Theta$. 
By compactness of $\Theta$ in condition (ii) of Proposition~\ref{prop:linear_exp} as well as the continuity 
of $\hat{S}(\theta)$ and $S(\theta)$, we know that 
the uniform convergence in (a) holds. 
The compactness of $\Theta$ and the uniqueness of solution $\theta_0^\new$ implies the well-separatedness condition (b)
(c.f.~Theorem 5.9 of~\cite{Vaart1998}). 
Thus we have $\hat\theta_n^{\new,(\ell)}\to \theta_0^\new$ 
in probability as $n\to \infty$. 

We now employ a Taylor expansion argument to show the 
asymptotic linearity. Expanding $\hat{S}(\hat\theta_n^{\new,(\ell)})$ 
at $\theta_0^\new$ yields
\$
0 &= \sum_{i \notin \cI_\ell} \hat{w}_{\ell}(Z_i) s(D_i,\theta_0^\new) 
+ \sum_{i \notin \cI_\ell} \hat{w}_{\ell}(Z_i) \dot{s}(D_i,\theta_0^\new) (\hat\theta_n^{\new,(\ell)} - \theta_0^\new) \\
&\qquad + 
\frac{1}{2}\sum_{i \notin \cI_\ell} \hat{w}_{\ell}(Z_i) (\hat\theta_n^{\new,(\ell)} - \theta_0^\new)^\top  \ddot{s}(D_i,\theta_0^\new) (\hat\theta_n^{\new,(\ell)} - \theta_0^\new).
\$
Here utilizing the fact that each entry of 
$\ddot{s}(D_i,\theta_0^\new)$ is 
controlled by an integrable $g(D_i)$, the random variable
\$
\frac{3}{2n}\sum_{i \notin \cI_\ell} \hat{w}_\ell (Z_i) \ddot{s}(D_i,\theta_0^\new) = \frac{3}{2n}\sum_{i \notin \cI_\ell} \big\{ \hat{w}_{\ell}(Z_i)-w(Z_i) \big\}\ddot{s}(D_i,\theta_0^\new)
+ \frac{3}{2n}\sum_{i \notin \cI_\ell} {w}(Z_i) \ddot{s}(D_i,\theta_0^\new) 
\$
is of order $O_P(1)$, hence 
\$
\frac{3}{4n}\sum_{i \notin \cI_\ell} \hat{w}_{\ell}(Z_i) (\hat\theta_n^{\new,(\ell)} - \theta_0^\new)^\top  \ddot{s}(D_i,\theta_0^\new) = o_P(1).
\$
The above $O_P(1)$ and $o_P(1)$ are both in the entry-wise sense. 
Reorganizing the Taylor expansion,  
\$
- \frac{1}{ \sqrt{2n/3}}\sum_{i \notin \cI_\ell} \hat{w}_{\ell}(Z_i) s(D_i,\theta_0^\new) 
= \bigg\{ o_P(1) +\frac{3}{2n}\sum_{i \notin \cI_\ell} \hat{w}_{\ell}(Z_i) \dot{s}(D_i,\theta_0^\new)\bigg\}\cdot \sqrt{2n/3} (\hat\theta_n^{\new,(\ell)} - \theta_0^\new) .
\$
Following similar arguments as before, we also have 
\$
\frac{3}{2n}\sum_{i \notin \cI_\ell} \hat{w}_{\ell}(Z_i) \dot{s}(D_i,\theta_0^\new) =  \EE\big\{ w(Z_i) \dot{s}(D_i,\theta_0^\new)  \big\}  + o_P(1).
\$
Since the expected matrix is invertible by condition (iv') in Proposition~\ref{prop:main_new_lin_exp}, we know 
\$
\sqrt{2n/3}(\hat\theta_n^{\new,(\ell)} - \theta_0^\new)
=- \frac{1}{\sqrt{2n/3}}\sum_{i \notin \cI_\ell} \Big[ \EE\big\{ w(Z_i) \dot{s}(D_i,\theta_0^\new)  \big\}\Big]^{-1} \hat{w}_{\ell}(Z_i) s(D_i,\theta_0^\new) + o_P(1),
\$
which is equivalent to~\eqref{eq:linear_hat_theta_est} 
by the definition of $\psi(\cdot)$. 
Therefore, we complete the proof of Proposition~\ref{prop:linear_hat_theta_est}.
\end{proof}

\section{Proofs for main results}

\subsection{Proofs of validity of conditional inference}\label{app:cond_inf_single}

This section contains the proof of Theorem~\ref{thm:cond_intv}.
Before proving Theorem~\ref{thm:cond_intv}, 
we first state and prove an intermediate result 
on the asymptotic distribution of $\hat\theta_n-\theta_n^\cond$.

\begin{proposition}
\label{prop:cond_distr}
Suppose Assumptions \ref{assump:linear_main} and \ref{assump:moment_main}
hold. For any fixed $x\in \RR$, the random variable 
    $ 
    \PP \{ \sqrt{n}(\hat \theta_n - \theta_n^{\cond}) \le x \given \cZ_n \} 
    $ 
converges in probability to  $\Phi(x/\sigma)$, 
where  $\Phi$ is the cumulative distribution function (c.d.f.) of standard Gaussian distribution, and $\sigma^2$ is defined in equation~\eqref{eq:def_asymp_var}.
\end{proposition}

\begin{proof}[Proof of Proposition \ref{prop:cond_distr}]
By Assumptions \ref{assump:linear_main} and \ref{assump:moment_main}, we have 
\$
\sqrt{n}\big( \hat\theta_n - \theta_n^\cond \big) = \frac{1}{\sqrt{n}}\sum_{i=1}^n \Big[  \phi(D_i) - \EE\big\{\phi(D_i)\given Z_i \big\} \Big] + o_{P}(1).
\$
For notational simplicity, we write
\$
d_n = \sqrt{n}\big(\hat\theta_n - \theta_n^\cond\big) - \frac{1}{\sqrt{n}}\sum_{i=1}^n \zeta_i,\quad \text{where }~ \zeta_i= \phi(D_i)  -\EE\{\phi(D_i)\given Z_i \},~i=1,\dots,n,
\$
where $d_n=o_P(1)$ follows from the given conditions. 
Hence Lemma \ref{lem:op_to_op} implies that for any fixed $\epsilon>0$, 
\#\label{eq:d_n_o1}
\PP\big\{ |d_n|>\epsilon \biggiven \cZ_n\big\} =o_P(1).
\#
On the other hand, we denote the conditional law of 
the essential term as 
$$
\mathcal{L}_n = \mathcal{L}\Big(\frac{1}{\sqrt{n}}\sum_{i=1}^n \big(\phi(D_i)-\EE[\phi(D_i)\given Z_i ]\big) \Biggiven Z_{1:n}\Big).
$$
By the conditional CLT in Lemma \ref{lem:cond_clt}, taking $g(X_i) = \phi(D_i)$ and the filtration $\cF_n =\sigma(Z_{1:n})= \sigma(\{Z_i \}_{i=1}^n)$, we know that the conditional law $\cL_n$ converges almost surely to $N\big(0,\sigma^2\big)$
with $\sigma^2$ defined in equation~\eqref{eq:def_asymp_var}. 
That is, for any $x\in \RR$, we have
$$
\PP\big\{\sqrt{n}(\hat\theta_n -\theta_n^\cond )+d_n \leq x\given Z_{1:n}\big\}~\asto~ \Phi\Big(\frac{x}{\sigma }\Big),
$$
where $\Phi(\cdot)$ is the cumulative distribution function of standard normal distribution. 
By equation \eqref{eq:d_n_o1}, for any constant $\epsilon>0$, it holds that 
\#\label{eq:p+}
&\PP\big\{\sqrt{n}(\hat\theta_n -\theta_n^\cond) \leq x\biggiven Z_{1:n} \big\}\notag \\
&=\PP\big\{\sqrt{n}(\hat\theta_n -\theta_n^\cond) \leq x, |d_n|\leq \epsilon \biggiven Z_{1:n} \big\} + \PP\big\{\sqrt{n}(\hat\theta_n -\theta_n^\cond) \leq x, |d_n|> \epsilon \biggiven Z_{1:n} \big\}\notag \\
&\leq \PP\big\{\sqrt{n}(\hat\theta_n -\theta_n^\cond)+d_n \leq x+\epsilon \biggiven Z_{1:n} \big\} + \PP\big(|d_n|> \epsilon \biggiven Z_{1:n} \big)
= \Phi\Big(\frac{x+\epsilon }{\sigma }\Big) + o_P(1).
\#
On the other hand, we have
\#\label{eq:p-}
&\PP\big\{\sqrt{n}(\hat\theta_n -\theta_n^\cond) \leq x\biggiven Z_{1:n} \big\}\notag \\
&\geq \PP\big\{\sqrt{n}(\hat\theta_n -\theta_n^\cond)+d_n \leq x-\epsilon,|d_n|\leq \epsilon \biggiven Z_{1:n} \big\}\notag \\
&\geq  \PP\big\{\sqrt{n}(\hat\theta_n -\theta_n^\cond)+d_n \leq x-\epsilon \biggiven Z_{1:n}\big\} -\PP\big(|d_n|> \epsilon \biggiven Z_{1:n} \big)    
= \Phi\Big(\frac{x-\epsilon }{\sigma }\Big) + o_P(1).
\#
By the arbitrariness of $\epsilon>0$ in equations \eqref{eq:p+} and \eqref{eq:p-},  for any fixed $x\in \RR$, it holds that
\$
\PP\big\{\sqrt{n}(\hat\theta_n -\theta_n^\cond) \leq x\biggiven Z_{1:n} \big\} = \Phi\Big(\frac{x}{\sigma }\Big) + o_P(1).
\$
Therefore, we conclude the proof of Proposition~\ref{prop:cond_distr}. 
\end{proof}

The proof of Theorem~\ref{thm:cond_intv} is as follows.

\begin{proof}[Proof of Theorem \ref{thm:cond_intv}]
By Proposition~\ref{prop:cond_distr}, 
for any fixed $x\in \RR$, 
\$
\PP\big\{ \sqrt{n}( \hat\theta_n  - \theta_n^{\cond} )\leq x \biggiven Z_{1:n}\big\} = \Phi(x/\sigma) + o_{P}(1).
\$
For any fixed constant $\epsilon>0$, we write $z^-(\epsilon) = z_{1-\alpha/2}(\sigma-\epsilon)$
and $z^+(\epsilon) = z_{1-\alpha/2}(\sigma+\epsilon)$. 
Denoting 
\$
\Delta^\pm(\epsilon) = \PP\Big\{ \sqrt{n}\big|\hat\theta_n  - \theta_n^{\cond} \big|\leq z^\pm(\epsilon) \Biggiven Z_{1:n}\Big\} - \Big\{2\Phi\big(z^\pm(\epsilon)/\sigma\big)-1\Big\},
\$
we have $\Delta^+(\epsilon), \Delta^-(\epsilon)=o_P(1)$ by 
Proposition~\ref{prop:cond_distr}. 
Since the estimator $\hat\sigma\stackrel{P}{\to}\sigma$,  we have
\#\label{eq:pval1}
&\PP\Big( \sqrt{n}\big|\hat\theta_n  - \theta_n^{\cond} \big|\leq z_{1-\alpha/2}\cdot \hat\sigma \Biggiven Z_{1:n}\Big) - (1-\alpha) \notag \\
&\geq \PP\Big(\{\sqrt{n}\big|\hat\theta_n  - \theta_n^{\cond} \big|\leq z_{1-\alpha/2}\cdot ( \sigma-\epsilon) \Biggiven Z_{1:n}\Big\} - (1-\alpha)  + \PP ( \hat\sigma < \sigma-\epsilon \given Z_{1:n} )\notag  \\
&= 2\Phi\big\{z^-(\epsilon)/\sigma\big\}-1 - (1-\alpha)  + \PP ( \hat\sigma < \sigma-\epsilon \given Z_{1:n} ) +   \Delta^-(\epsilon),
\#
where the conditional probability 
$\PP ( \hat\sigma < \sigma-\epsilon \given \cZ_n )=o_{P}(1)$ by Lemma \ref{lem:op_to_op}. 
On the other hand, for any fixed constant $\epsilon>0$, we have 
\#\label{eq:pval2}
&\PP\Big( \sqrt{n}\big|\hat\theta_n  - \theta_n^{\cond} \big|\leq z_{1-\alpha/2}\cdot \hat\sigma \Biggiven Z_{1:n}\Big)\notag \\
&\leq \PP\Big\{ \sqrt{n}\big|\hat\theta_n  - \theta_n^{\cond} \big|\leq z_{1-\alpha/2}\cdot ( \sigma+\epsilon) \Biggiven Z_{1:n}\Big\}  + \PP ( \hat\sigma > \sigma+\epsilon \given Z_{1:n} )\notag  \\
&= 2\Phi\big(z^+(\epsilon)/\sigma\big) -1 - (1-\alpha)  + \PP ( \hat\sigma < \sigma+\epsilon \given Z_{1:n} ) +   \Delta^+(\epsilon),
\#
where $\PP ( \hat\sigma > \sigma+\epsilon \given Z_{1:n} )=o_{P}(1)$ by Lemma \ref{lem:op_to_op}. 
Thus for any fixed constant $\delta>0$, we can choose some fixed $\epsilon>0$ such that 
$2\Phi(z^-(\epsilon)/\sigma)- 1 -(1-\alpha) >-\delta/2$ and $2\Phi\{z^+(\epsilon)/\sigma\}-1 -(1-\alpha)<\delta/2$. 
Combining equations \eqref{eq:pval1} and \eqref{eq:pval2}, 
we have 
\$
&\PP\bigg\{\Big| \PP\big( \sqrt{n}\big|\hat\theta_n  - \theta_n^{\cond} \big|\leq z_{1-\alpha/2}\cdot \hat\sigma \biggiven Z_{1:n}\big) - (1-\alpha) \Big| >\delta\bigg\} \\
&\leq \PP\big\{\PP ( \hat\sigma < \sigma-\epsilon \given Z_{1:n} ) +   \Delta^-(\epsilon) < -\delta/2\big\} + \PP\big\{\PP ( \hat\sigma < \sigma+\epsilon \given Z_{1:n} ) +   \Delta^+(\epsilon) >\delta/2\big\} \to 0.
\$
By the arbitrariness of $\delta>0$, we complete  the proof of Theorem \ref{thm:cond_intv}.
\end{proof}

\subsection{Proof of Theorem~\ref{thm:iid_simple}}
\label{app:thm_iid_simple}

\begin{proof}[Proof of Theorem~\ref{thm:iid_simple}]
We note that 
all the methods and conditions 
in Theorem~\ref{thm:iid_simple}  
can be viewed as a special case 
for those in Section~\ref{app:subsec_cross_fitting_known} 
with $w(z)\equiv 1$. 
Thus, Theorem~\ref{thm:iid_simple} could be 
viewed as a corollary for Theorem~\ref{thm:transfer}, 
whose proof is in Section~\ref{app:thm_transfer} 
in this supplementary material. 
\end{proof}

\subsection{Proof of Theorem~\ref{thm:est_cov_shift_simple}}
\label{app:thm_cov_shift_simple}

\begin{proof}[Proof of Theorem~\ref{thm:est_cov_shift_simple}]
Throughout this proof, we condition on the 
training process of $\hat{w}$ and $\hat{\eta}$, 
so that they are deterministic functions. 
All probabilities and expectations 
are with respect to the i.i.d.~samples in the two datasets. 
By Assumption~\ref{assump:linear_expansion_known_shift} 
and the fact that $m\geq \epsilon n$, 
we have  
\$
&\hat\theta_{m,n}^{\trans} - \theta_m^{\cond,\new} \\
&= \frac{1}{n}\sum_{i=1}^n \psi(D_i) \hat{w}(Z_i) 
- \frac{1}{n} \sum_{i =1}^n  \hat \eta (Z_i) \hat w(Z_i)
+ \frac{1}{m} \sum_{j=1}^m \big\{  \hat\eta(Z_j^\new) - \eta(Z_j^\new) \big\} + o_P(1/\sqrt{n}) \\ 
&= \underbrace{\frac{1}{n}\sum_{i=1}^n \big\{\psi(D_i) - \eta(Z_i)\big\} \hat{w}(Z_i) }_{\textrm{(i)}}
- \underbrace{\frac{1}{n}\sum_{i=1}^n \big\{\hat\eta(Z_i) - \eta(Z_i)\big\} \big\{\hat{w}(Z_i)- w(Z_i)\big\}}_{\textrm{(ii)}} \\
&\qquad + \underbrace{\frac{1}{m} \sum_{j=1}^m \big\{ \hat\eta(Z_j^\new) - \eta(Z_j^\new) \big\} - \frac{1}{n}\sum_{i=1}^n \big\{\hat\eta(Z_i) - \eta(Z_i)\big\} w(Z_i) }_{\textrm{(iii)}} + o_P(1/\sqrt{n}) .
\$
We now treat these three terms separately. 
Firstly, term (i) can be decomposed as 
\$
\textrm{(i)} = \frac{1}{n}\sum_{i=1}^n \big\{\psi(D_i) - \eta(Z_i)\big\} \big\{\hat{w}(Z_i)- w(Z_i)\big\}
+  \frac{1}{n}\sum_{i=1}^n \big\{\psi(D_i) - \eta(Z_i)\big\}  {w}(Z_i),
\$
where each i.i.d.~copy in the first summation obeys
\$
\EE\big[ \big\{\psi(D_i) - \eta(Z_i) \big\} \big\}\hat{w}(Z_i)- w(Z_i) \big\}\big] = 
\EE\Big[ \EE\big\{ \psi(D_i) - \eta(Z_i) \given Z_i   \big\}  
\cdot \big\{\hat{w}(Z_i)- w(Z_i) \big\}\Big] = 0.
\$
By Markov's inequality, we have 
\$
 &\frac{1}{n}\sum_{i=1}^n \big\{\psi(D_i) - \eta(Z_i)\big\} \big\{\hat{w}(Z_i)- w(Z_i)\big\} \\
 &= O_P\big\{ \|\psi(D)-\eta(Z)\|_{L_2(\mathbb{P})} 
 \cdot \|\hat{w}(Z)-w(Z)\|_{L_2(\mathbb{P})}/\sqrt{n} \big\}
 = o_P(1/\sqrt{n})
\$
due to the consistency condition on $\hat{w}$. 
Secondly, 
by the product rate of $\hat\eta$ and $\hat{w}$, 
term (ii) can be bounded by Cauchy-Schwarz inequality as 
\$
\big|\textrm{(ii)}\big|\leq O_P\Big\{ \big\|\hat{w} (\cdot) - w(\cdot) \big\|_{L_2(\mathbb{P})} \cdot \big\|  \hat\eta (\cdot) - \eta(\cdot)\big\|_{L_2(\mathbb{P})}  \Big\}
= o_P(1/\sqrt{n}).
\$
Noting the covariate shift between $Z_i$ and $Z_j^\new$, 
we know that  
\$
\EE\big\{\hat\eta(Z_j^\new)-\eta(Z_j^\new)\big\} 
= \EE\big\{ (\hat\eta(Z_i)-\eta(Z_i))w(Z_i)   \big\}.
\$
Subtracting both sides from term (iii), we obtain 
\$
\textrm{(iii)}& = \frac{1}{m} \sum_{j=1}^m \Big[   \hat\eta(Z_j^\new) - \eta(Z_j^\new) - \EE\big\{\hat\eta(Z_j^\new)-\eta(Z_j^\new)\big\}  \Big] \\ 
&\qquad - \frac{1}{n}\sum_{i=1}^n \Big(\big\{\hat\eta(Z_i) - \eta(Z_i)\big\} w(Z_i) - \EE\big[\{\hat\eta(Z_i)-\eta(Z_i)\} w(Z_i)   \big] \Big).
\$
In the first summation, 
the i.i.d.~copies are mean zero with variance 
bounded by $\|\hat\eta(Z_j^\new)-\eta(Z_j^\new)\|_{L_2(\mathbb{Q})} = o_P(1)$. 
By Markov's inequality, the first summation is  
bounded by $o_P(1/\sqrt{n})$. 
Similarly, the second summation is also bounded by 
$o_P(1/\sqrt{n})$. 
Thus, we have $\textrm{(iii)}=o_P(1/\sqrt{n})$. 
Combining the three terms together, we obtain 
\$
\sqrt{n}\big(\hat\theta_{m,n}^{\trans} - \theta_m^{\cond,\new} \big)
= \frac{1}{\sqrt{n}}\sum_{i=1}^n \big\{\psi(D_i) - \eta(Z_i)\big\}  {w}(Z_i) +o_P(1). 
\$
Applying the conditional CLT 
in Lemma~\ref{lem:cond_clt} to 
$g(X_i)=\psi(D_i)w(Z_i)$ and the filtration 
\$
\cF_n = \sigma\big(\{Z_i\}_{i=1}^n \cup \{Z_j^\new\}_{j=1}^m\big),
\$
we know that conditional on (almost all) $\cZ_m^\new\cup\cZ_n$, 
$\sqrt{n}\big(\hat\theta_{m,n}^{\trans} - \theta_m^{\cond,\new} \big)$ converges 
in distribution to $N(0,\sigma_{\shift}^2)$.
Finally, by the consisntency of $\hat\sigma_{\shift}^2$ 
to $\sigma_{\shift}^2$ and the Slutsky's theorem, 
we obtain the conditional validity of the 
confidence intervals. We thus complete the 
proof of Theorem~\ref{thm:transfer}. 
\end{proof}

\subsection{Proof of Theorem \ref{thm:transfer}} \label{app:thm_transfer}

\begin{proof}[Proof of Theorem \ref{thm:transfer}]
Recall that $\eta(z)=\EE\{\psi(D)\given Z=z\}$;
by the invariance of conditional distribution of $D$ given $Z$, 
we have 
$\EE\{\psi(D_i)\given Z_i\} = \eta(Z_i)$ and 
$\EE\{\psi(D_j^\new)\given Z_j^\new\} = \eta(Z_j^\new)$ 
for all $i\in[n]$ and all $j\in [m]$. 
In the following, we are to show that
\begin{equation}\label{eq:trans-expansion}
        \hat\theta_{m,n}^{\trans} - \theta_m^{\cond,\new} =  \frac{1}{n} \sum_{i=1}^n w(Z_i)\big\{\psi(D_i) - \eta(Z_i) \big\} + o_P\big\{1/\sqrt{\min(n,m)}\big\}. 
\end{equation}
By the asymptotic linearity in Assumption~\ref{assump:linear_expansion_known_shift} 
(with $\hat{w}:=w$), we have
\$
    \hat \theta_n^\trans - \theta_m^{\cond,\new} &=  \frac{1}{n} \sum_{i=1}^n  w(Z_i)\psi(D_i) - \frac{1}{m} \sum_{j=1}^m \eta(Z_j^\new) + o_P\big(1/\sqrt{n}+1/\sqrt{m}\big).
\$
By the definition of $\hat\theta_{m,n}^\trans$ in equation \eqref{eq:trans_new}, we have the decomposition
\#\label{eq:trans_decomp}
   & \hat\theta_{m,n}^\trans - \theta_m^{\cond,\new}  =  \hat\theta_n^\trans - \hat c^\trans - \theta_m^{\cond,\new} \notag \\
&=  \frac{1}{n} \sum_{i=1}^n w(Z_i)\psi(D_i) - \hat c^\trans - \frac{1}{m} \sum_{j=1}^m \eta(Z_j^\new) + o_P\big(1/\sqrt{n}+1/\sqrt{m}\big).  \notag \\
&=  \frac{1}{n}\sum_{i=1}^n w(Z_i)\big\{\psi(D_i) - \eta(Z_i)\big\} + o_P\big(1/\sqrt{n}+1/\sqrt{m}\big) +\textrm{(i)} +\textrm{(ii)},
\#
where 
\$
\textrm{(i)} &=  \frac{1}{n}\sum_{i=1}^n w(Z_i)\eta(Z_i) - \frac{1}{2|\cI_1|}\sum_{i\in \cI_1}  w(Z_i)\hat\eta^{\cI_2}(Z_i) - \frac{1}{2|\cI_2|}\sum_{i\in \cI_2}  w(Z_i)\hat\eta^{\cI_1}(Z_i) \\
&\quad \quad + \frac{1}{2}\EE_{\mathbb{P}}\big\{w(Z)\hat\eta^{\cI_1}(Z)\biggiven \cI_1\big\} + \frac{1}{2}\EE_{\mathbb{P}}\big\{w(Z)\hat\eta^{\cI_2}(Z)\biggiven \cI_2\big\}, \notag \\
\textrm{(ii)} &=  \frac{1}{2m}\sum_{j=1}^m \big\{ \hat\eta^{\cI_1}(Z_j^\new) + \hat\eta^{\cI_2}(Z_j^\new)  \big\}- \frac{1}{m} \sum_{j=1}^m \eta(Z_j^\new) - \frac{1}{2}\EE_{\mathbb{Q}}\big\{\hat\eta^{\cI_1}(Z)\biggiven \cI_1\big\} - \frac{1}{2}\EE_{\mathbb{Q}}\big\{\hat\eta^{\cI_2}(Z)\biggiven \cI_2\big\} .\notag
\$
Here the decomposition uitilizes the fact that 
\$
\EE_{\mathbb{P}}\big\{w(Z_i)\hat\eta^{\cI_k}(Z_i)\biggiven \cI_k\big\}
= \EE_{\mathbb{Q}}\big\{\hat\eta^{\cI_k}(Z_j^\new)\biggiven \cI_k\big\}
\$
for $i\notin \cI_k$ and $j\in \cI^\new$, $k=1,2$, 
which follows from the fact that $\mathbb{P}$ and $\mathbb{Q}$ 
are related with a covariate shift $w(Z)$, 
and the estimation of $\eta^{\cI_k}$ is 
independent of $\cI^\new$ when $w(\cdot)$ is known. 

In the sequel, we bound the terms (i) and (ii) separately. 
Since $\cI_1$ and $\cI_2$ are (approximately) equal-sized with $|\cI_1|+|\cI_2|=n$, we have 
\$
\text{(i)} &= \underbrace{\frac{1}{2|\cI_1|}\sum_{i\in \cI_1}   w(Z_i)\eta (Z_i) - \frac{1}{2|\cI_1|}\sum_{i\in \cI_1}  w(Z_i)\hat\eta^{\cI_2}(Z_i)+ \frac{1}{2}\EE\big\{w(Z)\hat\eta^{\cI_2}(Z)\biggiven \cI_2\big\} } _{\text{(i,a)}} \\
&\qquad +  \underbrace{\frac{1}{2|\cI_2|}\sum_{i\in \cI_2}  w(Z_i) \eta (Z_i) - \frac{1}{2|\cI_2|}\sum_{i\in \cI_2}w(Z_i)\hat\eta^{\cI_1}(Z_i)+ \frac{1}{2}\EE\big\{w(Z) \hat\eta^{\cI_1}(Z)\biggiven \cI_1\big\}}_{\text{(i,b)}} + O_P(1/n).
\$
For the term (i,a), we note that for $i\in \cI_1$ where $(D_i,Z_i)\sim \mathbb{P}$, 
\$
\EE\big\{w(Z_i)\eta(Z_i)\biggiven \cI_2\big\} = \EE_{\mathbb{Q}} \big\{\eta(Z)\big\}=0.
\$ 
Hence we can write $\text{(i,a)} = \frac{1}{2|\cI_1|}\sum_{i\in \cI_1} \xi_i$, 
where 
\$
\xi_i = w(Z_i)\big\{\eta (Z_i) - \hat\eta^{\cI_2}(Z_i)\big\} -\EE\big\{w(Z_i)\eta (Z_i) -w(Z_i) \hat\eta^{\cI_2}(Z_i)\biggiven \cI_2\big\}.
\$
Conditional on $\cI_2$, $\{\xi_i\}_{i\in \cI_1}$ are $\iidtext$ with mean zero, 
since the estimation of $\hat\eta^{\cI_2}$ 
does not use data in $\cI_1$. Therefore, we have 
\$
\EE\big\{ \text{(i,a)} ^2\biggiven \cI_2\big\} = \frac{1}{4|\cI_1|} \EE(\xi_i^2\given \cI_2)
\leq \frac{1}{4|\cI_1|} \big\|w(\cdot)\{\eta(\cdot) - \hat\eta^{\cI_2}(\cdot)\}\big\|_{L_2(\mathbb{P})}, 
\$
where $\big\|w(\cdot)(\eta(\cdot) - \hat\eta^{\cI_2}(\cdot))\big\|_{L_2(\mathbb{P})}^2 = \EE[\{w(Z)(\eta(Z)-\hat\eta^{\cI_2}(Z))\}^2]$ 
for an independent copy $Z\sim \mathbb{P}$. 
Thus, by Assumption~\ref{assump:cov_shift}, we have 
$
n\cdot \EE\big\{ \text{(i,a)} ^2\biggiven \cI_2\big\} = o_P(1).
$
Referring to Lemma~\ref{lem:op_to_op} for non-negative random variables 
$n\cdot \text{(i,a)}^2$ and the filtration composed of $\cI_2$, 
we have $\text{(i,a)} = o_P(1/\sqrt{n})$. 
The same arguments also apply to the term (i,b), which lead to 
\$
\big|\text{(i)}\big| = o_P(1/\sqrt{n}).
\$
Furthermore, the arguments apply similarly to the term (ii) with sample size $m$, hence 
\$
|\text{(ii)}| = o_P(1/\sqrt{m}).
\$
Putting them together, we have 
\$
\hat\theta_{m,n}^\trans - \theta_m^{\cond,\new}  = \frac{1}{n}\sum_{i=1}^n w(Z_i)\big\{\psi(D_i) - \eta(Z_i)\big)\} + o_P\big(1/\sqrt{n}+1/\sqrt{m}\big).
\$
Applying the conditional CLT result in Lemma \ref{lem:cond_clt} to $g(X_i)=w(Z_i)\psi(D_i)$ and filtrations 
\$
\cF_n = \sigma\big( \{Z_i\}_{i=1}^n, \{Z_j^\new\}_{j=1}^m \big),
\$ 
we know that 
conditional on $\cZ_m^\new \cup \cZ_n$, 
$\frac{1}{\sqrt{n}}\sum_{i=1}^n w(Z_i)\big\{\psi(D_i) - \eta(Z_i)\big\}$ converges in distribution 
to $N(0,\sigma_{\shift}^2)$ almost surely. 
Thus, with similar arguments as in the proof of Theorem~\ref{thm:cond_intv}
for a consistent estimator $\hat\sigma_{\shift}^2$,  
we obtain the desired results in Theorem \ref{thm:transfer}.
\end{proof}

\subsection{Proof of Theorem~\ref{thm:est_cov_shift}}
\label{app:subsec_proof_est_cov_shift}

\begin{proof}[Proof of Theorem~\ref{thm:est_cov_shift}]
For notational simplicity, 
for $\ell=1,2,3$, 
we denote $\cI_{\ell,1}$ and $\cI_{\ell,2}$ as 
the two remaining folds other than $\cI_\ell$, 
and similarly for $\cI_{\ell,1}^\new$ and $\cI_{\ell,2}^\new$. 
We also denote $\hat\eta_{\ell,1}$ and $\hat\eta_{\ell,2}$ 
as the estimators obtained with the two folds. 
For any dataset indexed by $\cI$, 
we use $\cI$ to represent the random variables 
when there is no confusion. 

For any fixed $\ell$, by~\eqref{eq:linear_hat_theta_est} 
and the definition of $\hat{c}^{(\ell)}$, we have 
\$
&\hat\theta_n^{\trans,(\ell)} - \frac{3}{2m} \sum_{j\notin \cI_\ell^\new} \eta(Z_j^\new) \\
= &~
\frac{3}{2n}\sum_{i \notin \cI_\ell}  \hat{w}_{\ell}(Z_i) \psi(D_i) 
- 
\frac{3}{2n} \sum_{i\in \cI_{\ell,1}} \hat{w}_\ell(Z_i) \hat\eta_{\ell,2}(Z_i) - \frac{3}{2n} \sum_{i\in \cI_{\ell,2}} \hat{w}_\ell(Z_i) \hat\eta_{\ell,1}(Z_i) \\
&\qquad + \frac{3}{2m}\sum_{j \in \cI_{\ell,2}^\new}   \hat\eta_{\ell,1}(Z_j^\new)
+ \frac{3}{2m}\sum_{j \in \cI_{\ell,1}^\new}  \hat\eta_{\ell,2}(Z_j^\new) 
- \frac{3}{2m} \sum_{j\notin \cI_\ell^\new} \eta(Z_j^\new)
+ o_P(1/\sqrt{n}).
\$
Writing $\Delta_w(\cdot) = \hat{w}_\ell(\cdot) - w(\cdot)$  
and $\Delta_\eta^{(k)}(\cdot) = \hat\eta_{\ell,k}(\cdot) - \eta(\cdot)$ for 
$k=1,2$, we have 
\$
&\hat\theta_n^{\trans,(\ell)} - \frac{3}{2m} \sum_{j\notin \cI_\ell^\new} \eta(Z_j^\new)  
-
\frac{3}{2n}\sum_{i \notin \cI_\ell}  \hat{w}_{\ell}(Z_i) \big( \psi(D_i) - \eta(Z_i) \big) \\
&= - \frac{3}{2n}\sum_{i\in \cI_{\ell,1}}  \hat{w}_{\ell}(Z_i)  \Delta_\eta^{(2)}(Z_i) 
- \frac{3}{2n} \sum_{i\in \cI_{\ell,2}} \hat{w}_\ell(Z_i) \Delta_\eta^{(1)}(Z_i)   + \frac{3}{2m}\sum_{j \in \cI_{\ell,1}^\new}   \Delta_\eta^{(2)}(Z_j^\new) + \frac{3}{2m}\sum_{j \in \cI_{\ell,2}^\new}   \Delta_\eta^{(1)}(Z_j^\new)  \\
&= \underbrace{- \frac{3}{2n}\sum_{i\in \cI_{\ell,1}}  \Delta_w(Z_i)  \Delta_\eta^{(2)}(Z_i) 
- \frac{3}{2n} \sum_{i\in \cI_{\ell,2}} \Delta_w(Z_i) \Delta_\eta^{(1)}(Z_i)}_{\text{(i)}} \\
&\quad \underbrace{-  \frac{3}{2n}\sum_{i\in \cI_{\ell,1}}  {w} (Z_i)  \Delta_\eta^{(2)}(Z_i) 
- \frac{3}{2n} \sum_{i\in \cI_{\ell,2}} w(Z_i) \Delta_\eta^{(1)}(Z_i)   + \frac{3}{2m}\sum_{j \in \cI_{\ell,1}^\new}   \Delta_\eta^{(2)}(Z_j^\new) + \frac{3}{2m}\sum_{j \in \cI_{\ell,2}^\new}   \Delta_\eta^{(1)}(Z_j^\new) }_{\text{(ii)}}.
\$
We bound the two terms (i) and (ii) separately. 
Since the folds 
$\cD_\ell, \cI_{\ell,1},\cI_{\ell,2}$ 
are disjoint, conditional on $\cI_{\ell}\cup \cI_\ell^\new \cup \cI_{\ell,1}$, 
$\{\Delta_w(Z_i)\Delta_\eta^{(1)}(Z_i)\}_{i\in \cI_{\ell,2}}$ are i.i.d.~random variables. By Cauchy-Schwarz inequality 
and Assumption~\ref{assump:cov_est_rate}, we have 
\$
\EE\bigg\{\frac{3}{2n}\sum_{i\in \cI_{\ell,2}} \big| \Delta_w(Z_i)  \Delta_\eta^{(1)}(Z_i)\big|\Biggiven \cI_{\ell}\cup \cI_\ell^\new \cup \cI_{\ell,1} \bigg\}
\leq \|\Delta_w(\cdot) \|_{L_2(\mathbb{P})} \cdot \|  \Delta_\eta^{(1)}(\cdot)\|_{L_2(\mathbb{P})} = o_P(1/\sqrt{n}).
\$
Invoking Lemma~\ref{lem:cond_to_op} and 
by symmetry of $\cI_{\ell,1}$ and $\cI_{\ell,2}$, we know that 
\$
\big|\text{(i)}\big| = \bigg| \frac{3}{2n}\sum_{i\in \cI_{\ell,1}}  \Delta_w(Z_i)  \Delta_\eta^{(2)}(Z_i) +\frac{3}{2n} \sum_{i\in \cI_{\ell,2}} \Delta_w(Z_i) \Delta_\eta^{(1)}(Z_i)\bigg| = o_P(1/\sqrt{n}).
\$
Furthermore, 
note that the estimation of $\hat\eta_{\ell,k}$ 
only depends on $\cI_{\ell,k}$ and $\cI_{\ell,k}^\new$
for each $k=1,2$. 
Since $\mathbb{P}$, $\mathbb{Q}$ admit a covariate shift, we have 
\$
\EE\big\{w(Z_i)\Delta_\eta^{(k)}(Z_i) \biggiven \cI_\ell\cup\cI_{\ell,k}\cup \cI_{\ell}^\new \cup \cI_{\ell,k}^\new \big\} = \EE_{\mathbb{Q}}\big\{ \Delta_\eta^{(k)}(Z_j^\new) \biggiven \cI_\ell \cup \cI_{\ell,k} \cup \cI_{\ell}^\new \cup \cI_{\ell,k}^\new \big\} := E_\Delta^{(k)}
\$
for $i\in \cI_{\ell,3-k}$ and $j \in \cI_{\ell,3-k}^\new$, $k=1,2$. 
Then we have 
\$
\text{(ii)} &= -  \frac{3}{2n}\sum_{i\in\cI_{\ell,1}}  \Big\{ {w} (Z_i)   \Delta_\eta^{(2)}(Z_i)  - E_\Delta^{(2)} \Big\}  
- \frac{3}{2n} \sum_{i\in \cI_{\ell,2}}  
\Big\{ {w} (Z_i)   \Delta_\eta^{(1)}(Z_i)  -  E_\Delta^{(1)} \Big\} \\
&\qquad + \frac{3}{2m}\sum_{j\notin \cI_\ell^\new}   \Big\{ \Delta_\eta^{(1)}(Z_j^\new) + \Delta_\eta^{(2)}(Z_j^\new)
 - E_\Delta^{(1)} - E_\Delta^{(2)} \Big\}.
\$
Note that conditional on $\cI_{\ell}\cup \cI_\ell^\new \cup\cI_{\ell,2} \cup \cI_{\ell,2}^\new$, 
the random variables 
$\{w(Z_i)\Delta_\eta^{(2)}(Z_i)  - E_\Delta^{(2)}\}_{i\in \cI_{\ell,1}}$
are i.i.d.~and mean zero. Hence 
\$
& n\cdot \EE \bigg(\Big[\frac{3}{2n}\sum_{i\in \cI_{\ell,1}}  \big\{ {w} (Z_i)   \Delta_\eta^{(2)}(Z_i)  - E_\Delta^{(2)} \big\} \Big]^2 \bigggiven \cI_{\ell}\cup \cI_\ell^\new \cup \cI_{\ell,2}\cup \cI_{\ell,2}^\new \bigg)\\
&= 3/2 \cdot\big\|{w} (\cdot)   \Delta_\eta^{(2)}(\cdot) \big\|_{L_2(\mathbb{P})}^2
= o_P(1)
\$
by Assumption~\ref{assump:cov_est_rate}. 
Invoking Lemma~\ref{lem:cond_to_op} again 
with similar arguments for all other terms, 
we have 
\$
\big|\text{(ii)} \big| = o_P(1/\sqrt{n}+1/\sqrt{m}) = o_{P}(1/\sqrt{n}).
\$
Putting the two bounds together, we have 
\#\label{eq:cov_shift_dr_eq1}
\hat\theta_n^{\trans,(\ell)} - \frac{3}{2m} \sum_{j\notin \cI_\ell^\new} \eta(Z_j^\new)  
-
\frac{3}{2n}\sum_{i \notin \cI_\ell}  \hat{w}_{\ell}(Z_i) \big\{ \psi(D_i) - \eta(Z_i) \big\} = o_P(1/\sqrt{n}).
\#
Furthermore, note that $\EE[\psi(D_i) - \eta(Z_i) \given Z_i]=0$ almost surely, hence conditional on $\cI_\ell \cup \cI_\ell^\new$, 
the random variables
$\{\Delta_{w}(Z_i)[\psi(D_i) - \eta(Z_i) ]\}_{i\notin \cI_\ell}$
are i.i.d.~and mean zero. Thus
\$
&n\cdot \EE\bigg(\Big[\frac{3}{2n}\sum_{i \notin \cI_\ell}  \Delta_{w}(Z_i) \big\{ \psi(D_i) - \eta(Z_i) \big\}\Big]^2 \bigggiven \cI_\ell \cup \cI_\ell^\new \bigg) \\
& = 3/2\cdot \big\|\Delta_{w}(Z_i) \big\{ \psi(D_i) - \eta(Z_i) \big\}\big\|_{L_2(\mathbb{P})}^2 \\
&\leq  3/2 \cdot \sup_{z}\big|\hat{w}_\ell(z)-w(z)\big|^2 \cdot \big\| \psi(D_i) - \eta(Z_i) \big\|_{L_2(\mathbb{P})}^2 = o_P(1).
\$
The last equation follows from $\sup_{z}\big|\hat{w}_\ell(z)-w(z)\big| = o_P(1)$ in Assumption~\ref{assump:cov_est_rate} 
as well as the fact that 
$\psi(D_i)$ and $\eta(Z_i)$ 
both have finite $L_2(\mathbb{P})$ norms. 
Invoking Lemma~\ref{lem:cond_to_op}, we know 
\#\label{eq:cov_shift_dr_eq2}
\frac{3}{2n}\sum_{i \notin \cI_\ell}  \Delta_{w}(Z_i) \big\{ \psi(D_i) - \eta(Z_i) \big\} = o_P(1/\sqrt{n}). 
\#
Combining equations~\eqref{eq:cov_shift_dr_eq1} and~\eqref{eq:cov_shift_dr_eq2}, we have 
\$
\hat\theta_n^{\trans,(\ell)} - \frac{3}{2m} \sum_{j\notin \cI_\ell^\new} \eta(Z_j^\new)  
-
\frac{3}{2n}\sum_{i \notin \cI_\ell}  w(Z_i) \big\{ \psi(D_i) - \eta(Z_i) \big\} = o_P(1/\sqrt{n}).
\$
Recalling the sample splitting protocol, averaging over $\ell=1,2,3$, we thus have 
\$
\hat\theta_{m,n}^{\trans,\shift} - \frac{1}{ m} \sum_{j=1}^m \eta(Z_j^\new)  
-
\frac{1}{n}\sum_{i =1}^n  w(Z_i) \big\{ \psi(D_i) - \eta(Z_i) \big\} = o_P(1/\sqrt{n}),
\$
which (since $m\geq \epsilon n$ for some $\epsilon>0$) 
further leads to 
\$
\sqrt{n}(\hat\theta_{m,n}^{\trans,\shift} - \theta_m^{\cond,\new} )
= 
\frac{1}{\sqrt{n}}\sum_{i =1}^n  w(Z_i) \big\{ \psi(D_i) - \eta(Z_i) \big\} +  o_P(1).
\$
Finally, applying the conditional central limit theorem 
of Lemma~\ref{lem:cond_clt} 
to $w(Z_i)\{\psi(D_i)-\eta(Z_i)\}$ which 
has finite fourth moment, 
and invoking Lemma~\ref{lem:op_to_op}, 
it holds for any $x\in \RR$ that  
\$
\PP\Big\{ \sqrt{n}(\hat\theta_{m,n}^{\trans,\shift} - \theta_m^{\cond,\new} ) \leq x \Biggiven Z_{1:m}^\new, Z_{1:n}\Big\} = \Phi(x/\sigma_{\shift}) + o_P(1). 
\$
Since $\hat\sigma_{\shift}\to \sigma_{\shift}$ in probability, 
with exactly the same arguments as those in the proof of Theorem~\ref{thm:cond_intv}, we obtain the desired result 
in Theorem~\ref{thm:est_cov_shift}. 
\end{proof}

\section{Proof of extension results}

\subsection{Proof of fixed-attributes results} 
\label{app:subsec_fix}

\begin{proof}[Proof of Proposition~\ref{prop:linear_fixed}]
We first show that $\hat\theta_n \stackrel{P}{\to} \theta_n^\cond$. 
By the optimality of $\theta_n^\cond$, 
we have the first-order condition 
$\nabla L_n(\theta_n^\cond) = 0$, 
i.e., $\sum_{i=1}^n \EE\{\nabla \ell(D_i,\theta)\given z_i\}=0$ 
with the exchangeability of 
gradient and expectation 
under Assumption~\ref{assump:fixed}. 
By definition, we have 
$
\nabla^2 \hat{L}_n(\theta_n^\cond) - \nabla^2  {L}_n(\theta_n^\cond) = \frac{1}{n}\sum_{i=1}^n 
\big( \nabla^2 \ell(D_i,\theta_n^\cond) - \EE [\nabla^2 \ell(D_i,\theta_n^\cond)\given z_i ] \big),
$
which converges (elementwise) to zero by the law of 
large numbers. 
In particular, 
it holds with probability tending to 1 that
\#\label{eq:conv_hessian}
\nabla^2 \hat{L}_n(\theta_n^\cond) - \nabla^2  {L}_n(\theta_n^\cond) \succeq - c_2/2 \cdot \mathbf{I}_{p\times p}.
\#
Recalling condition (iv), 
we now invoke Lemma~\ref{lem:convex} on the event 
that~\eqref{eq:conv_hessian} holds, 
and take $f=\hat{L}_n $, 
$\theta=\hat\theta_n$, $\theta_0 = \theta_n^\cond$, 
and $\lambda=c_2/2$, $c=c_1$. 
By the fact that $\hat{L}_n(\hat\theta_n)\leq \hat{L}_n(\theta_n^\cond)$, we have 
\$
\min\big\{ \|\hat\theta_n-\theta_n^\cond\|^2, c_1 \|\hat\theta_n-\theta_n^\cond\| \big\} 
&\leq \frac{2}{\lambda}\big\{ \hat{L}_n(\hat\theta_n) 
- \hat{L}_n(\theta_n^\cond) -\nabla \hat{L}_n(\theta_n^\cond)^\top (\hat\theta_n-\theta_n^\cond) \big\}\\ 
&\leq   \frac{2}{\lambda}  \big\| \nabla \hat{L}_n(\theta_n^\cond) \big\| \cdot  \|\hat\theta_n-\theta_n^\cond\|.
\$
We further note by the law of large numbers that  
$\nabla \hat{L}_n(\theta_n^\cond) - \nabla L_n(\theta_n^\cond)
= o_P(1)$ with entrywise convergence. 
Thus, $\|\nabla \hat{L}_n(\theta_n^\cond)\|=o_P(1)$, 
and the above inequality leads to $\|\hat\theta_n-\theta_n^\cond\| = o_P(1)$. 

We then use Taylor expansion of $\nabla \hat L_n$ 
around $\hat \theta_n$ to show the 
asymptotic linearity. 
As $\nabla \hat L_n(\hat\theta_n)= 0$, 
\#\label{eq:taylor_fix}
- \nabla \hat{L}_n(\theta_n^\cond)
= \nabla^2 \hat L_n(\tilde\theta_n) (\hat\theta_n - \theta_n^\cond),
\#
where $\tilde\theta_n$ lies on the 
segment between $\hat\theta_n$ and $\theta_n^\cond$. 
Thus, condition (iii) in Assumption~\ref{assump:fixed} implies  
\$
\big\|\nabla^2 \hat L_n(\tilde\theta_n) 
- \nabla^2 \hat L_n(\theta_n^\cond) \big\|_{\text{op}}
\leq \frac{1}{n}\sum_{i=1}^n m_n(D_i)^2 
\|\tilde\theta_n - \theta_n^\cond\| 
\leq \frac{1}{n}\sum_{i=1}^n m_n(D_i)^2  
\|\hat\theta_n - \theta_n^\cond\| = o_P(1). 
\$
Also, the law of large numbers implies 
$\|\nabla^2 L_n(\theta_n^\cond) - \nabla^2 \hat{L}_n(\theta_n^\cond) \|_{\text{op}}=o_P(1)$, 
which further implies 
$ \|\nabla^2 \hat L_n(\tilde\theta_n) 
- \nabla^2   L_n(\theta_n^\cond)  \|_{\text{op}}= o_P(1)$.
Combining this  fact with~\eqref{eq:taylor_fix}, 
we have 
\$
- \frac{1}{\sqrt{n}}\sum_{i=1}^n \nabla \ell(D_i,\theta_n^\cond) = \big\{\nabla^2   L_n(\theta_n^\cond) +o_P(1)\big\} \cdot \sqrt{n}(\hat\theta_n - \theta_n^\cond),
\$
which completes the proof of Proposition~\ref{prop:linear_fixed}.
\end{proof}

\begin{proof}[Proof of Theorem~\ref{thm:fix}]
By Proposition~\ref{prop:linear_fixed}, one has 
\$
\Sigma_n^{-1/2} \sqrt{n}(\hat\theta_n - \theta_n^\cond)
= - \Var\big\{\sqrt{n}\,\nabla \hat{L}_n(\theta_n^\cond) \big\}^{-1/2} 
 \frac{1}{\sqrt{n}}\sum_{i=1}^n\nabla \ell(D_i,\theta_n^\cond)+o_P(1).
\$
Note that $\nabla \ell(D_i,\theta_n^\cond)$ 
are mutually independent. 
By condition (v) in Assumption~\ref{assump:fixed} 
and invoking the Lyapunov's Central Limit Theorem~\citep{billingsley1995probability}, we obtain 
the asymptotic normal distribution 
and completes the proof of Theorem~\ref{thm:fix}. 
\end{proof}

\begin{proof}[Proof of Proposition~\ref{prop:sigma_consist}]
We first show that $\hat{M}^{-1} = \big[\frac{1}{n}\sum_{i=1}^n \EE\{\ddot{\ell}(D_i,\theta_n^\cond)\}\big]^{-1} + o_P(1)$. 
By condition (iii) in Assumption~\ref{assump:fixed}, we have 
$
\big| \hat{M} - \frac{1}{n}\sum_{i=1}^n \ddot{\ell}(D_i, \theta_n^\cond) \big|
\leq \frac{1}{n}\sum_{i=1}^n m_n(D_i) \cdot |\hat\theta_n - \theta_n^\cond| = o_P(1)
$
since $|\hat\theta_n - \theta_n^\cond| = o_P(1)$ from Proposition~\ref{prop:linear_fixed}. 
The law of large numbers thus implies 
the desired result by noting that  
$\big| \frac{1}{n}\sum_{i=1}^n \EE\big\{\ddot{\ell}(D_i,\theta_n^\cond)\big\} - \frac{1}{n}\sum_{i=1}^n \ddot{\ell}(D_i, \theta_n^\cond) \big|=o_P(1)$. Hence $\hat{M}^{-1} =\big[\frac{1}{n}\sum_{i=1}^n \EE \{\ddot{\ell}(D_i,\theta_n^\cond) \}\big]^{-1}+o_P(1)$ as 
the latter is strictly positive definite by condition (iv). 

Furthermore, by condition (iii) of Assumption~\ref{assump:fixed} as well as Proposition~\ref{prop:linear_fixed}, we have 
\$
&\frac{1}{n}\sum_{i=1}^n \big\{ \hat{s}(D_i)   - \dot{\ell}(D_i,\theta_n^\cond) \big\}^2 
=\frac{1}{n}\sum_{i=1}^n \big\{ \dot{\ell}(D_i,\hat\theta_n )   - \dot{\ell}(D_i,\theta_n^\cond) \big\}^2  \\
&\leq \frac{1}{n}\sum_{i=1}^n m_n(D_i)^2 \cdot (\hat\theta_n -\theta_n^\cond)^2 = O_P\big(  |\hat\theta_n -\theta_n^\cond|^2 \big) = o_P(1).
\$
As $\frac{1}{n}\sum_{i=1}^n \big\{\hat{t}(z_i) - \mu(z_i)\big\}^2 = o_P(1)$ 
for a fixed function $\mu\colon \mathbb{Z}\to \RR$, 
Cauchy-Schwarz inequality implies 
$
\frac{1}{n}\sum_{i=1}^n \big\{ \hat{s}(D_i)  - \hat{t}(z_i) - \dot{\ell}(D_i,\theta_n^\cond) + \mu(z_i)\big\}^2 =o_P(1),
$
which further leads to 
\$
\frac{1}{n}\sum_{i=1}^n \big\{ \hat{s}(D_i)  - \hat{t}(z_i)\big\}^2 
= \frac{1}{n}\sum_{i=1}^n \big\{ \dot{\ell}(D_i,\theta_n^\cond) - \mu(z_i) \big\}^2 +o_P(1).
\$
Writing $\mu^*(z_i)=\EE\{\dot{\ell}(D_i,\theta_n^\cond)\given z_i\}$,  condition (v) in Assumption~\ref{assump:fixed} and Markov's inequality implies
\$
\frac{1}{n}\sum_{i=1}^n \big\{ \dot{\ell}(D_i,\theta_n^\cond) - \mu^*(z_i) \big\}^2 
= \frac{1}{n}\sum_{i=1}^n \EE\big[ \big\{ \dot{\ell}(D_i,\theta_n^\cond) - \mu^*(z_i) \big\}^2 \big]  +o_P(1),
\$
where 
$\EE\big[ \{\dot{\ell}(D_i,\theta_n^\cond) - \mu^*(z_i) \}^2 \big] = \Var\big\{\dot{\ell}(D_i,\theta_n^\cond) \given z_i\big\}$. 
Finally, we have 
$\frac{2}{n}\sum_{i=1}^n  \{ \dot{\ell}(D_i,\theta_n^\cond) - \mu^*(z_i)  \} \cdot  \{ \mu(z_i) - \mu^*(z_i) \} = o_P(1)$ 
since each term in 
the summation is mean zero and the moment condition (v) holds. 
Thus
\$
\frac{1}{n}\sum_{i=1}^n \big\{ \dot{\ell}(D_i,\theta_n^\cond) - \mu(z_i) \big\}^2
&= \frac{1}{n}\sum_{i=1}^n \big\{ \dot{\ell}(D_i,\theta_n^\cond) - \mu^*(z_i) \big\}^2 + \frac{1}{n}\sum_{i=1}^n \big\{ \mu(z_i) - \mu^*(z_i) \big\}^2  \\
&= \frac{1}{n}\sum_{i=1}^n \Var\big\{\dot{\ell}(D_i,\theta_n^\cond) \given z_i\big\}+ \frac{1}{n}\sum_{i=1}^n \big\{ \mu(z_i) - \mu^*(z_i) \big\}^2  +o_P(1).
\$
Combining with the consistency of $\hat{M}^{-1}$, 
we know that 
$\hat\sigma_n  - \sigma_n = o_P(1)$ 
if $\mu^*=\mu$, and  
$\hat\sigma_n  - \tilde\sigma_n = o_P(1)$ 
for some $\tilde\sigma_n \geq \sigma_n$ otherwise.  
This completes the proof of Proposition~\ref{prop:sigma_consist}. 
\end{proof}

\subsection{Proof of results for conditioning on unobserved attributes} 
\label{app:subsec_cond_unobs}

\begin{proof}[Proof of Theorem~\ref{thm:cond_unobs}]
Following exactly the same arguments as Proposition~\ref{prop:linear_exp}, we obtain the 
linear expansion in Assumption~\ref{assump:linear_main} 
with $Z$ replaced by $X$, which implies 
\$
\hat\theta_n - \theta_n^\cond(X_{1:n}) = \frac{1}{n}\sum_{i=1}^n 
\big[ \phi(D_i) - \EE\{\phi(D_i)\given X_i\}\big] + o_P(1/\sqrt{n}).
\$
Using the same arguments in 
the proof of Proposition~\ref{prop:cond_distr}, 
we know that for any fixed $x\in \RR$, the random variable 
    $ 
    \PP [  \sqrt{n}\{\hat \theta_n - \theta_n^{\cond}(X_{1:n})\} \le x \given X_{1:n} ]
    $ 
converges in probability to  $\Phi(x/\sigma_X)$, 
where  $\Phi$ is the  c.d.f.~of standard Gaussian distribution, 
and $\sigma_X:= \EE([ \phi(D)-\EE\{\phi(D)\given X\}]^2)$.

We then slightly modify the proof of 
Theorem~\ref{thm:cond_intv} to show the desired results. 
To be specific,  
for any fixed $x\in \RR$ 
and any fixed constant $\epsilon>0$, we write $z^-(\epsilon) = z_{1-\alpha/2}(\sigma_Z-\epsilon)$
and $z^+(\epsilon) = z_{1-\alpha/2}(\sigma_Z+\epsilon)$. 
Denoting 
$
\Delta^\pm(\epsilon) = \PP\big\{ \sqrt{n} |\hat\theta_n  - \theta_n^{\cond}(X_{1:n})  | \leq z^\pm(\epsilon) \biggiven X_{1:n}\big\} - [2\Phi (\{z^\pm(\epsilon)/\sigma_X \}-1],
$
we have $\Delta^+(\epsilon), \Delta^-(\epsilon)=o_P(1)$ 
similar to 
Proposition~\ref{prop:cond_distr}. 
Since $\hat\sigma_Z$ converges in 
probability to $\sigma_Z$, 
following similar arguments as~\eqref{eq:pval1} 
and~\eqref{eq:pval2} in the proof of 
Theorem~\ref{thm:cond_intv}, we have
\#\label{eq:pval11}
 &\PP\Big\{ \sqrt{n}\big|\hat\theta_n  - \theta_n^{\cond}(X_{1:n}) \big|\leq z_{1-\alpha/2}\cdot \hat\sigma_Z \Biggiven X_{1:n}\Big\}  \\
 &\geq 2 \Phi\Big({\textstyle \frac{z^-(\epsilon)}{\sigma_X}}\Big)-1  + \PP ( \hat\sigma_z < \sigma_z-\epsilon \given X_{1:n} ) +   \Delta^-(\epsilon),
\#
where  
$\PP ( \hat\sigma_Z < \sigma_Z-\epsilon \given X_{1:n} )=o_{P}(1)$ by Lemma \ref{lem:op_to_op}. 
On the other hand, we similarly have 
\#\label{eq:pval22}
 &\PP\Big\{ \sqrt{n}\big|\hat\theta_n  - \theta_n^{\cond}(X_{1:n}) \big|\leq z_{1-\alpha/2}\cdot \hat\sigma_Z \Biggiven X_{1:n}\Big\} \\
&\leq 2 \Phi\Big({\textstyle \frac{z^+(\epsilon)}{\sigma_X}}\Big) -1   + \PP ( \hat\sigma_Z < \sigma_Z+\epsilon \given X_{1:n} ) +   \Delta^+(\epsilon),
\#
where $\PP ( \hat\sigma_Z > \sigma_Z +\epsilon \given X_{1:n} )=o_{P}(1)$ by Lemma \ref{lem:op_to_op}.  
For any fixed $\delta>0$, we can choose some fixed $\epsilon>0$ such that 
$2\Phi(z^-(\epsilon)/\sigma_X)-1-(1- \beta) >-\delta/2$ and $2\Phi(z^+(\epsilon)/\sigma_X)-1- (1-\beta)<\delta/2$ for 
\$
\beta = 1  - \Phi\Big( \frac{\sigma_Z}{\sigma_X} \cdot z_{1-\alpha} \Big)  
\$
Combining the above, 
we have 
$
\PP\big[ | \PP\{ \sqrt{n}|\hat\theta_n  - \theta_n^{\cond}(X_{1:n}) |\leq z_{1-\alpha/2}\cdot \hat\sigma_Z \given X_{1:n}\} - (1-\beta) | >\delta\big] \to 0,
$
hence prove the convergence in probability. 
Finally, we note that 
$\sigma_Z^2 = \Var\{\phi(D)\}- \EE([\EE\{\phi(D)\given Z\}]^2)$ 
and $\sigma_X^2 = \Var\{\phi(D)\}- \EE([ \EE\{\phi(D)\given X\}]^2)$. 
If $\EE([ \EE\{\phi(D)\given X\}]^2)\geq \EE([ \EE\{\phi(D)\given Z\}]^2)$, 
we have $\sigma_Z \geq \sigma_X$ hence $\beta \leq \alpha$, 
thus completing the proof of Theorem~\ref{thm:cond_unobs}.  
\end{proof}

\subsection{Proof of transferring to subsets} 
\label{app:proof_trans_subset}

\begin{proof}[Proof of Theorem~\ref{thm:trans_subset}]
In any of the setups in Theorem~\ref{thm:trans_subset}, 
we have shown that (c.f.~the respective proofs)
\$
\hat\theta_{m,n}^\trans  - \theta_0^\new 
= \frac{1}{n}\sum_{i=1}^n w(Z_i) \big\{\psi(D_i)-\eta(Z_i)\big\} + 
\frac{1}{m} \sum_{j=1}^m \eta(Z_j^\new)+ o_P(1/\sqrt{n}). 
\$
Under regularity conditions that are similar to 
Proposition~\ref{prop:main_new_lin_exp}, 
we have the following asymptotic linear expansion 
of $\theta_m^\cond(X_{1:m}^\new)$ that is similar 
to Assumption~\ref{assump:linear_expansion_known_shift}: 
\$
\theta_m^\cond(X_{1:m}^\new)  - \theta_0^\new 
= \frac{1}{m} \sum_{j=1}^m \EE\big\{\psi(D_j^\new)\biggiven X_j^\new\big\} + o_P(1/\sqrt{m}). 
\$
Here since $X\subset Z$, 
we note by the tower property of conditional expectations that almost surely 
\$
\EE\big\{\psi(D_j^\new)\biggiven X_j^\new\big\}
= \EE\big\{\eta(Z_j^\new)\biggiven X_j^\new\big\}.
\$
Combining the above results, we have 
\$
\hat\theta_{m,n}^\trans - \theta_m^\cond(X_{1:m}^\new)
&=  \frac{1}{n}\sum_{i=1}^n w(Z_i) \big\{\psi(D_i)-\eta(Z_i)\big\} \\ 
&\qquad + 
\frac{1}{m} \sum_{j=1}^m \big[ \eta(Z_j^\new)
- \EE\big\{\eta(Z_j^\new)\biggiven X_j^\new\big\} \big]+ o_P(1/\sqrt{n}).
\$
Each term in the above summation is 
mean zero conditional on $X_{1:m}^\new$. 
Thus, applying the conditional CLT in Lemma~\ref{lem:cond_clt} to the filtration
$\cF_k = \sigma(\{X_j^\new\}_{i=1}^k)$,  
and dealing with the $o_P(1/\sqrt{n})$ 
term similar to the proof of Theorem~\ref{thm:cond_intv}, 
we complete the proof of Theorem~\ref{thm:trans_subset}. 
\end{proof}

\section{Proofs of estimation} \label{app:est}

\subsection{Proof of Proposition~\ref{prop:consist_hat_varphi}}
\label{app:est_varphi}

\begin{proof}[Proof of Proposition~\ref{prop:consist_hat_varphi}]
We first analyze the entry-wise error in $\hat{M}$.
Note that 
\$
\hat{M} - M(\theta) &= \hat{M}(s,\mathbf{1},\hat\theta,\cI_2) - M(s,\mathbf{1},\theta_0)\\
&= \hat{M}(s,\mathbf{1},\hat\theta,\cI_2) -  M(s,\mathbf{1},\hat\theta)   
+M(s,\mathbf{1},\hat\theta)    - M(s,\mathbf{1},\theta_0).
\$
On the other hand, by Assumption~\ref{assump:meta_matrix}, we have 
\$
\|\hat{M}(s,\mathbf{1},\hat\theta,\cI_2) -  M(s,\mathbf{1},\hat\theta)   \|_\infty 
\leq O_P\big\{\cR_m(|\cI_2|)\big\} = O_P\big\{\cR_m(|\cI|)\big\},
\$
and $\|M(s,\mathbf{1},\hat\theta)    - M(s,\mathbf{1},\theta_0)\|_\infty = O(\|\hat\theta-\theta_0\|_2) = O( |\cI|^{-1/2})$. Hence 
\#\label{eq:eta_bound1}
\|\hat{M} - M(\theta_0)\|_\infty \leq O_P\big\{\cR_m(|\cI|) +  |\cI|^{-1/2} \big\}.
\#
By Assumption~\ref{assump:meta_reg}, we have 
$\big\|\mathcal{G}(\hat{s},\cI_2) - \mathcal{G}(\hat{s})\big\|_{L_2(\mathbb{P})}= O_P\big\{\mathcal{R}_r(|\cI|)\big\}$. Meanwhile, 
writing $\mathcal{G}(s) = \EE\{s(D,\theta_0)\given Z=\cdot\}$, 
by the definition of $\cG(\cdot)$, we have 
\$
\big\|\mathcal{G}(\hat{s}) - \mathcal{G}(s) \big\|_{L_2(\mathbb{P})} 
&= \Big\|\EE\big\{s(D,\hat\theta)-s(D,\theta_0)\given Z = \cdot\} \Big\|_{L_2(\mathbb{P})}  \\
&\leq \big\| s(\cdot,\hat\theta)-s(\cdot,\theta_0) \big\|_{L_2(\mathbb{P})}
= O(\|\hat\theta-\theta_0\|_2) = O_P( |\cI|^{-1/2}),
\$
where $\big\| s(\cdot,\hat\theta)-s(\cdot,\theta) \big\|_{L_2(\mathbb{P})}$
views $\hat\theta$ as fixed and the $L_2$-norm is with respect to $D\sim \mathbb{P}$. 
Putting them together, the estimated conditional mean function satisfies 
\#\label{eq:eta_bound2}
\big\|\hat{t}(\cdot) - \cG(s)(\cdot)\big\|_{L_2(\mathbb{P})}
&\leq \big\|\mathcal{G}(\hat{s},\cI_2) - \mathcal{G}(\hat{s})\big\|_{L_2(\mathbb{P})}
+ \big\|\mathcal{G}(\hat{s}) - \mathcal{G}(s)\big\|_{L_2(\mathbb{P})} \notag \\
&\leq O_P\big\{\mathcal{R}_r(|\cI|) + |\cI|^{-1/2}\big\}.
\#
Altogether, we have 
\$
\big\|\hat\varphi(\cdot) - \varphi(\cdot) \big\|_{L_2(\mathbb{P})}
&= \big\| \hat{M} \hat{t}(\cdot) - M(\theta_0) \cG(s)(\cdot) \big\|_{L_2(\mathbb{P})}\\
&\leq \Big\| \big\{\hat{M}  - M(\theta_0) \big\}\hat{t}(\cdot)  \Big\|_{L_2(\mathbb{P})}
+ \Big\|  M(\theta_0) \big\{\hat{t}(\cdot) -  \cG(s)(\cdot) \big\} \Big\|_{L_2(\mathbb{P})} \\
&\leq p\cdot \|\hat{M} - M(\theta_0)\|_\infty  \cdot \big\|\hat{t}(\cdot)\big\|_{L_2(\mathbb{P})} 
+ p\cdot  \|M(\theta_0)\|_\infty \cdot \big\|\hat{t}(\cdot) -  \cG(s)(\cdot) \big\|_{L_2(\mathbb{P})} \\
&\leq  p\cdot O_P\big\{ \cR_m(|\cI|) + \mathcal{R}_r(|\cI|) + |\cI|^{-1/2} \big\},
\$
where the last inequality follows from~\eqref{eq:eta_bound1} and~\eqref{eq:eta_bound2} 
and the fact that 
\$
\big\|\hat{t}(\cdot)\big\|_{L_2(\mathbb{P})} \leq 
\big\|\mathcal{G}(s)\big\|_{L_2(\mathbb{P})} + O_P\big\{\mathcal{R}_r(|\cI|) + |\cI|^{-1/2}\big\} = O_P(1).
\$
We thus complete the proof of Proposition~\ref{prop:consist_hat_varphi}.
\end{proof}

\subsection{Proof of Proposition~\ref{prop:consist_sigma}}
\label{app:est_sigma}

\begin{proof}[Proof of Proposition~\ref{prop:consist_sigma}]
For simplicity, we denote $\Delta\phi_i = \phi(D_i)-\hat\phi_i$ and 
$\Delta\varphi_i = \varphi(Z_i) - \hat\varphi_i$, where 
$\hat\phi_i$ and $\hat\varphi_i$ 
are estimated in Algorithm~\ref{alg:sigma_est}. 
Firstly, by Cauchy-Schwarz inequality, 
\$
\frac{1}{|\cI_2|}\sum_{i\in \cI_2} \Delta\phi_i^2 
&= \frac{1}{|\cI_2|}\sum_{i\in \cI_2}\big\{\hat{M}s(D_i,\hat\theta) - M s(D_i,\theta_0)\big\}^2 \\
&= \frac{1}{|\cI_2|}\sum_{i\in \cI_2}\big\{\hat{M}s(D_i,\hat\theta) - \hat{M}s(D_i, \theta_0)+  \hat{M}s(D_i, \theta_0)- M s(D_i,\theta_0)\big\}^2 \\
&\leq \frac{2}{|\cI_2|}\sum_{i\in \cI_2}\big\{\hat{M}s(D_i,\hat\theta) - \hat{M}s(D_i, \theta_0) \big\}^2 +  \frac{2}{|\cI_2|}\sum_{i\in \cI_2}\big\{\hat{M}s(D_i, \theta_0)- M s(D_i,\theta_0)\big\}^2.
\$
Here since $\hat\theta$ is independent of $\cI_2$, we have 
\$
\EE\bigg[\frac{1 }{|\cI_2|}\sum_{i\in \cI_2}\big\{ s(D_i,\hat\theta) -  s(D_i, \theta_0) \big\}^2  \bigggiven \cI_1\bigg] 
= \big\|s(\cdot,\hat\theta) -s(\cdot, \theta_0)\big\|_{L_2(\mathbb{P})}^2
= O\big(\|\hat\theta-\hat\theta_0\|_2\big) = o_P(1).
\$
Employing Lemma~\ref{lem:op_to_op}, we have 
\$
\frac{2}{|\cI_2|}\sum_{i\in \cI_2} \hat{M}s(D_i,\hat\theta)
= 2\hat{M}^2 \cdot \frac{1 }{|\cI_2|}\sum_{i\in \cI_2}\big\{ s(D_i,\hat\theta) -  s(D_i, \theta_0) \big\}^2  = o_P(1). 
\$
Following the same arguments as in the 
proof of Proposition~\ref{prop:consist_hat_varphi}, 
we have $\hat{M}=M+o_P(1)$, hence 
\$
\frac{2}{|\cI_2|}\sum_{i\in \cI_2}\big\{\hat{M}s(D_i, \theta_0)- M s(D_i,\theta_0)\big\}^2 = 2(\hat{M}-M)^2 \cdot \frac{1}{|\cI_2|}\sum_{i\in \cI_2} s(D_i, \theta_0) ^2 = o_P(1).
\$
Thus $\frac{1}{|\cI_2|}\sum_{i\in \cI_2} \Delta\phi_i^2  = o_P(1)$. 
On the other hand, by the construction, $\hat\varphi$ is independent of $\cI_2$, hence by Proposition~\ref{prop:consist_hat_varphi}, we have 
\$
\EE\bigg[\frac{1}{|\cI_2|}\sum_{i\in \cI_2} \Delta\phi_i^2  \bigggiven \cI_1\bigg] 
= \big\|\hat\varphi - \varphi(\cdot)\big\|_{L_2(\mathbb{P})} = o_P(1),
\$
which, combined with Lemma~\ref{lem:op_to_op}, leads to 
$\frac{1}{|\cI_2|}\sum_{i\in \cI_2} \Delta\varphi_i^2  = o_P(1)$. 
Therefore, by Cauchy-Schwarz inequality, we have 
\$
\frac{1}{|\cI_2|}\sum_{i\in \cI_2} \big(   \Delta\phi_i -   \Delta\varphi_i\big)^2 \leq \frac{2}{|\cI_2|}\sum_{i\in \cI_2} \Delta\phi_i^2 
+\frac{2}{|\cI_2|}\sum_{i\in \cI_2} \Delta\varphi_i^2 = o_P(1). 
\$
Finally, by Algorithm~\ref{alg:sigma_est} and Cauchy-Schwarz inequality, 
\$
\hat\sigma^2 &= \frac{1}{|\cI_2|}\sum_{i\in \cI_2} \big\{ \phi(D_i) - \Delta\phi_i - \varphi(Z_i) + \Delta\varphi_i\big\}^2 \\
&\leq \frac{1}{|\cI_2|}\sum_{i\in \cI_2} \big\{ \phi(D_i) -  \varphi(Z_i) \big\}^2 + \frac{1}{|\cI_2|}\sum_{i\in \cI_2} \big(   \Delta\phi_i -   \Delta\varphi_i\big)^2\\
&\qquad + 2\sqrt{\frac{1}{|\cI_2|}\sum_{i\in \cI_2} \big\{ \phi(D_i) -  \varphi(Z_i) \big\}^2 } \cdot \sqrt{\frac{1}{|\cI_2|}\sum_{i\in \cI_2} \big(   \Delta\phi_i -   \Delta\varphi_i\big)^2} = \sigma^2 +o_P(1),
\$
where the last equality follows from the law of large numbers.
\end{proof}

\subsection{Proof of Proposition~\ref{prop:consist_eta}}
\label{app:est_eta}

\begin{proof}[Proof of Proposition~\ref{prop:consist_eta}]
To begin with, we write 
\$
\mathcal{G}(s) = \EE\{s(D,\theta_0^\new)\given Z=\cdot\} = \EE\{s(D^\new,\theta_0^\new)\given Z^\new=\cdot\}
\$
and for any fixed $\theta\in\Theta$,
\$
M(s,w,\theta ) = - \big[\EE\{w(Z)\dot{s}(D ,\theta )\} \big]^{-1},
\$
so that the ground truth satisfies 
$
\eta(z) = M(s,w,\theta_0^\new) \cG(s)(z).
$

We first prove the result with ground truth of $w(\cdot)$. 
In this case, with regularity conditions we know $\|\hat\theta-\theta_0^\new\|_2 = O_P(|\cI_2|^{-1/2}) = O_P(|\cI|^{-1/2})$. Thus, 
following exactly the same arguments as in the proof of Proposition~\ref{prop:consist_hat_varphi}, we have 
\$
\big\|\hat{t}(\cdot) - \cG(s)(\cdot)\big\|_{L_2(\mathbb{P})}
&\leq O_P\big\{\mathcal{R}_r(|\cI|) + |\cI|^{-1/2}\big\}.
\$
On the other hand, by Algorithm~\ref{alg:eta_est}, we know 
\$
\big\|\hat{M} - M(s,w,\theta_0^\new)\big\|_{\infty} &= \big\|\hat{M}(s,w,\hat\theta,\cI_3) - M(s,w,\theta_0^\new) \big\|_{\infty} \\
& \leq \big\|\hat{M}(s,w,\hat\theta,\cI_3) - M(s,w,\hat\theta) \big\|_{\infty} 
+ \big\| {M}(s,w,\hat\theta ) - M(s,w,\theta_0^\new) \big\|_{\infty}.
\$
Since $\hat\theta$ is independent of $\cI_3$, by Assumption~\ref{assump:meta_matrix}, 
we have 
\$
\big\|\hat{M}(s,w,\hat\theta,\cI_3) - M(s,w,\hat\theta) \big\|_{\infty} \leq O_P\big\{\mathcal{R}_m(|\cI_3|)\big\} = O_P\big\{\mathcal{R}_m(|\cI|)\big\}.
\$
The given conditions also imply
\$
\big\| {M}(s,w,\hat\theta ) - M(s,w,\theta_0^\new) \big\|_{\infty}
\leq O\big(\|\hat\theta - \theta_0^\new\|_{2}\big) = O_P\big(|\cI_2|^{-1/2}\big)
= O_P\big(|\cI|^{-1/2}\big).
\$
Putting them together, we have 
\$
\big\|\hat{M} - M(s,w,\theta_0^\new)\big\|_{\infty} \leq O_P\big\{\mathcal{R}_m(|\cI|) + |\cI|^{-1/2}\big\}.
\$
Following the same arguments as in the proof 
of Proposition~\ref{prop:consist_hat_varphi}, 
we obtain the desired result~\eqref{eq:consist_eta_w}. 

We now consider the result for estimated $\hat{w}(\cdot)$. 
Since it is obtained from $\cI_1$, it is independent of 
subsequent estimation steps. 
By similar regularity conditions as Proposition~\ref{prop:linear_hat_theta_est},
we know that 
\$
\|\hat\theta - \theta_0^\new\|_2 = \bigg\|\frac{1}{|\cI_2|}\sum_{i\in\cI_2}\hat{w}(Z_i)\psi(D_i)\bigg\|_2+o_P(|\cI_2|^{-1/2})
= O_P(|\cI_2|^{-1/2}) = O_P(|\cI|^{-1/2}).
\$
With estimated $\hat{w}$, by Algorithm~\ref{alg:eta_est}, we know 
\$
\big\|\hat{M} - M(s,w,\theta_0^\new)\big\|_{\infty} &= \big\|\hat{M}(s,\hat{w},\hat\theta,\cI_3) - M(s, {w},\theta_0^\new) \big\|_{\infty} \\
& \leq \big\|\hat{M}(s,w,\hat\theta,\cI_3) - M(s,\hat{w},\hat\theta) \big\|_{\infty} 
+ \big\| M(s,\hat{w},\hat\theta) - M(s,w,\theta_0^\new) \big\|_{\infty}.
\$
Here by Assumption~\ref{assump:meta_matrix}, 
since $\cI_3$ is independent of $\hat{w}$ and $\hat\theta$, 
the estimation error is bounded as 
\$
\big\|\hat{M}(s,w,\hat\theta,\cI_3) - M(s,\hat{w},\hat\theta) \big\|_{\infty} 
\leq O_P\big\{\cR_m(|\cI_3|)\big\} =  O_P\big\{\cR_m(|\cI|)\big\}.
\$
By the stability assumptions of $M(s,w,\theta)$, we have 
\$
&\big\| M(s,\hat{w},\hat\theta) - M(s,w,\theta_0^\new) \big\|_{\infty}\notag \\
&\leq \big\| M(s,\hat{w},\hat\theta) - M(s, {w},\hat\theta) \big\|_{\infty}
+ \big\| M(s, {w},\hat\theta)- M(s,w,\theta_0^\new) \big\|_{\infty}\notag \\
&\leq O_P\big\{ \|\hat{w}(\cdot)-w(\cdot )\|_{L_2(\mathbb{P})}\big\}
+ O_P\big(\|\hat\theta - \theta_0^\new\|_{ 2 }\big) \notag \\
&\leq O_P\big\{\|\hat{w}(\cdot)-w(\cdot )\|_{L_2(\mathbb{P})} + |\cI|^{-1/2}\big\},
\$
hence 
\#\label{eq:eta_M}
\big\|\hat{M} - M(s,w,\theta_0^\new)\big\|_{\infty}
\leq O_P\big\{\|\hat{w}(\cdot)-w(\cdot )\|_{L_2(\mathbb{P})} + \cR_m(|\cI|) + |\cI|^{-1/2}\big\}
\#
On the other hand, 
since $\cI_3$ is independent of the function $\hat{s}(\cdot) = s(\cdot,\hat\theta)$, 
we know $\big\|\mathcal{G}(\hat{s},\cI_3)(\cdot) -  \mathcal{G}(\hat{s})(\cdot) \big\|_{L_2(\mathbb{P})} \leq O_P\{\cR_r(|\cI_3|)\}= O_P\{\cR_r(|\cI|)\}$. 
Also, the stability of $s(\cdot,\theta)$ implies 
\$
\big\|\mathcal{G}(\hat{s})(\cdot)- \cG(s)(\cdot)\big\|_{L_2(\mathbb{P})}
\leq \big\| \hat{s} (\cdot,\hat\theta)- s(\cdot,\theta)\big\|_{L_2(\mathbb{P})}
= O(\|\hat\theta-\theta_0\|_2) = O_P\big(|\cI|^{-1/2}\big).
\$
Therefore, the error in $\hat{t}(\cdot)$ can be bounded as 
\#\label{eq:eta_t}
\big\|\hat{t}(\cdot) - \cG(s)(\cdot)\big\|_{L_2(\mathbb{P})}
&\leq \big\|\mathcal{G}(\hat{s},\cI_3)(\cdot) -  \mathcal{G}(\hat{s})(\cdot) \big\|_{L_2(\mathbb{P})} + \big\|\mathcal{G}(\hat{s})(\cdot)- \cG(s)(\cdot)\big\|_{L_2(\mathbb{P})}\notag \\
&\leq  O_P\big\{\|\hat{w}(\cdot)-w(\cdot )\|_{L_2(\mathbb{P})} + |\cI|^{-1/2}\big\}.
\#
Following similar arguments as the case 
with ground truth of $w(\cdot)$, 
we combine~\eqref{eq:eta_M},~\eqref{eq:eta_t}, 
and obtain 
\$
\big\| \eta(s,\cI)(\cdot) - \eta(\cdot) \big\|_{L_2(\mathbb{P})}
&\leq \Big\| \big\{\hat{M}  - M(s,w,\theta_0^\new) \big\}\hat{t}(\cdot)  \Big\|_{L_2(\mathbb{P})}
+ \Big\|  M(s,w,\theta_0^\new) \big\{\hat{t}(\cdot) -  \cG(s)(\cdot) \big\} \Big\|_{L_2(\mathbb{P})} \\
&\leq p\cdot \|\hat{M} - M(s,w,\theta_0^\new)\|_\infty  \cdot \big\|\hat{t}(\cdot)\big\|_{L_2(\mathbb{P})} \\
&\quad + p\cdot  \|M(s,w,\theta_0^\new)\|_\infty \cdot \big\|\hat{t}(\cdot) -  \cG(s)(\cdot) \big\|_{L_2(\mathbb{P})} \\
&\leq  p\cdot O_P\big\{\|\hat{w}(\cdot)-w(\cdot )\|_{L_2(\mathbb{P})} + \cR_m(|\cI|) + \mathcal{R}_r(|\cI|) + |\cI|^{-1/2} \big\},
\$
which completes the proof of Proposition~\ref{prop:consist_eta}.
\end{proof}

\subsection{Proof of Proposition~\ref{prop:est_sigma_shift}}
\label{app:est_sigma_shift}

\begin{proof}[Proof of Proposition~\ref{prop:est_sigma_shift}]
Firstly, we write $M = -[\EE\{w(Z)\dot{s}(D,\theta_0^\new)\}]^{-1}$, 
so that $\psi(d) = Ms(d,\theta_0^\new)$. 
Following the same arguments as in the proof of Proposition~\ref{prop:consist_eta}, 
we have $\hat{M}=M+o_P(1)$ 
under the diminishing rate of $\cR_m(|\cI|)\to 0$ as $|\cI|\to \infty$.
By the regularity conditions, 
we have $\|\hat\theta - \theta_0^\new \|_2= o_P(1)$. 
Writing $\Delta \psi_i = \hat\psi_i - \psi(D_i) $, we have 
\$
&\frac{1}{|\cI_3|}\sum_{i\in \cI_3} \Delta \psi_i^2 
= \frac{1}{|\cI_3|}\sum_{i\in \cI_3} \big\{ \hat{M} s(D_i,\hat\theta) - M s(D_i,\theta_0^\new)\big\}^2  \\
&\leq  \frac{2}{|\cI_3|}\sum_{i\in \cI_3} \big\{ \hat{M} s(D_i,\hat\theta) - \hat{M} s(D_i,\theta_0^\new)\big\}^2   +  \frac{2}{|\cI_3|}\sum_{i\in \cI_3} \big\{ \hat{M} s(D_i,\theta_0^\new) - M s(D_i,\theta_0^\new)\big\}^2 .
\$
Since $\hat\theta$ is independent of $\cI_3$, 
we know 
\$
&\EE\bigg[ \frac{1}{|\cI_3|}\sum_{i\in \cI_3} \big\{ s(D_i,\hat\theta) -  s(D_i,\theta_0^\new)\big\}^2 \bigggiven \cI_1\cup\cI_2\bigg]\\
&= \big\|s( \cdot ,\hat\theta) -  s( \cdot ,\theta_0^\new) \big\|_{L_2(\mathbb{P})}^2 = O\big( \|\hat\theta-\theta_0^\new\|_2\big) = o_P(1).
\$
Hence Lemma~\ref{lem:op_to_op} yields 
\$
\frac{2}{|\cI_3|}\sum_{i\in \cI_3} \big\{ \hat{M} s(D_i,\hat\theta) - \hat{M} s(D_i,\theta_0^\new)\big\}^2
= 2\hat{M}^2 \cdot  \frac{1}{|\cI_3|}\sum_{i\in \cI_3} \big\{ s(D_i,\hat\theta) -  s(D_i,\theta_0^\new)\big\}^2 = o_P(1). 
\$
Also, since $\hat{M}-M = o_P(1)$, 
we have 
\$
2(M-\hat{M})^2\cdot \frac{1}{|\cI_3|}\sum_{i\in \cI_3} \big\{ s(D_i,\theta_0^\new) -  s(D_i,\theta_0^\new)\big\}^2 = o_P(1), 
\$
which further leads to $\frac{1}{|\cI_3|}\sum_{i\in \cI_3} w(Z_i)^2\Delta \psi_i^2 =o_P(1)$ since $ \|w(\cdot)\|_{\infty}<\infty$. 
On the other hand, by the rate conditions and 
the convergence result of $\hat\eta$ 
in Proposition~\ref{prop:consist_eta}, we know that 
$\|\hat\eta(\cdot)-\eta(\cdot)\|_{L_2(\mathbb{P})}=o_P(1)$. 
Since $\cI_3$ is independent of $\hat\eta$, 
writing $\Delta\eta_i = \hat\eta_i - \eta(Z_i)$, 
\$
&\EE\bigg\{ \frac{1}{|\cI_3|}\sum_{i\in \cI_3} w(Z_i)^2 \Delta\eta_i^2 \bigggiven \cI_1\cup\cI_2\bigg\}
\leq \|w(\cdot)\|_{\infty} \cdot  \big\| \hat\eta(\cdot)-\eta(\cdot)\big\|_{L_2(\mathbb{P})}^2 = o_P(1).
\$
Invoking Lemma~\ref{lem:op_to_op} yields $\frac{1}{|\cI_3|}\sum_{i\in \cI_3} w(Z_i)^2\Delta\eta_i^2 = o_P(1)$. 
Therefore, 
\$
&\frac{1}{| \cI_3|}\sum_{i\in  \cI_3} {w}(Z_i)^2(\hat\psi_i - \hat\eta_i)^2 - \frac{1}{| \cI_3|}\sum_{i\in  \cI_3} {w}(Z_i)^2\big\{ \psi(D_i) - \eta(Z_i)  \big\}^2  \\
&\leq  \frac{1}{| \cI_3|}\sum_{i\in  \cI_3} {w}(Z_i)^2\big(  \Delta\psi_i  - \Delta\eta_i\big)^2 \\
&\qquad + 2\sqrt{\frac{1}{| \cI_3|}\sum_{i\in  \cI_3} {w}(Z_i)^2\big\{ \psi(D_i) - \eta(Z_i)  \big\}^2 } \cdot \sqrt{\frac{1}{| \cI_3|}\sum_{i\in  \cI_3} {w}(Z_i)^2\big(  \Delta\psi_i  - \Delta\eta_i\big)^2} = o_P(1). 
\$
Similar arguments also yield 
\$
\frac{1}{| \cI_3|}\sum_{i\in  \cI_3}  (\hat\psi_i - \hat\eta_i)^2 
= \frac{1}{| \cI_3|}\sum_{i\in  \cI_3}  \big\{\psi(D_i) - \eta(Z_i)\big\}^2  + o_P(1). 
\$
Combining the above two results, we have 
\$
\hat\sigma_{\shift}^2 
&= \frac{1}{| \cI_3|}\sum_{i\in  \cI_3} {w}(Z_i)^2(\hat\psi_i - \hat\eta_i)^2 
+ \sup_z\big|\hat{w}(z) - w(z)\big|^2 \cdot \frac{1}{| \cI_3|}\sum_{i\in  \cI_3}  (\hat\psi_i - \hat\eta_i)^2 \\
&= \frac{1}{| \cI_3|}\sum_{i\in  \cI_3} {w}(Z_i)^2\big\{ \psi(D_i) - \eta(Z_i)  \big\}^2 +o_P(1) = \sigma_\shift^2 + o_P(1),
\$
which completes the proof. 
\end{proof}

\section{Auxiliary Results}
In this section, we provide auxiliary technical results for the proofs in preceding sections.
\subsection{Auxiliary results for conditional laws}
\label{app:subsec:cond_law}

\begin{lemma}\label{lem:cond_clt}
  Let $g(\cdot)$ be a function such that $\EE\{|g(X_i)|^4\}<\infty$, where $\{(X_i,Z_i)\}_{i=1}^n$ are i.i.d.\ data. Define the filtration $\mathcal{F}_n = \sigma(\{Z_i\}_{i=1}^n)$. Then for any $x\in \mathbb{R}$, it holds that 
  \#\label{eq:cond_clt}
  \PP\bigg( \frac{1}{\sqrt{n}} \sum_{i=1}^n \big[g(X_i) - \EE\{g(X_i)\given Z_i\} \big] \leq x\bigggiven \cF_n \bigg) 
  \#
  converges almost surely to $\Phi(x/\sigma)$,
  where $\Phi$ is the cumulative distribution function of standard normal distribution, and 
  \$
  \sigma^2 = \EE\Big( \big[g(X_i) - \EE\{g(X_i)\given Z_i\}\big]^2\Big).
  \$
  Moreover, for any filtration $\cG_n\subset \cF_n$, we also have 
  \#\label{eq:cond_clt_smaller}
  \PP\bigg( \frac{1}{\sqrt{n}} \sum_{i=1}^n \big[g(X_i) - \EE\{g(X_i)\given Z_i]\} \big] \leq x\bigggiven \cG_n \bigg) 
  \#
  converges almost surely to $ \Phi(x/\sigma).$
  \end{lemma}
  
  \begin{proof}[Proof of Lemma \ref{lem:cond_clt}]
  Let $\cL_n$ denote the conditional law of $\frac{1}{\sqrt{n}} \sum_{i=1}^n \zeta_i$ given $\cF_n$, where $\zeta_i:=g(X_i) - \EE[g(X_i)\given Z_i]$. 
  Since the data are $\iidtext$, $\{X_{i}\}_{i=1}^n$ are mutually independent conditional on $\cF_n=\sigma(\{Z_i\}_{i=1}^n)$. 
  Thus
  the characteristic function of $\cL_n$ is
  \$
  \varphi_{\cL_n}(t)= \EE\big( e^{\frac{it}{\sqrt{n}}\sum_{j=1}^n \zeta_j}   \biggiven \cF_n \big)=\prod_{j=1}^n \EE\big( e^{\frac{it}{\sqrt{n}} \zeta_j}   \biggiven \cF_n \big),\quad \text{for all }t\in \mathbb{R}.
  \$

By Lemma \ref{lem:conv_charac_fct}, we know that the conditional law $\cL_n$ converges almost surely to $N(0,\sigma^2)$, which completes the proof of equation \eqref{eq:cond_clt}. Since the conditional probabilitites are bounded within $[0,1]$, equation \eqref{eq:cond_clt_smaller} follows from dominated convergence theorem. Therefore we conclude the proof of Lemma \ref{lem:cond_clt}.
\end{proof}

\begin{lemma}
  Under the same assumption as Lemma \ref{lem:cond_clt}, we have 
  $
  \varphi_{\cL_n}(t)$
  converges almost surely to $\exp\big(-t^2 \sigma^2/2\big),
  $
  for all $t\in \mathbb{R}$, where $\sigma^2$ is defined in Lemma \ref{lem:cond_clt}.
  \label{lem:conv_charac_fct}
  \end{lemma}
  
  \begin{proof}[Proof of Lemma~\ref{lem:conv_charac_fct}]

  We now focus on $z_{n,j}=\EE\big( e^{\frac{it}{\sqrt{n}} \zeta_j}   \biggiven Z_j \big)-1$. By the tower property of conditional expectations, we have $\EE(\zeta_j \given Z_j)=0$ for all $j\in[n]$. Therefore
  $$
  z_{n,j} = -\frac{t^2}{2n}\EE(\zeta_j^2\given Z_j) + R_{n,j},~~\text{where }~R_{n,j} = \EE\Big( e^{\frac{it}{\sqrt{n}}\zeta_j}-1-\frac{it}{\sqrt{n}}\zeta_j+\frac{t^2}{2n}\zeta_j^2\Biggiven Z_i \Big).
  $$
  Since the random variables $\{\EE\big(\zeta_i^2\given Z_j\big)\big\}_{i=1}^n$ are $\iidtext$, by the law of large numbers, it holds that
  $$
  \sum_{m=1}^n \Big\{ -\frac{t^2}{2n}\EE\big(\zeta_j^2\given Z_j\big) \Big\} \stackrel{\text{a.s.}}{\to} -\frac{t^2}{2} \EE(\zeta_j^2) = -\frac{t^2}{2}\sigma^2,
  $$
  where $\sigma^2$ is defined in Lemma \ref{lem:cond_clt}. 
  Note  $|e^{ix}-1-i x +x^2/2|\leq \min\{|x|^2,|x|^3/6\}$ for any $x\in \mathbb{R}$, thus
  \$
  |R_{n,j}| &= \bigg|\EE\Big(e^{\frac{it}{\sqrt{n}}\zeta_j}-1-\frac{it}{\sqrt{n}}\zeta_j+\frac{t^2}{2n}\zeta_j^2\biggiven Z_j\Big)\bigg|\\
  &\leq \EE\bigg[ \min\Big\{ \frac{t^2}{2n}\zeta_j^2, ~\frac{t^3}{6n^{3/2}}|\zeta_j|^3  \Big\} \Biggiven Z_{j}  \bigg] \leq \frac{t^3}{6n^{3/2}} \EE\big( |\zeta_j|^3\biggiven Z_j\big).
  \$
  Under the finite fourth-moment condition, 
  by the law of large numbers we have
  \$
  \frac{1}{n}\sum_{j=1}^n \EE\big( |\zeta_j|^3\biggiven Z_j\big) ~\asto~ \EE\big( |\zeta_j|^3 \big) <\infty,
  \$ 
  hence
  $\sum_{i=1}^n|R_{n,j}|$ converges to zero almost surely, which leads to 
  $
  \sum_{j=1}^n z_{n,j} \to  -\frac{t^2}{2}\sigma^2$
  almost surely. 
  We now show $\sum_{j=1}^n |z_{n,j}|^2\stackrel{\text{a.s.}}{\to}0$. Simply note that $(x+y)^2\leq 2x^2+2y^2$, so 
  \#\label{eq:p2}
  \sum_{j=1}^n |z_{n,j}|^2 &\leq \frac{t^4}{2n^2}\sum_{i=1}^n \big\{\EE(\zeta_j^2\given Z_j)\big\}^2 + 2\sum_{j=1}^n R_{n,j}^2   \leq \frac{t^4}{2n^2}\sum_{j=1}^n  \EE(\zeta_j^4\given Z_j)  + 2\sum_{j=1}^n R_{n,j}^2 
  \#
  which converges to zero almost surely. 
  where the second inequality follows from Jensen's inequality. The a.s. convergence follows from the strong law of large numbers under the moment condition in Assumption \ref{assump:moment_main}, as well as the fact that 
  $
  \sum_{j=1}^n R_{n,j}^2 \leq \sum_{j=1}^n |R_{n,j}| \cdot \max_j|R_{n,j}|\leq \big(\sum_{j=1}^n |R_{n,j}|\big)^2, 
  $
  which converges to zero almost surely. 
  Combining equation \eqref{eq:p2} and Lemma \ref{lem:conv_char_fct}, we conclude the proof of Lemma \ref{lem:conv_charac_fct}. 
  \end{proof}

We quote the following well-known complex analysis result without proof.
  \begin{lemma}
  Suppose $z_{n,k}\in \mathbb{C}$ are such that $z_n=\sum_{k=1}^n z_{n,k}\to z_{\infty}$ and $\eta_n = \sum_{k=1}^n |z_{n,k}|^2\to 0$ as $n\to \infty$. Then 
  $ 
  \varphi_n   \prod_{k=1}^n (1+z_{n,k})\to \exp(z_\infty)~~\text{as }n\to\infty.
  $ 
  \label{lem:conv_char_fct}
  \end{lemma}

\subsection{Auxiliary technical lemmas}

\begin{lemma}
Suppose a sequence of random variables $E_n$ satisfies $E_n=o_P(1)$ as $n\to \infty$. Then for any $\sigma$-algebras $\cF_n$ and any constant $\epsilon>0$, it holds that
$
\PP\big( |E_n|>\epsilon \given \cF_n\big) =o_P(1).
$
\label{lem:op_to_op}
\end{lemma}

\begin{proof}[Proof of Lemma \ref{lem:op_to_op}]
Note that $\EE\big\{\PP\big( |E_n|>\epsilon \given \cF_n\big)\big\} = \PP\big( |E_n|>\epsilon \big)$. Thus for any $\delta>0$, we have 
\$
\PP\Big\{\PP\big( |E_n|>\epsilon \given \cF_n\big)>\delta\Big\} \leq \frac{1}{\delta}\PP\big( |E_n|>\epsilon \big)\to 0.
\$
Therefore we have $\PP\big( |E_n|>\epsilon \given \cF_n\big) =o_P(1)$ and completes the proof of Lemma~\ref{lem:op_to_op}.
\end{proof}

\begin{lemma}\label{lem:cond_to_op}
Let $\cF_n$ be a sequence of $\sigma$-algebra, and let $A_n\geq 0$ be a sequence of nonnegative random variables. If $\EE(A_n\given \cF_n) = o_P(1)$, then $A_n=o_P(1)$. 
\end{lemma}

\begin{proof}[Proof of Lemma~\ref{lem:cond_to_op}]
By Markov's inequality, for any $\epsilon>0$, we have 
\$
B_n := \PP( A_n >\epsilon \given \cF_n) \leq \frac{\EE(A_n\given \cF_n)}{\epsilon } = o_P(1),
\$
and $B_n\in[0,1]$ are bounded random variables. For any subsequence $\{n_k\}_{k\geq 1}$ of $\NN$, since $B_{n_k}\stackrel{P}{\to} 0$, there exists a subsequence $\{n_{k_i}\}_{i\geq 1} \subset \{n_k\}_{k\geq 1}$ such that $B_{n_{k_i}} \stackrel{\text{a.s.}}{\to} 0$ as $i\to \infty$. By the dominated convergence theorem, we have $\EE[B_{n_{k_i}}]\to 0$, or equivalently, 
$
\PP(A_{n_{k_i}} >\epsilon) \to 0.
$
Therefore, for any subsequence $\{n_k\}_{k\geq 1}$ of $\NN$, there exists a subsequence $\{n_{k_i}\}_{i\geq 1} \subset \{n_k\}_{k\geq 1}$ such that $A_{n_{k_i}} \stackrel{P}{\to} 0$ as $i\to \infty$. By the arbitrariness of $\{n_k\}_{k\geq 1}$, we know $A_n\stackrel{P}{\to} 0$ as $n\to \infty$, which completes the proof. 
\end{proof}

We cite without proof 
the following result on convex functions; 
see, e.g., Ex 2.5 in~\citet{mestimate}.
\begin{lemma}\label{lem:convex}
If $f\colon \Theta\to \RR$ is convex 
in $\Theta\subset \RR^p$ and 
$\nabla^2 f(\theta) \succeq \lambda\mathbf{I}_{p\times p}$ 
for all $\theta\in \Theta$ with $\|\theta-\theta^0\|\leq c$ 
for constants $\lambda,c$, 
then 
$f(\theta) \geq f(\theta_0) + \nabla f(\theta_0)^\top (\theta-\theta_0) + \lambda/2 \cdot \min\{ \|\theta-\theta_0\|^2 , c\|\theta_0-\theta\|\}$. 
\end{lemma}

\end{document}